\documentclass[11pt]{article} 

\usepackage{hyperref}  
\hypersetup{
    colorlinks=true,      
    linkcolor=blue,       
    citecolor=blue,        
    filecolor=magenta,    
    urlcolor=blue         
}

\usepackage{fullpage}
\usepackage{enumitem}
\usepackage[page]{appendix}
\usepackage{amssymb}
\usepackage{mathrsfs} 
\usepackage{natbib}
\usepackage{algorithm}
\usepackage{algorithmic}
\usepackage{hyperref}
\usepackage{comment}
\usepackage{multirow}
\usepackage{amsthm}
\usepackage{mathtools}
\usepackage{epstopdf,xcolor}
\usepackage{bbm}
\usepackage[colorinlistoftodos,prependcaption]{todonotes}
\usepackage{dsfont}
\usepackage{framed}
\usepackage{pifont}
\usepackage{subcaption}
\usepackage{tikz, subcaption}
\usetikzlibrary{shapes, arrows, positioning}
\usepackage{authblk}
\usepackage{listings}
\usepackage{tabularx}

\newtheorem{theorem}{Theorem}

\newtheorem{proposition}{Proposition}

\newtheorem{remark}{Remark}

\newtheorem{assumption}{Assumption}

\newcommand{\E}{\mathbb{E}}

\newcommand{\independent}{\perp\mkern-9.5mu\perp}

\allowdisplaybreaks

\title{\bf \LARGE Bias-Targeted Nonparametric Balancing for Stable Causal Mediation Analysis\\~}

\author[]{Chang Liu}
\author[]{AmirEmad Ghassami\thanks{Corresponding author (\texttt{ghassami@bu.edu)}}}

\affil[]{Department of Mathematics and Statistics, Boston University}

\date{First Version: March 31, 2024; Current Version: Feburary 9, 2026}

\begin{document}

\maketitle

\begin{abstract}
Influence function (IF)-based estimators are widely used in mediation analysis due to their modeling flexibility, but standard implementations require direct estimation of the distribution functions of the mediator and treatment variables. Since these functions appear in the denominator of IF-based estimators, they can induce significant instability, particularly with continuous mediators. In this work, we propose an alternative implementation of IF-based estimators for both single- and multiple-mediator settings, based on reparametrizations of the likelihood. The key idea is to estimate the involved nuisance functions according to their role in the bias structure of the IF-based estimators. In our approach, key nuisance functions that are potential sources of instability are estimated using a novel nonparametric weighted balancing method---which can be viewed as a nonparametric extension of covariate balancing generalized to mediation analysis---fully stabilizing the estimators. We establish consistency and multiple robustness under suitable regularity conditions, and asymptotic normality. Simulation studies demonstrate substantial reductions in bias and variance relative to existing methods for continuous mediators. We further illustrate the approach using NHANES 2013--2014 data to estimate the effect of obesity on coronary heart disease mediated by Glycohemoglobin.\\


\end{abstract}

\section{Introduction}\label{secLintro}

Understanding the mechanisms through which a treatment or exposure influences an outcome is of critical importance in many studies across the biological and social sciences. Causal mediation analysis provides a principled framework for studying these mechanisms by partitioning the total causal effect of a treatment on an outcome into a direct component and an indirect component transmitted through post-treatment mediator variables. Such decompositions help elucidate biological pathways and inform targeted interventions, making mediation analysis an indispensable tool in epidemiology, behavioral medicine, and health economics \citep{vanderweele2015explanation,mackinnon2012introduction,hayes2017introduction,richiardi2013mediation,ten2012review}.

It is well known that, under suitable causal assumptions, natural direct and indirect effects can be identified from observed data through a functional of the observed distribution commonly referred to as the mediation formula (or mediation functional) \citep{robins1992identifiability,pearl2001direct}. When these identification conditions are deemed plausible in a given application, a central methodological question is how to estimate the resulting mediation functional from finite samples in a stable and statistically efficient way.

The estimation is more challenging when the mediators are continuous:
in many modern applications, mediators are continuous and represent quantities such as gene expression levels, neural connectivity, metabolic activity, or inflammatory biomarkers \citep{shan2019identification,lindquist2012functional,wu2019sedentary,you2023association}. For example, in assessing the effect of obesity on cardiovascular risk, blood pressure and blood glucose are important continuous mediators \citep{pozuelo2017obesity,xu2017role}; in studying the influence of a cancer therapy on overall survival, protein expression levels may serve as mediators \citep{huang2023unified}. While continuous mediators capture rich biological variation, they also pose substantial challenges for semiparametric estimation of direct and indirect effects.

Several estimation strategies have been proposed based on regression and weighting techniques \citep{imai2010general,lange2012simple,valeri2013mediation,vanderweele2014effect,diaz2020causal,hejazi2023nonparametric,rudolph2023efficient,xia2022decomposition}. In particular, \cite{tchetgen2012semiparametric} proposed an influence-function (IF)-based estimator for the mediation functional that is locally semiparametrically efficient and multiply robust. Their estimator relies on three nuisance functions: the treatment mechanism (the conditional distribution of treatment given covariates), the mediator mechanism (the conditional distribution of the mediator given treatment and covariates), and the outcome regression (the conditional expectation of the outcome given treatment, mediator, and covariates). The resulting estimator remains consistent if any two of these three nuisance functions are consistently estimated.
However, in settings with continuous mediators, a straightforward implementation of IF-based estimators can perform poorly in practice for three related reasons. First, estimating the mediator mechanism entails estimating a conditional density for a (potentially multivariate) continuous mediator given multivariate covariates and treatment, which is statistically and computationally challenging. Second, the estimated treatment and mediator mechanisms enter the IF representation through inverse-probability-type factors; when these estimated quantities are small (e.g., due to practical positivity violations), the resulting estimator can exhibit severe variance inflation and numerical instability. Third, and more fundamentally, standard nuisance-learning objectives are typically not aligned with the structure of the second-order bias of the IF-based estimator: improving global predictive fit of nuisance components does not necessarily translate into reducing the bias and variance of the target mediation functional.

To address the difficulty of estimating the mediator mechanism, \cite{diaz2021nonparametric} proposed a likelihood reparametrization that avoids direct estimation of the mediator density.\footnote{Their framework additionally allows for mediator--outcome confounders affected by treatment.} Instead, their representation introduces an additional nuisance function involving the conditional distribution of treatment given covariates and the mediator. While this approach also achieves multiple robustness, it retains key practical challenges: the newly introduced nuisance component still appears in an inverse form and may therefore lead to large variance in finite samples. Moreover, when nuisance components are modeled separately (e.g., under parametric working models), incompatibilities across conditional models can arise, which may further degrade performance. These considerations suggest that stable estimation in continuous-mediator settings requires nuisance estimators that are constructed to directly control the dominant bias terms of the final IF-based estimator, rather than aiming only for globally accurate nuisance fits.

In this paper, we propose a bias-targeted nuisance estimation framework for stable causal mediation analysis. Our approach is based on a likelihood reparametrization that yields an alternative IF representation of the mediation functional in terms of nuisance components that can be estimated without directly fitting the treatment and mediator densities. In Section \ref{sec:single} we focus on the single-mediator case. We develop a two-stage estimation procedure. In the first stage, we estimate two basic nuisance functions, including an inverse treatment mechanism and an outcome regression component. In the second stage, we construct additional nuisance functions that encode the information needed from the mediator mechanism, and we estimate the nuisance components that drive instability using a novel nonparametric weighted balancing method. Crucially, the balancing equations are chosen to directly target the leading bias terms of the final IF-based estimator, thereby stabilizing the estimator in settings with continuous mediators.

Our main contributions regarding the single-mediator case can be summarized as follows:
\begin{itemize}
    \item We introduce a likelihood reparametrization and a two-stage nuisance estimation strategy for the mediation functional that avoids direct estimation of treatment and mediator densities.
    \item We propose a bias-targeted nonparametric balancing approach for estimating key nuisance components that otherwise induce instability through inverse-probability-type factors.
    \item We provide a reproducing kernel Hilbert space RKHS-based implementation of the proposed balancing estimators and obtain closed-form solutions for the corresponding minimax problems.
    \item We establish consistency, multiple robustness, and asymptotic normality of the resulting IF-based estimator under suitable regularity conditions, and we demonstrate substantial improvements in bias and variance in simulation studies with continuous mediators.
    \item We illustrate the method using NHANES 2013--2014 data to estimate the effect of obesity on coronary heart disease mediated by Glycohemoglobin; see Section~\ref{sec:application}.
\end{itemize}

Our work is related to reparametrization- and balancing-based approaches for causal mediation analysis; see, for example, \citep{chan2016efficient} and \citep{kawato2025balancing}. While both works derive estimation procedures that can be interpreted through balancing conditions, neither develops the kind of \emph{nonparametric balancing strategy} we propose for estimating nuisance components in a highly flexible way. In contrast, our approach constructs nuisance estimators to directly control the leading bias terms of the final influence function (IF)-based estimator through a regularized minimax balancing formulation. This design yields stable estimation for continuous mediators without requiring direct density estimation, and it admits closed-form solutions under reproducing kernel Hilbert space (RKHS) function classes. Regarding \citep{chan2016efficient} in particular, their calibration conditions are primarily motivated by a weighting representation and efficiency considerations, rather than by explicitly targeting the bias structure of an IF-based estimator. Correspondingly, their development does not focus on bias-cancellation arguments that underpin multiple-robustness guarantees, and they do not establish multiple-robustness results of the type we obtain in this paper.

We generalize our framework in Section \ref{sec:multi} to multiple-mediator settings without imposing causal ordering or structural assumptions among mediators. Such settings are scientifically important when mediators are measured contemporaneously and may interact through complex, potentially bidirectional pathways. For example, lifestyle interventions may affect blood pressure, cholesterol, inflammation, and body mass index, which jointly influence cardiovascular outcomes \citep{lima2020healthy}; maternal smoking may affect multiple intermediate pregnancy conditions that jointly influence infant health outcomes \citep{engel2009maternal}. In these settings, imposing a strict mediator ordering can be unrealistic. We adapt the approach of \cite{xia2022decomposition} and decompose the indirect effect to multiple components and focus on the mediator-specific exit effects. We develop an RKHS-based extension of our bias-targeted balancing framework for estimating mediator-specific exit indirect effects under the corresponding identification conditions. The existing IF-based approach for multiple mediators in \citep{xia2022decomposition} inherits the same three concerns discussed above for the single-mediator case and, in particular, requires estimating several treatment and mediator distribution components that enter the estimator through inverse-probability-type factors; consequently, it can remain unstable under practical positivity violations and may be sensitive to model misspecification. Our proposal mitigates these issues by replacing direct modeling of these distribution components with balancing conditions designed to control the dominant bias terms of the target estimator.

In the next section, we start by describing the single- and multiple-mediator models and the corresponding identification results.

\section{Model Description}
\label{sec:desc}

\subsection{Single Mediator}
\begin{figure}[t!]
\centering
		\tikzstyle{block} = [draw, circle, inner sep=2.5pt, fill=lightgray]
		\tikzstyle{input} = [coordinate]
		\tikzstyle{output} = [coordinate]
        \begin{tikzpicture}
            \tikzset{edge/.style = {->,> = latex'}}
            \node[] (a) at  (-2,0) {$A$};
            \node[] (x) at  (-3,2) {$X$};
            \node[] (m) at  (0,1) {$M$};
            \node[] (y) at  (2,0) {$Y$};                       
            \draw[-stealth] (a) to (m);                                                         
            \draw[-stealth] (a) to (y);
            \draw[-stealth] (x) to (a);                                                         
            \draw[-stealth] (x) to (m);
            \draw[-stealth,bend left=45] (x) to (y);
            \draw[-stealth] (m) to (y);
        \end{tikzpicture}
        \caption{A graphical causal model for the mediation model.}
        \label{fig:MemMod}
\end{figure}
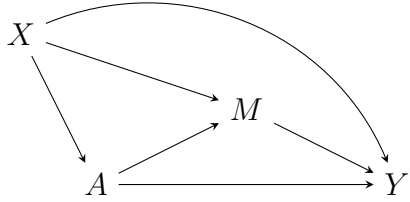

In this section, we describe the single-mediator model and review conditions sufficient for identification of the direct and indirect causal effects.
We consider a model with binary treatment variable $A\in\{0,1\}$, pre-treatment covariates $X\in\mathcal{X}$, post-treatment mediator variables $M\in\mathcal{M}$, and outcome variable $Y\in\mathcal{Y}$. In this work, we use the potential outcomes framework \citep{rubin1974estimating}.
We denote the potential outcome variable of $Y$, had the treatment and mediator variables been set to value $A=a$ and $M=m$ (possibly contrary to the fact) by $Y^{(a,m)}$. 
Similarly, we define $M^{(a)}$ as the potential outcome variable of $M$ had the treatment variables been set to value $A=a$.
Based on variables $Y^{(a,m)}$ and $M^{(a)}$, we define $Y^{(a)}=Y^{(a,M^{(a)})}$.

Consider the parameter average treatment effect (ATE) of the treatment $A$ on the outcome $Y$, represented by $\theta_{ATE}=E[Y^{(1)}-Y^{(0)}]$. The main goal in causal mediation analysis is to quantify the part of the ATE that is mediated by the variable $M$, and the part that is directly reaching the outcome variable. To this end, the ATE can be partitioned as follows \citep{robins1992identifiability,pearl2001direct},
\begingroup
\allowdisplaybreaks	
\begin{align*}
\theta_{ATE}
&=E[Y^{(1)}-Y^{(0)}]\\
&=E[Y^{(1,M^{(1)})}-Y^{(0,M^{(0)})}]\\
&=\underbrace{E[Y^{(1,M^{(1)})}-Y^{(1,M^{(0)})}]}_{\text{NIE}}
+\underbrace{E[Y^{(1,M^{(0)})}-Y^{(0,M^{(0)})}]}_{\text{NDE}}.
\end{align*}
\endgroup
The first and the second terms in the last expression are called the natural indirect effect (NIE) and the natural direct effect (NDE) of the treatment on the outcome, respectively \citep{pearl2001direct}.\footnote{These two terms are called the total indirect effect and the pure direct effect, respectively by \cite{robins1992identifiability}.}
NIE captures the change in the expectation of the outcome in a hypothetical scenario where the value of the treatment variable is fixed at $a=1$, while the mediator behaves as if the treatment had been changed from value $a=0$ to $a=1$. 
NDE captures the change in the expectation of the outcome in a hypothetical scenario where the value of the treatment variable is changed from value $a=0$ to $a=1$, while the mediator behaves as if the treatment is fixed at $a=0$. 

Following majority of the work in the field of causal inference, we require consistency, meaning that $M^{(a)}=M$ if $A=a$, and $Y^{(a,m)}=Y$ if $A=a$ and $M=m$, as well as positivity, meaning that for all $m$, $a$, and $x$, we have $p(m\mid a,x)>0$, and  $p(a\mid x)>0$. We also require that the model does not contain any unobserved confounders, which is a common and standard assumption in the causal mediation analysis literature. This can be formalized using the sequential exchangeability assumption \citep{imai2010general}. 
\begin{assumption}[Sequential exchangeability]
\label{assumption:med}
For any $a,a'\in\{0,1\}$, and any value of the mediator $m$, we have
\begin{enumerate}[label=(\roman*)]
	\item $M^{(a)}\independent A\mid X$,
    \item $Y^{(a,m)}\independent A\mid X$, and $Y^{(a,m)}\independent M^{(a)}\mid \{A=a,X\}$,    
	\item $Y^{(a,m)}\independent M^{(a')}\mid X$.
\end{enumerate}
\end{assumption}
Intuitively, Part (i) requires that the treatment-mediator relation is not confounded by unobserved variables, Part (ii) of the assumption requires that the treatment-outcome and mediator-outcome relations are not confounded by unobserved variables, (see Figure \ref{fig:MemMod}) and Part (iii), which is a so-called cross-world independence assumption, requires that the potential outcomes of $Y$ and $M$ even for different assignments of the treatment are conditionally independent. All three parts are untestable and should be decided based on expert's knowledge and the researcher's belief about how the variables in the system are generated.

We notice that under consistency, positivity, and sequential exchangeability, we have
\begin{align*}
E[Y^{(1,M^{(1)})}]=E[Y^{(1)}]=E\{E[Y^{(1)}\mid X]\}=E\{E[Y^{(1)}\mid A=1,X]\}=E\{E[Y\mid A=1,X]\}.
\end{align*}
That is $E[Y^{(1,M^{(1)})}]$ is identified. Similar argument holds for $E[Y^{(0,M^{(0)})}]$.
The parameter $\theta_0=E[Y^{(1,M^{(0)})}]$ requires more caution.
It is known that under the assumptions of our model, this parameter can be identified as follows \citep{pearl2001direct,imai2010general}.
\begingroup
\allowdisplaybreaks	
\begin{align*}
\theta_0
&=E[Y^{(1,M^{(0)})}]\\
&=\int yp(Y^{(1,m)}=y\mid x)p(M^{(0)}=m\mid x)p(x)dydmdx\\
&=\int yp(y\mid a=1,m,x)p(m\mid a=0,x)p(x)dydmdx\\
&=:\psi_0.
\end{align*}
\endgroup
The functional $\psi_0$ is called the mediation functional (also known as mediation formula). This quantity is purely a functional of the observed data distribution, $P_0$. Note that under the identification assumptions in our model, $\psi_0=\theta_0$, but in general, these two parameters are not necessarily equal. In the single-mediator model, we refer to $\psi_0$ as the parameter of interest. Our goal is to estimate this parameter of interest from the independent and identically distributed observations $\{O_i:=(X,A,M,Y)_i\}_{i=1}^n$ drawn from $P_0$.

\subsection{Multiple Mediators}
In this section, we introduce the multiple-mediator model. The binary treatment variable $A$, pre-treatment covariates $X$, and outcome variable $Y$ are defined as in the single-mediator setting. Following the terminology of \cite{xia2022decomposition}, we denote the post-treatment mediator vector by $M=(M_1,M_2,\dots,M_k)$, where $M_1$, $M_2$, $\dots$, $M_k$ represent $k$ mediators of interest, each of which may be multivariate. No causal ordering or structural assumptions are imposed among the mediators. Analogously, we denote the potential outcome variable of $Y$, had the treatment and mediators been set to value $A=a$, $M_1=m_1$, $M_2=m_2$, $\dots$, $M_k=m_k$ by $Y^{(a,m_1,m_2,\dots,m_k)}$; $M_j^{(a)}$ ($j=1,2,\dots,k$) is the potential outcome variable of $M_j$ had the treatment variables been set to value $A=a$. Based on variables $Y^{(a,m_1,m_2,\dots,m_k)}$ and $M_j^{(a)}$, we define $Y^{(a)}=Y^{(a,M_1^{(a)},M_2^{(a)},\dots,M_k^{(a)})}$.

When multiple mediators are present in the system, the NIE contains a joint effect of the multiple mediators. Therefore, further decomposition of the total NIE is often desired. We consider the decomposition,
\begin{align*}
    \text{NIE}&=E[Y^{(1,M_1^{(1)},M_2^{(1)},\dots,M_k^{(1)})}-Y^{(1,M_1^{(0)},M_2^{(0)},\dots,M_k^{(0)})}]\\
    &=\sum_{j=1}^k\underbrace{E[Y^{(1,M_1^{(1)},M_2^{(1)},\dots,M_k^{(1)})}-Y^{(1,M_{j}^{(0)},M_{-j}^{(1)})}]}_{\text{EIE}_{M_j}}\\
    &~~~-\underbrace{(k-1)E[Y^{(1,M_1^{(1)},M_2^{(1)},\dots,M_k^{(1)})}]+E[Y^{(1,M_1^{(0)},M_2^{(0)},\dots,M_k^{(0)})}]-\sum_{j=1}^k E[Y^{(1,M_{j}^{(0)},M_{-j}^{(1)})}]}_{\text{INT}},
\end{align*}
where $M_{-j}^{(a)}$ denotes the vector of potential outcome variables for all mediators other than the $j$th one, under treatment level $a$. The first $k$ terms in the last expression are called the exit indirect effect (EIE) through $M_j$, $j=1,2,\dots,k$, respectively, and the last term is an interaction effect between $M_1$, $M_2$, $\dots$, $M_k$.

We still require consistency, meaning that, $M_j^{(a)}=M_j$ if $A=a$ for $j=1,2,\dots,k$, and $Y^{(a,m_1,m_2,\dots,m_k)}=Y$ if $A=a$, $M_1=m_1$, $M_2=m_2$, $\dots$, $M_k=m_k$, as well as positivity, meaning that for all $m_j,a,$ and $x$, we have $p(m_1,m_2,\dots,m_k\mid a,x)>0$, $p(m_j\mid a,x)>0$, for $j=1,2,\dots,k$, and $p(a\mid x)>0$.

The no-unobserved-confounder assumption is still required, and the corresponding sequential exchangeability assumptions are provided as follows.

\begin{assumption}[Sequential exchangeability (multiple)]
\label{assumption:med-multi}
    For any $a,a'\in\{0,1\}$, and any values $m_1$ and $m_2$ of the mediators, we have
\begin{enumerate}[label=(\roman*)]
	\item $\{Y^{(a,m_1,m_2,\dots,m_k)},M_1^{(a)},M_2^{(a)},\dots,M_k^{(a)}\}\independent A\mid X$,
	\item $Y^{(a,m_1,m_2,\dots,m_k)}\independent \{M_1^{(a)},M_2^{(a)},\dots,M_k^{(a)}\}\mid A=a,X$,
	\item $Y^{(a,m_1,m_2,\dots,m_k)}\independent \{M_1^{(a')},M_2^{(a')},\dots,M_k^{(a')}\}\mid X$.
\end{enumerate}
\end{assumption}

Interpretation of Assumption \ref{assumption:med-multi} is similar to that for Assumption \ref{assumption:med}.

Under Assumption \ref{assumption:med-multi}, $E[Y^{(1,M_1^{(1)},M_2^{(1)},\dots,M_k^{(1)})}]$ and $E[Y^{(1,M_1^{(0)},M_2^{(0)},\dots,M_k^{(0)})}]$ can be identified in the same manner as in the single-mediator setting, since the vector $(M_1^{(a)},M_2^{(a)},\dots,M_k^{(a)})$ may be treated as a single composite mediator $M^{(a)}$. Therefore, NIE in the multiple-mediator case is also identified, which we denote as $\xi_0^{\text{multi}}$. 

\begin{assumption}
\label{assumption:no-heterogeneity} 
For any values $m_1$, $m_2$, $\dots$, $m_k$ of the mediators, we have
\begin{align*}
    E[Y^{(1,M_1^{(1)},M_2^{(1)},\dots,M_k^{(1)})}-Y^{(1,M_{j}^{(0)},M_{-j}^{(1)})}\mid M_{-j}^{(1)}=m_{-j},X]=E[Y^{(1,M_j^{(1)},m_{-j})}-Y^{(1,M_j^{(0)},m_{-j})}\mid X],
\end{align*}
where $m_{-j}$ denotes the vector of mediator values corresponding to all mediators except the $j$th one.
\end{assumption}
The conditions in Assumption \ref{assumption:no-heterogeneity} state that one mediator does not introduce additional effect heterogeneity of the treatment via the other mediators.
Note that the identification of the cross-world quantities $E[Y^{(1,M_{j}^{(0)},M_{-j}^{(1)})}]$ for $j=1,2,\dots,k$ require stronger sequential exchangeability assumptions \citep{taguri2018causal}. However, since our primary interest lies in the exit indirect effects $\text{EIE}_{M_j}$, we instead impose weaker identification conditions that suffice for identifying these effects jointly, without requiring identification of each cross-world mean separately.

Under Assumptions \ref{assumption:med-multi} and \ref{assumption:no-heterogeneity}, the interventional indirect effects $\text{EIE}_{M_j}$ $(j=1,2,\dots,k)$ are identifiable \citep{xia2022decomposition}.\footnote{\cite{xia2022decomposition} presented identification for the case of two mediators.} Specifically,
\begin{align*}
    &\text{EIE}_{M_j}\\
    &=E[Y^{(1,M_1^{(1)},M_2^{(1)},\dots,M_k^{(1)})}-Y^{(1,M_{j}^{(0)},M_{-j}^{(1)})}]\\
    &=\iiiint yp(Y^{(1,m_{j},m_{-j})}=y\mid x)p(x)\\
    &~~~~~~~~~~~\times\big(p(M_j^{(1)}=m_j,M_{-j}^{(1)}=m_{-j},\mid x)-p(M_j^{(0)}=m_j,M_{-j}^{(1)}=m_{-j}\mid x)\big)dydm_jdm_{-j}dx\\
    &=\iiiint yp(y\mid a=1,m_{j},m_{-j},x)\\
    &~~~~~~~~~~~\times\big(p(m_j\mid a=1,x)-p(m_j\mid a=0,x)\big)p(m_{-j}\mid a=1,x)p(x)dydm_jdm_{-j}dx\\
    &=:\Delta^{M_j}.
\end{align*}
The interaction effect is then identified as
\begin{align*}
    INT=\xi_0^{\text{multi}}-\sum_{j=1}^k\Delta^{M_j}:=\Delta^{\text{INT}}.
\end{align*}
Both $\Delta^{M_j}$ and $\Delta^{\text{INT}}$ are functionals of the observed data distribution $P_0$. Since $\Delta^{\text{INT}}$ is fully determined by $\xi_0^{\text{multi}}$ and $\Delta^{M_j}$, we focus on estimation of $\Delta^{M_j}$ in what follows. Under the multiple-mediator model, our objective is to estimate $\Delta^{M_j}$ from the independent and identically distributed observations $\{O_i:=(X,A,M_1,M_2,\dots,M_k,Y)_i\}_{i=1}^n$ drawn from $P_0$.

\section{Single-Mediator Case}\label{sec:single}

\subsection{Influence Function-Based Estimator}
\label{sec:IF}
In this section, we review the influence function-based estimator of $\psi_0$ proposed by \cite{tchetgen2012semiparametric} and its multiple-robustness property. The explicit expression for the influence function of $\psi_0$ is given in the Appendix, and we denote it by $IF_{\psi_0}$. Based on $IF_{\psi_0}$, the authors propose first estimating the following nuisance functions: $ p(A=1\mid X=\cdot)$, $E[Y\mid A=1,M=\cdot,X=\cdot]$, and $ p(M=\cdot\mid A=a,X=\cdot)$, for $a\in\{0,1\}$. Then, estimate the mediation functional by
\begin{align*}
\hat\psi^{\text{TTS}}=&E_n\bigg[\frac{I(A=1)\hat{p}(M\mid A=0,X)}{\hat{p}(A=1\mid X)\hat{p}(M\mid A=1,X)}\{Y-\hat{E}[Y\mid A=1,M,X]\}\\
&+\frac{I(A=0)}{1-\hat{p}(A=1\mid X)}\{\hat{E}[Y\mid A=1,M,X]-\hat \eta(1,0,X)\}+\hat \eta(1,0,X)\bigg],
\end{align*}
where $E_n$ is the empirical expectation operator (treating nuisance functions as non-random), and 
\[
\hat \eta(1,0,X)=\int\hat{E}[Y\mid A=1,M=m,X]\hat{p}(M=m\mid A=0,X)dm.
\]
The authors showed that this estimator is multiply robust in the following sense,
\begin{proposition}
\label{prop:MRTTS}
$\hat\psi^{\text{TTS}}$ is unbiased if at least one of the following pairs of nuisance function estimators is correctly specified.
\begin{enumerate}[label=(\roman*)]
\item $\{\hat{p}(A=1\mid X=\cdot),\hat E[Y\mid A=1,M=\cdot,X=\cdot]\}$,
\item $\{\hat{p}(A=1\mid X=\cdot),\hat{p}(M=\cdot\mid A=a,X=\cdot)\}$,
\item $\{\hat E[Y\mid A=1,M=\cdot,X=\cdot],\hat{p}(M=\cdot\mid A=a,X=\cdot)\}$.
\end{enumerate}
\end{proposition}

As mentioned in Section \ref{secLintro}, the main concerns with $\hat\psi^{\text{TTS}}$ are the complexity of the nuisance
functions and variance inflation from treatment and mediator mechanisms in the denominator. Below, we propose an alternative estimation strategy that mitigates these issues.

\subsection{Proposed Estimation Approach}
\label{sec:proposal}
\subsubsection{Two-Stage Nuisance Function Estimation}

In this section, we present our two-stage nuisance function estimation approach. We start by introducing a reparametrization of the likelihood function based on the following nuisance functions.
\begin{align*}
&\pi_1(X=\cdot):=\frac{1}{p(A=0\mid X=\cdot)},\\
&\pi_2(M=\cdot,X=\cdot):=\frac{p(M=\cdot\mid A=0,X=\cdot)}{p(M=\cdot,A=1\mid X=\cdot)},\\
&\mu_1(M=\cdot,X=\cdot):=E[Y\mid A=1,M=\cdot,X=\cdot],\\
&\mu_2(X=\cdot):=E[\mu_1(M,X)\mid A=0,X=\cdot].
\end{align*}

Note that the parameter of interest can be written as
\[
\psi_0=E[\mu_2(X)].
\]

We start by presenting the steps of our proposed estimation method, which contain cross-fitting \citep{chernozhukov2018double} for separating the samples used for estimating the nuisance functions from that used for estimating the parameter of interest.

\medskip
\noindent
\begin{sloppypar}{\bf Estimating the nuisance functions.} 
We partition the samples into $L$ folds $\{I_1,\dots,I_L\}$ with roughly equal size $m=n/L$. For $\ell\in\{1,\dots,L\}$, we estimate the nuisance functions as $(\hat{\pi}_{1,\ell},\hat{\pi}_{2,\ell},\hat{\mu}_{1,\ell},\hat{\mu}_{2,\ell})$ on data from all folds but $I_{\ell}$, which we denote as $I_\ell^c$. In order to estimate the nuisance functions, we proceed in the following two stages.
\end{sloppypar}

\medskip
\noindent
{\bf Stage 1:} 
\begin{enumerate}[label=(\roman*)]
\item 
We estimate $\pi_1$ using the following balancing estimator:\\ We choose $\hat\pi_{1,\ell}$ such that it satisfies
\begin{equation}
\label{eq:balance1}
E\left[\left\{(1-A)\hat\pi_{1,\ell}(X)-1\right\}h_1(X)\right]\approx0,~~~~~~\forall h_1\in\mathcal{H}_1,
\end{equation}
where $\mathcal{H}_1\subset L^2(P_{0,X})$ is a function space chosen by the researcher.
\item We estimate $\mu_1$ by regressing $Y$ on $\{M,X,A=1\}$ to obtain $\hat\mu_{1,\ell}$.
\end{enumerate}

\noindent
{\bf Stage 2:} 
\begin{enumerate}[label=(\roman*)]
\item 
We estimate $\pi_2$ using the following balancing estimator:\\ Given $\hat\pi_{1,\ell}$ from Stage 1, we choose $\hat\pi_{2,\ell}$ such that it satisfies
\begin{equation}
\label{eq:balance2}
E\left[\left\{A\hat\pi_{2,\ell}(M,X)-(1-A)\hat\pi_{1,\ell}(X)\right\}h_2(M,X)\right]\approx0,~~~~~~\forall h_2\in\mathcal{H}_2,
\end{equation}
where $\mathcal{H}_2\subset L^2(P_{0,(M,X)})$ is a function space chosen by the researcher.
\item Given $\hat\mu_{1,\ell}$ from Stage 1, we estimate $\mu_2$ by regressing $\hat\mu_{1,\ell}(X,M)$ on $\{X,A=0\}$ to obtain $\hat\mu_{2,\ell}$.
\end{enumerate}

\medskip
\noindent
{\bf Estimating the parameter of interest.}
For any choice of nuisance functions $\{\tilde\pi_1,\tilde\pi_2,\tilde\mu_1,\tilde\mu_2\}$, let
\begin{align*}
\phi(O;\tilde\pi_1,\tilde\pi_2,\tilde\mu_1,\tilde\mu_2):=&\tilde\mu_2(X)+(1-A)\tilde\pi_1(X)\{\tilde\mu_1(M,X)-\tilde\mu_2(X)\}\\
&+A\tilde\pi_2(M,X)\{Y-\tilde\mu_1(M,X)\}.
\end{align*}

For all $\ell$, let $\hat\psi^{\text{2S}}_\ell$ be the estimator of $\psi_0$ defined as,
\begin{align*}
\hat\psi^{\text{2S}}_\ell
=E_m\left[ \phi(O;\hat\pi_{1,\ell},\hat\pi_{2,\ell},\hat\mu_{1,\ell},\hat\mu_{2,\ell})\right],
\end{align*}
where the empirical expectation is evaluated on $I_\ell$. Our final estimator of the parameter of interest will be the following.
\begin{align*}
    \hat\psi^{\text{2S}}=\frac{1}{L}\sum_{\ell=1}^L \hat\psi^{\text{2S}}_\ell.
\end{align*}

The use of the cross-fitting approach plays a key role in establishing consistency and asymptotic normality of our proposed estimator, shown in Theorem \ref{thm:CAN}.

We propose a nonparametric estimator for Stage 1 based on minimax optimization. To make the notation simple, throughout the rest of Section \ref{sec:proposal}, assume we use a sample of size $n$ for estimating the nuisance functions. A direct implementation is to solve the following optimization problem.

\[
\tilde\pi_1=\arg\min_{\pi_1\in\Pi_1}\max_{h_1\in\mathcal{H}_1}E_n\left[\left\{(1-A)\pi_1(X)-1\right\}h_{1}(X)\right],
\]
where $\Pi_1\subset L^2(P_{0,X})$ is a function space chosen by the researcher. Unfortunately, this estimator will be unstable and highly sensitive to noise. The following result suggests an alternative.
\begin{proposition}
\label{prop:stabel_pi1}
$\breve{\pi}_1$ obtained from the following minimax optimization satisfies Equation \eqref{eq:balance1} with equality.
\begin{equation*}
\breve{\pi}_1=\arg\min_{\pi_1\in L^2(P_{0,X})}\max_{h_1\in L^2(P_{0,X})}E\left[\left\{(1-A)\pi_1(X)-1\right\}h_{1}(X)-\frac{1}{4}h^2_{1}(X)\right].
\end{equation*}
\end{proposition}

Inspired by Proposition \ref{prop:stabel_pi1}, we estimate $\pi_1$ as,
\begin{equation}
\label{eq:estpi1}
\hat\pi_1=\arg\min_{\pi_1\in\Pi_1}\max_{h_1\in\mathcal{H}_1}E_n\left[\left\{(1-A)\pi_1(X)-1\right\}h_{1}(X)-\frac{1}{4}h^2_{1}(X)\right]+R_{n}^1(\pi_1,h_1),
\end{equation}
where $R_n^1(\cdot,\cdot)$ is a regularizer. We will give a specific choice for $R_n^1(\cdot,\cdot)$ in Section \ref{sec_drest}. 
\begin{remark}
\cite{imai2014covariate} proposed the covariate balancing propensity score (CBPS) method for estimating the average treatment effect. Our proposed technique for estimating $\hat\pi_1$ in Equation \eqref{eq:estpi1} can be viewed as a nonparametric version of their covariate balancing approach. Therefore, our proposed nonparametric estimator $\hat\pi_1$ can also be used in the setting of \cite{imai2014covariate} for estimating average treatment effect without resorting to parametric models for the inverse of the propensity score.
\end{remark}

We have similar arguments regarding $\pi_2$ in Stage 2, which lead to the following result.
\begin{proposition}
\label{prop:stabel_pi2}
For any give $\hat\pi_1$,
$\breve{\pi}_2$ obtained from the following minimax optimization satisfies Equation \eqref{eq:balance2} with equality.
\begin{equation*}
\breve{\pi}_2=\arg\min_{\pi_2\in L^2(P_{0,(M,X)})}\max_{h_2\in L^2(P_{0,(M,X)})}E\Big[\{A\pi_2(M,X)-(1-A)\hat\pi_1(X)\}h_{2}(M,X)-\frac{1}{4}h^2_{2}(M,X)\Big].
\end{equation*}
\end{proposition}

Therefore, in Stage 2, to estimate $\pi_2$, we can use the following estimator.
\begin{equation}
\label{eq:estpi2}
\hat\pi_2=\arg\min_{\pi_2\in\Pi_2}\max_{h_2\in\mathcal{H}_2}E_n\Big[\{A\pi_2(M,X)-(1-A)\hat\pi_1(X)\}h_2(M,X)-\frac{1}{4}h^2_2(M,X)\Big]+R_{n}^2(\pi_2,h_2),
\end{equation}
where $\Pi_2\subset L^2(P_{0,(M,X)})$ is a user-specified function class, and $R_n^2(\cdot,\cdot)$ is a regularizer.

In order to estimate $\mu_1$ in Stage 1 and $\mu_2$ in Stage 2 of the approach, one may use any standard parametric or nonparametric regression method.

\subsubsection{A Kernel Method-Based Implementation}
\label{sec_drest}

In this section, we describe solving the minimax problems \eqref{eq:estpi1} and \eqref{eq:estpi2} when RKHSes are used as the hypothesis classes. We use the method of \cite{pmlr-v151-ghassami22a}, which extends the nonparametric estimation approach proposed by \cite{dikkala2020minimax} for solving integral equations via adversarial learning within the semiparametric proximal causal inference framework.

Let $\Pi_1$, $\mathcal{H}_1$, $\Pi_2$, and $\mathcal{H}_2$ be RKHSes with kernels $K_{\Pi_1}$, $K_{\mathcal{H}_1}$, $K_{\Pi_2}$, and $K_{\mathcal{H}_2}$, respectively, equipped with the RKHS norms $\|\cdot\|_{\Pi_1}$, $\|\cdot\|_{\mathcal{H}_1}$, $\|\cdot\|_{\Pi_2}$, and $\|\cdot\|_{\mathcal{H}_2}$. We propose the following Tikhonov regularization-based optimization problems:
\begin{align}
\hat\pi_1&=\arg\min_{\pi_1\in\Pi_1}\max_{h_1\in\mathcal{H}_1}E_n\left[\left\{(1-A)\pi_1(X)-1\right\}h_{1}(X)-\frac{1}{4}h^2_{1}(X)\right]-\lambda_{\mathcal{H}_1}\|h_1\|_{\mathcal{H}_1}^2+\lambda_{\Pi_1}\|\pi_1\|_{\Pi_1}^2, \label{eq:estpi1reg}\\
\hat\pi_2&=\arg\min_{\pi_2\in\Pi_2}\max_{h_2\in\mathcal{H}_2}E_n\left[\left\{A\hat\pi_2(M,X)-(1-A)\hat\pi_1(X)\right\}h_2(M,X)-\frac{1}{4}h^2_2(M,X)\right]\notag\\
&~~~~~~~~~~~~~~~~~~~~~~~~~~~~~~-\lambda_{\mathcal{H}_2}\|h_2\|_{\mathcal{H}_2}^2+\lambda_{\Pi_2}\|\pi_2\|_{\Pi_2}^2, \label{eq:estpi2reg}
\end{align}
where $\lambda_{\mathcal{H}_1}$, $\lambda_{\Pi_1}$, $\lambda_{\mathcal{H}_2}$, and $\lambda_{\Pi_2}$ are regularization hyperparameters that are decreasing in $n$. The convergence analysis of the proposed minimax estimators can be done similar to the analysis presented in \citep{dikkala2020minimax,pmlr-v151-ghassami22a}, and hence omitted here.

\begin{sloppypar} Define $K_{\Pi_1,n}=(K_{\Pi_1}(X_i,X_j))_{i,j=1}^n$, $K_{\Pi_2,n}=(K_{\Pi_2}((M,X)_i,(M,X)_j))_{i,j=1}^n$, $K_{\mathcal{H}_1,n}=(K_{\mathcal{H}_1}(X_i,X_j))_{i,j=1}^n$, and $K_{\mathcal{H}_2,n}=(K_{\mathcal{H}_2}((M,X)_i,(M,X)_j))_{i,j=1}^n$ as the empirical kernel matrices corresponding to spaces $\Pi_1$, $\Pi_2$, $\mathcal{H}_1$,  and $\mathcal{H}_2$, respectively. The following result provides closed-form solutions for $\hat\pi_1$ and $\hat\pi_2$.
\end{sloppypar}

\begin{proposition}\label{prop:pisol}
~\\\vspace{-5mm}
\begin{enumerate}[label=(\alph*)]
	\item Equation \eqref{eq:estpi1reg} achieves its optimum at $\hat{\pi}_1=\sum_{i=1}^n\alpha_i K_{\Pi_1}(X_i,\cdot)$, with $\alpha$ defined as,
    \begin{align*}
        \alpha = \Big(K_{\Pi_1,n} diag(1-A)\Gamma_1 diag(1-A)K_{\Pi_1,n}+n^2\lambda_{\Pi_1} K_{\Pi_1,n}\Big)^{\dagger}K_{\Pi_1,n}diag(1-A)\Gamma_1 e_n,
    \end{align*}
    where $\Gamma_1=\frac{1}{4}K_{\mathcal{H}_1,n}(\frac{1}{4n}K_{\mathcal{H}_1,n}+\lambda_{\mathcal{H}_1}I_n)^{-1}$, $diag(1-A)$ is a diagonal matrix with $1-A_i$ as the $i$-th diagonal entry, $e_n :=\begin{pmatrix}
        1 & 1 & \cdots & 1
    \end{pmatrix}^\top$, and $\dagger$ denotes Moore-Penrose pseudo-inverse.
    \bigskip

    \item  Equation \eqref{eq:estpi2reg} achieves its optimum at $\hat{\pi}_2=\sum_{i=1}^n\beta_i K_{\Pi_2}((M,X)_i,\cdot)$, with $\beta$ defined as,
    \begin{align*}
        \beta = \Big(K_{\Pi_2,n} diag(A)\Gamma_2 diag(A)K_{\Pi_2,n}+n^2\lambda_{\Pi_2} K_{\Pi_2,n}\Big)^{\dagger}K_{\Pi_2,n}diag(A)\Gamma_2 \big((1-A)\hat{\pi}_1(X)\big)_n,
    \end{align*}
    where $\Gamma_2=\frac{1}{4}K_{\mathcal{H}_2,n}(\frac{1}{4n}K_{\mathcal{H}_2,n}+\lambda_{\mathcal{H}_2}I_n)^{-1}$, $diag(A)$ is a diagonal matrix with $A_i$ as the $i$-th diagonal entry, and $\big((1-A)\hat{\pi}_1(X)\big)_n :=\begin{pmatrix}
        (1-A_1)\hat{\pi}_1(X_1) & (1-A_2)\hat{\pi}_1(X_2) & \cdots & (1-A_n)\hat{\pi}_1(X_n)
    \end{pmatrix}^\top$.
    \end{enumerate}
\end{proposition}

\subsection{Theoretical Analysis}
\label{sec:analyses}

\subsubsection{Bias Analysis and Robustness Properties}

We begin by showing that the true nuisance functions $\pi_1$ and $\pi_2$ satisfy conditions \eqref{eq:balance1} and \eqref{eq:balance2} with equality, as stated in the following result.
\begin{proposition}
\label{prop:true}
~\\\vspace{-5mm}
\begin{enumerate}[label=(\alph*)]
	\item For the true nuisance function $\pi_1$, we have
	\[
	E\left[\left\{(1-A)\pi_1(X)-1\right\}h_1(X)\right]=0,~~~~~~\forall h_1\in L^2(P_{0,X}).
	\]
	\item For the true nuisance functions $\pi_1$ and $\pi_2$, we have
	\[
	E\left[\left\{A\pi_2(M,X)-(1-A)\pi_1(X)\right\}h_2(M,X)\right]=0,~~~~~~\forall h_2\in L^2(P_{0,(M,X)}).
	\]
\end{enumerate}
\end{proposition}
Part~2 of Proposition~\ref{prop:true} was also leveraged by \cite{chan2016efficient}. However, our motivation is fundamentally different: we use this identity as a bias-targeting device for constructing nuisance estimators that directly control the leading bias terms of the final IF-based estimator. Moreover, we operationalize it through a nonparametric adversarial (minimax) optimization framework for balancing, yielding flexible RKHS-based estimators with closed-form solutions.

Suppose conditions \eqref{eq:balance1} and \eqref{eq:balance2} hold with equality and $\mathcal{H}_1=L^2(P_{0,X})$ and $\mathcal{H}_2=L^2(P_{0,(M,X)})$. Then, it is easy to see that $\hat\pi_1=\pi_1$ and $\hat\pi_2=\pi_2$ with probability one. However, if $\mathcal{H}_1$ and $\mathcal{H}_2$ are restricted hypothesis spaces rather than the full $L^2$ spaces, conditions \eqref{eq:balance1} and \eqref{eq:balance2} may still hold with equality even though $\hat\pi_1 \neq \pi_1$ and $\hat\pi_2 \neq \pi_2$. Nevertheless, satisfying conditions \eqref{eq:balance1} and \eqref{eq:balance2} can be sufficient for controlling the bias of the final estimator of the parameter of interest. We formalize this argument by looking at the structure of the bias of $\hat\psi^{\text{2S}}$.

\begin{theorem}
\label{thm:main}
 The bias of the estimator $\hat\psi^{\text{2S}}$ satisfies the following equality:
\begin{align*}
E[\hat\psi^{\text{2S}}]-\psi_0
&=E\left[ \left\{(1-A)\hat\pi_1(X)-1\right\}\{\mu_2(X)-\hat\mu_2(X)\}\right]\\
&~~~+E\left[\left\{A\hat\pi_2(M,X)-(1-A)\hat\pi_1(X)\right\}\{\mu_1(M,X)-\hat\mu_1(M,X)\} \right].
\end{align*}
\end{theorem}

Note that our estimator $\hat\psi^{\text{2S}}$ is based on an influence function, which always leads to a second-order bias. From the structure of the bias of the estimator $\hat\psi^{\text{2S}}$, it is evident that if there exists $h_1\in\mathcal{H}_1$ which closely approximates $\{\mu_2(X)-\hat\mu_2(X)\}$, and there exists $h_2\in\mathcal{H}_2$ which closely approximates $\{\mu_1(M,X)-\hat\mu_1(M,X)\}$, and further conditions \eqref{eq:balance1} and \eqref{eq:balance2} hold, then we can eliminate the bias of the estimator $\hat\psi^{\text{2S}}$ even if we do not have a perfect fit of $\pi_1$ and $\pi_2$. This is in fact the main motivation of our two-stage approach. It substitutes fitting $p(A=1\mid X)$ and $p(M\mid A,X)$, which can be a difficult task, by a much weaker requirement based on balancing.

We, next, analyze the robustness of our proposed method against misspecification of the nuisance functions. Under the following positivity assumption, our approach possesses a multiple-robustness property, as stated in Proposition \ref{prop:MR}.

\begin{assumption}\label{ass:stg_pos}
    For some $\epsilon>0$, we have $p(A=a\mid X=x)>\epsilon$ and $p(M=m\mid A=a, X=x)>\epsilon$, for all $m$, $a$, and $x$.
\end{assumption}

\begin{proposition}
\label{prop:MR} Under Assumption \ref{ass:stg_pos} , 
$\hat\psi^{\text{2S}}$ is asymptotically unbiased if at least one of the following conditions holds:
\begin{enumerate}[label=(\roman*)]
\item $\|\hat\pi_1(X)-\pi_1(X)\|_2=o_p(1),$ and $\|\hat\pi_2(M,X)-\pi_2(M,X)\|_2=o_p(1)$;
\item $\|\hat\mu_1(M,X)-\mu_1(M,X)\|_2=o_p(1),$ and $\|\hat\mu_2(X)-\mu_2(X)\|_2=o_p(1)$;
\item $\|\hat\pi_1(X)-\pi_1(X)\|_2=o_p(1),$ and $\|\hat\mu_1(M,X)-\mu_1(M,X)\|_2=o_p(1)$.
\end{enumerate}
\end{proposition}

Figure \ref{fig:MR} summarizes the result in Proposition \ref{prop:MR}.

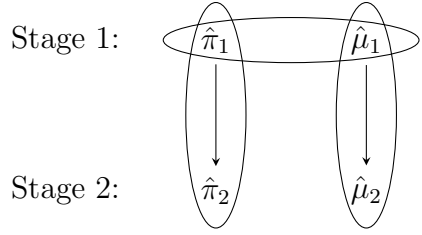
\begin{figure}[t!]
\centering
		\tikzstyle{block} = [draw, circle, inner sep=2.5pt, fill=lightgray]
		\tikzstyle{input} = [coordinate]
		\tikzstyle{output} = [coordinate]
        \begin{tikzpicture}
            \tikzset{edge/.style = {->,> = latex'}}
            \node[] (pi1) at  (-2,0) {$\hat\pi_1$};
            \node[] (pi2) at  (-2,-2) {$\hat\pi_2$};
            \node[] (mu1) at  (0,0) {$\hat\mu_1$};
            \node[] (mu2) at  (0,-2) {$\hat\mu_2$};
            \node[] (y) at  (-4,0) { Stage 1:};
            \node[] (y) at  (-4,-2) { Stage 2:};
            \draw (-1,0) ellipse (17mm and 3mm);
            \draw (-2,-1) ellipse (4mm and 15mm);
            \draw (0,-1) ellipse (4mm and 15mm);                        
            \draw[-stealth] (pi1) to (pi2);                                                         
            \draw[-stealth] (mu1) to (mu2);                                                         
        \end{tikzpicture}
        \caption{Nuisance functions estimated in each stage of the approach. The arrow from $\hat\pi_1$ to $\hat\pi_2$ indicates that the estimator of $\pi_1$ is used in estimating $\pi_2$. The arrow from $\hat\mu_1$ to $\hat\mu_2$ is interpreted similarly. Correct specification of any of the pairs of functions in an ellipse suffices unbiasedness of the final estimator of the parameter of interest.}
        \label{fig:MR}
\end{figure}

\subsubsection{Asymptotic Normality}

We next provide sufficient conditions for root-$n$ consistency and asymptotic normality of the estimator $\hat\psi^{\text{2S}}$.

\begin{assumption}
\label{assm:CAN}
~\\\vspace{-5mm}
\begin{enumerate}[label=(\alph*)]
	\item $\|\mu_1(M,X)-\hat\mu_{1,\ell}(M,X)\|_2=o_p(1)$, $\|\mu_2(X)-\hat\mu_{2,\ell}(X)\|_2=o_p(1)$, $\|\mu_1^2(M,X)-\hat\mu_{1,\ell}^2(M,X)\|_2=o_p(1)$, and $\|\mu_2^2(X)-\hat\mu_{2,\ell}^2(X)\|_2=o_p(1)$.
    \item $\|(\pi_1(X)-\hat\pi_{1,\ell}(X))^2\|_2=o_p(1)$ and $\|(\pi_2(M,X)-\hat\pi_{2,\ell}(M,X))^2\|_2=o_p(1)$.
    \item There exists some constant $C$ such that $\|\mu_{1}^2(M,X)\|_2$, $\|\mu_2^2(X)\|_2$, and  $E[Y^4]$, are upper bounded by $C$.
\item $\big\|p(A=0\mid X)\hat\pi_{1,\ell}(X)-1\big\|_2\|\mu_2(X)-\hat\mu_{2,\ell}(X)\|_2=o_p(n^{-1/2})$, and\\ $\big\|p(A=1\mid M,X)\hat\pi_{2,\ell}(M,X)-p(A=0\mid M,X)\hat\pi_{1,\ell}(X)\big\|_2\|\mu_1(M,X)-\hat\mu_{1,\ell}(M,X)\|_2=o_p(n^{-1/2})$.
\end{enumerate}
\end{assumption}

For any choice of the nuisance functions $\{\tilde\pi_1,\tilde\pi_2,\tilde\mu_1,\tilde\mu_2\}$, define,
\begin{align*}
    \Phi(O;\tilde\pi_1,\tilde\pi_2,\tilde\mu_1,\tilde\mu_2,\psi):=\phi(O;\tilde\pi_1,\tilde\pi_2,\tilde\mu_1,\tilde\mu_2)-\psi.
\end{align*}
We have the following result regarding the asymptotic behavior of $\hat\psi^{\text{2S}}$.
\begin{theorem}
\label{thm:CAN}
Under Assumptions \ref{ass:stg_pos} and ~\ref{assm:CAN}, the estimator $\hat\psi^{\text{2S}}$ is $\sqrt{n}$-consistent and asymptotically normal and satisfies
\[
\sqrt{n}(\hat\psi^{\text{2S}}-\psi_0)=\sqrt{n}E_n\left[\Phi(O;\pi_1,\pi_2,\mu_1,\mu_2,\psi_0)\right]+o_p(1).
\]
\end{theorem}
As a corollary of Theorem \ref{thm:CAN}, one can use the influence function of $\psi_0$ to obtain confidence intervals for the parameter of interest.
For $\ell\in\{1,\dots,L\}$, we estimate the variance of the influence function on fold $I_\ell$ as,
\begin{align*}
    \hat\sigma_{\ell}^2=E_m\big[ \Phi^2(O;\hat{\pi}_{1,\ell},\hat{\pi}_{2,\ell},\hat{\mu}_{1,\ell},\hat{\mu}_{2,\ell},\hat{\psi}^{\text{2S}})\big],
\end{align*}
and define,
\begin{align*}
    \hat\sigma^2=\frac{1}{L}\sum_{\ell=1}^L \hat\sigma_{\ell}^2.
\end{align*}
Then the $(1-\alpha)$ confidence interval of $\psi_0$ can be obtained as,
\begin{align*}
    \hat{\psi}^{\text{2S}}\pm z_{1-\alpha/2}\frac{\hat\sigma}{\sqrt{n}},
\end{align*}
where $z_{1-\alpha/2}$ is the $(1-\alpha/2)$-quantile of the standard normal distribution.

\section{Multiple-Mediator Case}\label{sec:multi}
\subsection{Influence Function-Based Estimator}
In this section, we review the influence function-based estimator of $\Delta^{M_j}$ proposed by \cite{xia2022decomposition} and discuss its multiple-robustness property. The explicit expression for the influence function of $\Delta^{M_j}$ is given in the Appendix, and we denote it by $IF_{\Delta^{M_j}}$. Based on $IF_{\Delta^{M_j}}$, one can first estimate the following nuisance functions.
\begin{enumerate}[label=(\roman*)]
	\item $ p(A=1\mid X=\cdot)$,
    \item $ p(M_j=\cdot\mid A=a,X=\cdot)$, for $a\in\{0,1\}$,
    \item $ p(M_{-j}=\cdot\mid A=a,X=\cdot)$, for $a\in\{0,1\}$,
    \item $ p(M_j=\cdot,M_{-j}=\cdot\mid A=a,X=\cdot)$, for $a\in\{0,1\}$,
	\item $E[Y\mid A=1,M_j=\cdot,M_{-j}=\cdot,X=\cdot]$,
\end{enumerate}
where $M_{-j}$ denotes the vector of all mediator variables other than the $j$th one.
Then, estimate $\Delta^{M_j}$ by
\begin{align*}
&\hat\Delta^{M_j,\text{XC}}\\
&=E_n\Big[\frac{I(A=1)}{\hat{p}(A=1\mid X)}\Big(1-\frac{\hat{p}(M_j\mid A=0,X)}{\hat{p}(M_j\mid A=1,X)}\Big)\frac{\hat{p}(M_j\mid A=1,X)\hat{p}(M_{-j}\mid A=1,X)}{\hat{p}(M_j,M_{-j}\mid A=1,X)}\\
    &~~~~~~\times\big(Y-\hat{E}[Y\mid A=1,M_j,M_{-j},X]\big)\\
    &~~~+\frac{I(A=1)}{\hat{p}(A=1\mid X)}\int \hat{E}[Y\mid A=1,M_j=m_j,M_{-j},X]\\
    &~~~~~~~~~~~~~~~~~~~~~~~\times\Big(\hat{p}(M_j=m_j\mid A=1,X)-\hat{p}(M_j=m_j\mid A=0,X)\Big)dm_j\\
    &~~~+\frac{I(A=1)}{\hat{p}(A=1\mid X)}\hat{\eta}_1(M_j,X)-\frac{I(A=0)}{1-\hat{p}(A=1\mid X)}\hat{\eta}_1(M_j,X)\\
    &~~~+\Big(1-\frac{2I(A=1)}{\hat{p}(A=1\mid X)}\Big)\int\hat{\eta}_1(m_j,X)\hat{p}(M_j=m_j\mid A=1,X)dm_j\\
    &~~~-\Big(1-\frac{I(A=1)}{\hat{p}(A=1\mid X)}-\frac{I(A=0)}{1-\hat{p}(A=1\mid X)}\Big)\int\hat{\eta}_1(m_j,X)\hat{p}(M_j=m_j\mid A=0,X)dm_j\Big],
\end{align*}
where \[
\hat \eta_1(M_j,X)=\int\hat{E}[Y\mid M_j,M_{-j}=m_{-j},X]\hat{p}(M_{-j}=m_{-j}\mid A=1,X)dm_{-j}.
\]
This estimator is multiply robust in the following sense \citep{xia2022decomposition}:
\begin{proposition}
\label{prop:MRXC}
$\hat\Delta^{M_j,\text{XC}}$ is unbiased if at least one of the following pairs of nuisance function estimators is correctly specified.
\begin{enumerate}[label=(\roman*)]
\item $\{\hat{p}(A=1\mid X=\cdot),\hat{p}(M_j=\cdot\mid A=a,X=\cdot),\hat{p}(M_{-j}=\cdot\mid A=a,X=\cdot),\\~~\hat{p}(M_j=\cdot,M_{-j}=\cdot\mid A=1,X=\cdot)\}$
\item $\{\hat{p}(A=1\mid X=\cdot),\hat{p}(M_j=\cdot\mid A=a,X=\cdot),\hat E[Y\mid A=1,M_j=\cdot,M_{-j}=\cdot,X=\cdot]\}$,
\item $\{\hat{p}(A=1\mid X=\cdot),\hat{p}(M_{-j}=\cdot\mid A=a,X=\cdot),\hat E[Y\mid A=1,M_j=\cdot,M_{-j}=\cdot,X=\cdot]\}$,
\item $\{\hat{p}(M_j=\cdot\mid A=a,X=\cdot),\hat{p}(M_{-j}=\cdot\mid A=a,X=\cdot),\hat E[Y\mid A=1,M_j=\cdot,M_{-j}=\cdot,X=\cdot]\}$.
\end{enumerate}
\end{proposition}

Similar to $\hat\psi^{\text{TTS}}$, $\hat\Delta^{M_j,\text{XC}}$ has the same issue with the complexity of the nuisance
functions and variance inflation from treatment and mediator mechanisms in the denominator. In the next section, we propose an alternative estimation strategy that mitigates these issues in the two-mediator case.

\subsection{Proposed Estimation Approach}
In this section, we present our two-stage nuisance function estimation approach. We introduce a reparametrization of the likelihood function based on the following nuisance functions.
\begin{align*}
&\pi(X=\cdot):=\frac{1}{p(A=1\mid X=\cdot)},\\
&\rho(M_j=\cdot,X=\cdot)=\frac{p(M_j=\cdot\mid A=0,X=\cdot)}{p(M_j=\cdot\mid A=1,X=\cdot)}\\
&\omega(M_j=\cdot,M_{-j}=\cdot,X=\cdot)=\frac{p(M_j=\cdot\mid A=1,X=\cdot)p(M_{-j}=\cdot\mid A=1,X=\cdot)}{p(M_j=\cdot,M_{-j}=\cdot\mid A=1,X=\cdot)}\\
&\mu(M_j=\cdot,M_{-j}=\cdot,X=\cdot)=E[Y\mid A=1,M_j=\cdot,M_{-j}=\cdot,X=\cdot]\\
&\eta_1(M_{j}=\cdot,X=\cdot)=\int E[Y\mid A=1,M_j=\cdot,m_{-j},X=\cdot]f(m_{-j}\mid A=1,X=\cdot)dm_{-j}
\end{align*}

Using the nuisances above, we further define,
\begin{itemize}
    \item $R(M_j,M_{-j},X)=\big(1-\rho(M_j,X)\big)\omega(M_j,M_{-j},X)$,
    \item $\gamma_{1,1}(X)=E[\eta_1(M_j,X)\mid A=1,X]$,
    \item $\gamma_{1,0}(X)=E[\eta_1(M_j,X)\rho(M_j,X)\mid A=1,X]$,
    \item $\gamma_{1}(X)=\gamma_{1,1}(X)-\gamma_{1,0}(X)$,
    \item $\eta_2(M_{-j},X)=E_{M_j\mid A=1,X}[\mu(M_j,M_{-j},X)\big(1-\rho(M_j,X)\big)\mid A=1,X]$.
\end{itemize}

For any choice of nuisance functions $\{\tilde\pi,\tilde\rho,\tilde\omega,\tilde\mu,\tilde\eta_1\}$, the working nuisance function for $R$ is defined as,
\begin{align*}
    \tilde{R}(m_j,m_{-j},x)=(1-\tilde\rho(m_j,x))\tilde\omega(m_j,m_{-j},x).
\end{align*}
We further consider working nuisance functions $\{\tilde\gamma_{1,1},\tilde\gamma_{1,0},\tilde\gamma_{1},\tilde\eta_2\}$. We impose the following compatibility condition:
\begin{assumption}[Compatibility]\label{ass:compatibility}
\begin{enumerate}[label=(\roman*)]
\item[]
    \item $\tilde\eta_2$ is compatible with $\{\tilde\rho,\tilde\mu\}$ in the sense that for all $x$, $m_{-j}$,
    \begin{align*}
        \tilde\eta_2(m_{-j},x)=E[\tilde\mu(M_j,m_{-j},X)\big(1-\tilde\rho(M_j,X)\big)\mid A=1,X=x].
    \end{align*}
    \item $\tilde\gamma_1$ is compatible with $\{\tilde\rho,\tilde\eta_1\}$ in the sense that for all $x$,
    \begin{align*}
    \tilde\gamma_{1,1}(x)&=E[\tilde\eta_1(M_j,X)\mid A=1,X=x],\\
        \tilde\gamma_{1,0}(x)&=E[\tilde\eta_1(M_j,X)\tilde\rho(M_j,X)\mid A=1,X=x],\\
        \tilde\gamma_1(x)&=\tilde\gamma_{1,1}(x)-\tilde\gamma_{1,0}(x).
    \end{align*}
\end{enumerate}
    
\end{assumption}

\begin{remark}
Assumption~\ref{ass:compatibility} does not require that $\eta_2$ or $\gamma_1$ be correctly specified as stand-alone nuisance components. Rather, it requires that the \emph{working} versions of these functions be constructed so that they are internally consistent with the working models for $\rho,\mu,$ and $\eta_1$ through the defining conditional-expectation relationships in the assumption.

A convenient way to guarantee compatibility is via a sequential regression (or plug-in) construction. For example, given working estimators $\tilde\mu$ and $\tilde\rho$, one may \emph{define} $\tilde\eta_2(m_{-j},x)$ as the conditional expectation
\[
\tilde\eta_2(m_{-j},x)
:=E\!\left[\tilde\mu(M_j,m_{-j},X)\big(1-\tilde\rho(M_j,X)\big)\mid A=1,X=x\right],
\]
and estimate it by regressing the pseudo-outcome $\tilde\mu(M_j,m_{-j},X)\{1-\tilde\rho(M_j,X)\}$ on $X$ among treated units. Similarly, given $\tilde\eta_1$ and $\tilde\rho$, one can obtain $\tilde\gamma_{1,1}$ and $\tilde\gamma_{1,0}$ by regressing $\tilde\eta_1(M_j,X)$ and $\tilde\eta_1(M_j,X)\tilde\rho(M_j,X)$ on $X$ among treated units, and then set $\tilde\gamma_1=\tilde\gamma_{1,1}-\tilde\gamma_{1,0}$. Under standard smoothness and regularity conditions, flexible nonparametric regression methods can consistently estimate these conditional expectations, ensuring compatibility by construction. Importantly, this strategy only requires regressions of functions of $(M_j,X)$ on $X$ within the treated group and imposes no additional modeling restrictions elsewhere.

Alternatively, compatibility can be particularly transparent in settings where the conditional distribution $p(M_j\mid A=1,X)$ is known by design. This occurs, for example, when $M_j$ is a treatment-induced \emph{randomized intermediate} whose distribution is controlled by the intervention protocol among treated units---such as an assigned dosage, a randomized encouragement intensity, or a scheduled treatment component delivered according to a known rule that depends on baseline covariates $X$ (e.g., adaptive randomization with known assignment probabilities). In such cases, the conditional expectations in Assumption~\ref{ass:compatibility} can be evaluated using the known $p(M_j\mid A=1,X)$, and the compatibility condition is automatically satisfied for nuisance functions constructed from these defining equations.
\end{remark}

We next describe the steps of our proposed estimation method in the multiple-mediator case.

\medskip
\noindent
\begin{sloppypar}{\bf Estimating the nuisance functions.} 
We partition the samples into $L$ folds $\{I_1,\dots,I_L\}$ with roughly equal size $m=n/L$. For $\ell\in\{1,\dots,L\}$, we estimate the nuisance functions as $(\hat{\pi}_{\ell},\hat{\rho}_{\ell},\hat{\omega}_{\ell},\hat{\mu}_{\ell},\hat{\eta}_{1,\ell})$ on data from all folds but $I_{\ell}$, which we denote as $I_\ell^c$. To estimate the nuisance functions, we proceed in the following two stages.
\end{sloppypar}

\medskip
\noindent
{\bf Stage 1:} 
\begin{enumerate}[label=(\roman*)]
\item 
We estimate $\pi$ using the following balancing estimator:\\ We choose $\hat\pi_{\ell}$ such that it satisfies
\begin{equation}
\label{eq:balance3}
E\left[\left\{A\hat\pi_{\ell}(X)-1\right\}g_1(X)\right]\approx0,~~~~~~\forall g_1\in\mathcal{G}_1,
\end{equation}
where $\mathcal{G}_1\subset L^2(P_{0,X})$ is a function space chosen by the researcher.
\item We estimate $\mu$ by regressing $Y$ on $\{M_1,M_2,,\dots,M_k,X,A=1\}$ to obtain $\hat\mu_{\ell}$.
\end{enumerate}

\noindent
{\bf Stage 2:} 
\begin{enumerate}[label=(\roman*)]
\item 
We estimate $\rho$ using the following balancing estimator:\\ Given $\hat\pi_{\ell}$ from Stage 1, we choose $\hat\rho_{\ell}$ such that it satisfies
\begin{equation}
\label{eq:balance4}
E\left[\left\{A\hat\pi_{\ell}(X)\hat\rho(M_j,X)-\frac{(1-A)\hat\pi_{\ell}(X)}{\hat\pi_{\ell}(X)-1}\right\}g_2(M_j,X)\right]\approx0,~~~~~~\forall g_2\in\mathcal{G}_2,
\end{equation}
where $\mathcal{G}_2\subset L^2(P_{0,(M_j,X)})$ is a function space chosen by the researcher.
\item 
We estimate $\omega$ using the following balancing estimator:\\ Given $\hat\pi_{\ell}$ from Stage 1, we choose $\hat\omega_{\ell}$ such that it satisfies
\begin{equation}
\label{eq:balance5}
E\left[A\hat\pi_{\ell}(X)\hat\omega(M_j,M_{-j},X)g_3(M_j,M_{-j},X)-A\hat\pi_{\ell}(X)\hat{f}(M_{-j},X)\right]\approx0,~~~~~~\forall g_3\in\mathcal{G}_3,
\end{equation}
where $\mathcal{G}_3\subset L^2(P_{0,(M_1,M_2,\dots,M_k,X)})$ is a function space chosen by the researcher, and $\hat{f}(m_{-j},x)$ is an estimator for
\begin{equation*}
    E[g_3(M_j,m_{-j},x)\mid A=1,X]~~~~~~\forall m_{-j},x.
\end{equation*}
\item Given $\hat\mu_{\ell}$ from Stage 1, for each fixed value $m_j$, we estimate $\eta_1(m_j,X)$ by first fixing $M_j=m_j$ in $\mu_{\ell}$, and then regressing $\hat\mu_{\ell}(M_j=m_j,M_{-j},X)$ on $\{M_{-j},X,A=1\}$ to obtain $\hat\eta_{1,\ell}$.
\end{enumerate}

\medskip
\noindent
{\bf Estimating the parameter of interest.}
For any choice of nuisance functions $\{\tilde\pi,\tilde\rho,\tilde\omega,\tilde\mu,\tilde\eta_1\}$, let
\begin{align*}
\zeta(O;\tilde\pi,\tilde\rho,\tilde\omega,\tilde\mu,\tilde\eta_1,\tilde{R},\tilde\gamma_{1,1},\tilde\gamma_{1,0},\tilde\gamma_{1},\tilde\eta_2):=&A\tilde\pi(X)\{\tilde{R}(M_j,M_{-j},X)\big(Y-\tilde\mu(M_j,M_{-j},X)\big)\\
&+\tilde\eta_2(M_{-j},X)-\tilde\gamma_1(X)+\tilde\eta_1(M_j,X)-\tilde\gamma_{1,1}(X)\}\\
&+\frac{(1-A)\tilde\pi(X)}{\tilde\pi(X)-1}\{\tilde\gamma_{1,0}(X)-\tilde\eta_1(M_j,X)\}+\tilde\gamma_1(X).
\end{align*}

For all $\ell$, let $\hat\Delta^{M_j,\text{2S}}_\ell$ be the estimator of $\Delta^{M_j}$ defined as,
\begin{align*}
\hat\Delta^{M_j,\text{2S}}_\ell
=E_m\left[ \zeta(O;\hat\pi_{\ell},\hat\rho_{\ell},\hat\omega_{\ell},\hat\mu_{\ell},\hat\eta_{1,\ell},\tilde{R},\tilde\gamma_{1,1},\tilde\gamma_{1,0},\tilde\gamma_1,\tilde\eta_2)\right],
\end{align*}
where the empirical expectation is evaluated on $I_\ell$. Our final estimator of the parameter of interest will be the following.
\begin{align*}
    \hat\Delta^{M_j,\text{2S}}=\frac{1}{L}\sum_{\ell=1}^L \hat\Delta^{M_j,\text{2S}}_\ell.
\end{align*}

We propose the following nonparametric minimax estimator for $\pi,\rho$, and $\omega$ in the multiple-mediator case.

\begin{proposition}
\label{prop:stabel_pi}
$\breve{\pi}$ obtained from the following minimax optimization satisfies Equation \eqref{eq:balance3} with equality.
\begin{equation*}
\breve{\pi}=\arg\min_{\pi\in L^2(P_{0,X})}\max_{g_1\in L^2(P_{0,X})}E\left[\left\{A\pi(X)-1\right\}g_{1}(X)-\frac{1}{4}g^2_{1}(X)\right].
\end{equation*}
\end{proposition}

Therefore, we estimate $\pi$ as,
\begin{equation}
\label{eq:estpi}
\hat\pi=\arg\min_{\pi\in\Pi}\max_{g_1\in\mathcal{G}_1}E_n\left[\left\{A\pi(X)-1\right\}g_{1}(X)-\frac{1}{4}g^2_{1}(X)\right]+R_{n}^3(\pi,g_1),
\end{equation}
where $R_n^3(\cdot,\cdot)$ is a regularizer.

\begin{proposition}
\label{prop:stabel_rho}
For any give $\hat\pi$,
$\breve{\rho}$ obtained from the following minimax optimization satisfies Equation \eqref{eq:balance4} with equality.
\begin{align*}
\breve{\rho}=\arg\min_{\rho\in L^2(P_{0,(M_j,X)})}\max_{g_2\in L^2(P_{0,(M_j,X)})}E\Big[&\Big\{A\hat\pi(X)\rho(M_j,X)-\frac{(1-A)\hat\pi(X)}{\hat\pi(X)-1}\Big\}g_{2}(M_j,X)\\
&-\frac{1}{4}g^2_{2}(M_j,X)\Big].
\end{align*}
\end{proposition}

Therefore, to estimate $\rho$ in stage 2, we can use the following estimator.
\begin{equation}
\label{eq:estrho}
\hat\rho=\arg\min_{\rho\in\Re}\max_{g_2\in\mathcal G_2}E_n\Big[\Big\{A\hat\pi(X)\rho(M_j,X)-\frac{(1-A)\hat\pi(X)}{\hat\pi(X)-1}\Big\}g_{2}(M_j,X)-\frac{1}{4}g^2_{2}(M_j,X)\Big]+R_{n}^4(\rho,g_2),
\end{equation}
where $\Re\subset L^2(P_{0,(M_j,X)})$ is a user-specified function class, and $R_n^4(\cdot,\cdot)$ is a regularizer.

\begin{proposition}
\label{prop:stabel_omega}
For any give $\hat\pi$,
$\breve{\omega}$ obtained from the following minimax optimization satisfies Equation \eqref{eq:balance5} with equality.
\begin{align*}
\breve{\omega}=\arg\min_{\omega\in L^2(P_{0,(M_1,\dots,M_k,X)})}\max_{g_3\in L^2(P_{0,(M_1,\dots,M_k,X)})}&E\Big[A\hat\pi(X)\omega(M_j,M_{-j},X)g_{3}(M_j,M_{-j},X)\\
&-A\hat\pi(X)\hat{f}(M_{-j},X)-\frac{1}{4}g^2_{3}(M_j,M_{-j},X)\Big],
\end{align*}
where $\hat{f}(m_{-j},x)$ is an estimator for $E[g_{3}(M_j,m_{-j},X)\mid A=1,X=x]$ for any $m_{-j}$, $x$.
\end{proposition}

Therefore, we can estimate $\omega$ in stage 2 as,
\begin{align}
\label{eq:estomega}
\hat\omega=\arg\min_{\omega\in\Omega}\max_{g_3\in\mathcal G_3}E_n\Big[&A\hat\pi(X)\omega(M_j,M_{-j},X)g_{3}(M_j,M_{-j},X)-A\hat\pi(X)\hat{f}(M_{-j},X)\nonumber\\
&-\frac{1}{4}g^2_{3}(M_j,M_{-j},X)\Big]
+R_{n}^5(\omega,g_3),
\end{align}
where $\Omega\subset L^2(P_{0,(M_1,M_2,\dots,M_k,X)})$ is a user-specified function class, and $R_n^5(\cdot,\cdot)$ is a regularizer.

One may use any standard parametric or nonparametric regression method to estimate $\mu$ in Stage 1 and $\eta_1$ in Stage 2 of the approach.

\subsubsection{A Kernel-Based Approach to Solving for $\omega$}
\label{sec_drest_multi}
We omit the derivation for $\pi$ and $\rho$ when RKHSes are used as the corresponding hypothesis classes, as the arguments closely parallel those for $\pi_1$ and $\pi_2$ in the single-mediator setting. Given $\hat\pi$, we focus on deriving the solution for 
$\omega$ when an RKHS is adopted as its hypothesis class.

Let $\mathcal{G}_3$, $\Omega$, and $\Gamma$ be RKHSes with kernels $K_{\mathcal{G}_3}$, $K_{\Omega}$, and $K_{\Gamma}$, respectively, equipped with the RKHS norms $\|\cdot\|_{\mathcal{G}_3}$, $\|\cdot\|_{\Omega}$, and $\|\cdot\|_{\Gamma}$. We propose the following Tikhonov regularization-based optimization problems:
\begin{align}
\label{eq:estomegareg}
\hat\omega=\arg\min_{\omega\in\Omega}\max_{g_3\in\mathcal G_3}E_n\Big[&A\hat\pi(X)\omega(M_j,M_{-j},X)g_{3}(M_j,M_{-j},X)-A\hat\pi(X)\hat{f}(M_{-j},X)\nonumber\\
&-\frac{1}{4}g^2_{3}(M_j,M_{-j},X)\Big]
-\lambda_{\mathcal{G}_3}\|g_3\|_{\mathcal{G}_3}^2+\lambda_{\Omega}\|\omega\|_{\Omega}^2,
\end{align}
where
\begin{align*}
\hat{f}(m_{-j},x)&=\langle g_3,\hat\mu_{\lambda,m_{-j},n_1}(x)\rangle_{\mathcal{G}_3},
\end{align*}
in which $n_1$ is the sample size in the subset with $\{A_i=1\}$ and
\begin{align}\label{eq:estg3}
    \hat{\mu}_{\lambda,m_{-j},n_1}=\arg\min_{\mu\in\mathcal{H}_{\Gamma}}\sum_{r:A_r=1}\|K_{\mathcal{G}_3}\big((M_{j,r},m_{-j},X_r),\cdot\big)-\mu(X_r)\|_{\mathcal{G}_3}^2+\lambda\|\mu\|_{\mathcal{H}_\Gamma}^2,
\end{align}
where $\mathcal{H}_\Gamma$ is an RKHS of functions from $\mathcal{X}$ to $\mathcal{G}_3$ with elements defined as functions via $(g_3,K_{\Gamma}(x,\cdot))(x'):=K_{\Gamma}(x,x')g_3$, with inner product
\begin{align*}
    \langle g_{3,1}K_{\Gamma}(x,\cdot),g_{3,2}K_{\Gamma}(x',\cdot)\rangle_{\Gamma}:=\langle g_{3,1},g_{3,2}\rangle K_{\Gamma}(x,x'),
\end{align*}
for all $g_{3,1},g_{3,2}\in\mathcal{G}_3$.

The estimation of the conditional expectation $E[g_{3}(M_j,m_{-j},X)\mid A=1,X=x]$ builds on the framework of conditional mean embeddings in reproducing kernel Hilbert spaces, as developed by \cite{song2009hilbert} and \cite{grunewalder2012conditional}.

\begin{sloppypar} Define $K_{\mathcal{G}_3,n}=(K_{\mathcal{G}_3}((M_j,M_{-j},X)_i, (M_j,M_{-j},X)_s))_{i,s=1}^n$, $K_{\Omega,n}=(K_{\Omega}((M_j,M_{-j},X)_i, (M_j,M_{-j},X)_s)))_{i,s=1}^n$, and $K_{\Gamma,n_1}=(K_{\Gamma}(X_i,X_s))_{i,s:A_i=1,A_s=1}$ as the empirical kernel matrices corresponding to spaces $\mathcal{G}_3$, $\Omega$ and $\Gamma$ respectively. The following result provides closed-form solutions for $\hat\omega$.
\end{sloppypar}

\begin{proposition}\label{prop:omegasol}

Equation \eqref{eq:estomegareg} achieves its optimum at $\hat{\omega}=\sum_{i=1}^n\alpha_{\omega,i}K_{\Omega}((M_j,M_{-j},X)_i,\cdot)$, with $\alpha_{\omega}$ defined as,
    \begin{align*}
    \alpha_{\omega}&=\Bigg(K_{\Omega,n}diag(A\hat{\pi}(X))\frac{1}{4}K_{\mathcal{G}_3,n}\Big(\frac{1}{4n}K_{\mathcal{G}_3,n}+\lambda_{\mathcal{G}_3}I_n\Big)^{-1}diag(A\hat{\pi}(X))K_{\Omega,n}+n^2\lambda_{\Omega}K_{\Omega,n}\Bigg)^{\dagger}\\
    &~~~~K_{\Omega,n}diag(A\hat{\pi}(X))\frac{1}{4}\Big(\frac{1}{4n}K_{\mathcal{G}_3,n}+\lambda_{\mathcal{G}_3}I_n\Big)^{-1}C(A\hat{\pi}(X))_n
\end{align*}
\begin{sloppypar}where $diag(A\hat{\pi}(X))$ is a diagonal matrix with $A_i\hat\pi(X_i)$ as the $i$-th diagonal entry, $\big(A\hat{\pi}(X)\big)_n :=\begin{pmatrix}
        A_1\hat{\pi}(X_1) & A_2\hat{\pi}(X_2) & \cdots & A_n\hat{\pi}(X_n)
        \end{pmatrix}^\top$, and $C$ is an $n\times n$ matrix with $C_{si}=\sum_{r:A_r=1}K_{\mathcal{G}_3}\big((M_{j,s},M_{-j,s},X_s),(M_{j,r},M_{-j,i},X_r)\big)((K_{\Gamma,n_1}+\lambda I_{n_1})^{-1}K_{\Gamma,n_1n})_{ri}$, where $K_{\Gamma,n_1n}$ is an $n_1\times n$ matrix with $(K_{\Gamma,n_1n})_{rs}=K(X_r,X_s)$ for $r,s$ such that $A_r=1$ and $s=1,\dots,n$.
\end{sloppypar}

\end{proposition}

\begin{remark}
    Solving for $\omega$ requires first estimating $\hat{f}$ via Equation \eqref{eq:estg3} to obtain $\hat{\mu}_{\lambda,m_{-j},n_1}$. The resulting estimate $\hat{\mu}_{\lambda,m_{-j},n_1}$ is reflected in the estimated coefficient vector $\alpha_\omega$ through the matrix $C$. Although $C$ does not admit a simple matrix-matrix multiplication representation and must be constructed elementwise, its computation remains efficient. Specifically, the only matrix inversion involved in forming each element of $C$ is $(K_{\Gamma,n_1}+\lambda I_{n_1})^{-1}$, which needs to be computed only once. Consequently, the overall construction of $C$, and hence the construction of $\alpha_\omega$, is not computationally intensive.
\end{remark}

\subsection{Bias Analysis}
We begin by showing that the true nuisance functions $\pi$, $\rho$ and $\omega$ satisfy conditions \eqref{eq:balance3}, \eqref{eq:balance4}, and \eqref{eq:balance5} with
equality, as stated in the following result.
\begin{proposition}
\label{prop:true_multi}
~\\\vspace{-5mm}
\begin{enumerate}[label=(\alph*)]
	\item For the true nuisance function $\pi$, we have
	\[
	E\left[\left\{A\pi(X)-1\right\}g_1(X)\right]=0,~~~~~~\forall g_1\in L^2(P_{0,X}).
	\]
	\item For the true nuisance functions $\pi$ and $\rho$, we have
	\[
	E\left[\left\{A\pi(X)\rho(M_j,X)-\frac{(1-A)\pi(X)}{\pi(X)-1}\right\}g_2(M_j,X)\right]=0,~~~~~~\forall g_2\in L^2(P_{0,(M_j,X)}).
	\]
    \item For the true nuisance functions $\pi$ and $\omega$, we have
	\begin{align*}
	E[A\pi(X)\omega(M_j,M_{-j},X)&g_3(M_j,M_{-j},X)-A\pi(X)f_3(M_{-j},X)]=0,\\
    &\forall g_3\in L^2(P_{0,(M_1,
    \dots,M_k,X)}),
	\end{align*}
    where $f_3(m_{-j},x)=E[g_3(M_j,m_{-j},X)\mid A=1,X=x]$ for any $m_{-j}$, $x$.
\end{enumerate}
\end{proposition}

In the multiple-mediator setting, satisfying conditions \eqref{eq:balance3}, \eqref{eq:balance4}, and \eqref{eq:balance5} is sufficient to control the bias of the final estimator $\hat\Delta^{M_j,\text{2S}}$, even when the corresponding nuisance functions are estimated within restricted hypothesis spaces. As in the single-mediator case, this conclusion follows from examining the structure of the bias of $\hat\Delta^{M_j,\text{2S}}$.

\begin{theorem}
\label{thm:main_multi}
 The bias of the estimator $\hat\Delta^{M_j,\text{2S}}$ satisfies the following equality:
\begin{align*}
&E[\hat\Delta^{M_j,\text{2S}}]-\Delta^{M_j}\\
&=E\Big[\Big\{A\hat{\pi}(X)-1\Big\}E\big[\hat{R}(M_j,M_{-j},X)\big\{\mu(M_j,M_{-j},X)-\hat{\mu}(M_j,M_{-j},X)\big\}\mid A=1,X\big]\Big]\\
    &~~~+E\Big[\Big\{A\hat{\pi}(X)-1\Big\}
    \Big\{E[\hat{\eta}_2(M_{-j},X)\mid A=1,X]-\hat{\gamma}_1(X)\Big\}\Big]\\
    &~~~+E\Big[\Big\{\frac{(1-A)\hat{\pi}(X)}{\hat{\pi}(X)-1}-1\Big\}\Big\{\hat{\gamma}_{1,0}(X)-E[\hat{\eta}_1(M_j,X)\mid A=0,X]\Big\}\Big]\\
    &~~~+{E\Big[E\big[\{\hat{R}(M_j,M_{-j},X)-R(M_j,M_{-j},X)\}\big\{\mu(M_j,M_{-j},X)-\hat{\mu}(M_j,M_{-j},X)\big\}\mid A=1,X\big]\Big]}\\
&~~~+E\Big[E\big[\{\eta_1^\ast(M_j,X)-\hat{\eta}_1(M_j,X)\}\big\{\rho(M_j,X)-\hat{\rho}(M_j,X)\big\}\mid A=1,X\big]\Big],
\end{align*}
where $\eta_1^\ast(M_j,X)=\int \hat{\mu}(M_j,m_{-j},X)p(M_{-j}=m_{-j}\mid A=1,X)dm_{-j}$.
\end{theorem}

The influence function-based estimator $\hat\Delta^{M_j,\text{2S}}$ exhibits a second-order bias. An examination of its bias structure shows that this bias can be eliminated if the function classes $\mathcal{G}_1$, $\mathcal{G}_2$, and $\mathcal{G}_3$ contain elements that adequately approximate the corresponding outcome regression estimation errors, and if the balancing conditions \eqref{eq:balance3}, \eqref{eq:balance4}, and \eqref{eq:balance5} are satisfied. Importantly, bias control does not require correct specification of the nuisance functions $\pi$, $\rho$ and $\omega$. This observation motivates our two-stage procedure in the multiple-mediator setting, which replaces the difficult task of modeling joint and conditional mediator distributions with weaker balancing requirements.

We next study the robustness of the proposed method to misspecification of the nuisance functions. Under the positivity condition stated in Assumption \ref{ass:stg_pos_multi}, the estimator enjoys a multiple-robustness property, formalized in Proposition \ref{prop:MR_multi}.

\begin{assumption}\label{ass:stg_pos_multi}
    For some $\epsilon>0$, we have $p(A=a\mid X=x)>\epsilon$, $p(M_j=m_j\mid A=a, X=x)>\epsilon$ for $j=1,2,\dots,k$, and $p(M_1=m_1,M_2=m_2,\dots,M_j=m_j\mid A=a, X=x)>\epsilon$, for all $m_1$, $m_2$,$\dots$, $m_k$ $a$, and $x$.
\end{assumption}

\begin{proposition}
\label{prop:MR_multi} Under Assumptions \ref{ass:compatibility} and \ref{ass:stg_pos_multi}, 
$\hat\Delta^{M_j,\text{2S}}$ is unbiased if at least one of the following sets of nuisance function estimators is correctly specified:
\begin{enumerate}[label=(\roman*)]
\item $\{\hat\pi,\hat\rho,\hat\omega\}$;
\item $\{\hat\pi,\hat\rho,\hat\mu\}$; 
\item $\{\hat\pi,\hat\mu,\hat\eta_1\}$;
\item  $\{\hat\mu,\rho,\hat\omega,\hat\eta_1\}$.
\end{enumerate}
\end{proposition}

\section{Simulation Study}
\label{sec:simulation}
In this section, we evaluate our proposed method on synthetic data for estimating the mediation functional $\psi_0$. We consider a data generating process (DGP) which satisfies Assumptions~\ref{assumption:med} and~\ref{ass:stg_pos}, and study the robustness property of our estimator under the scenarios specified in Assumption~\ref{assm:CAN}. In particular, consider a three-dimensional covariate vector $X=(X_1,X_2,X_3)$. The continuous covariates $X_1$ and $X_2$ are independently generated from the standard normal distribution $\mathcal{N}(0,1)$, and the binary covariate $X_3$ is generated from a Bernoulli distribution with $p(X_3=1)=0.5$. Treatment $A$ is generated according to
\begin{align*}
    p(A=1\mid X_1,X_2,X_3)=\text{expit}(0.35 + 0.1X_1 - 0.2X_2 + 0.3X_3 - 0.15X_1^2 + 0.25X_2X_3).
\end{align*}
The mediator $M$ is drawn from a two-component Gaussian mixture centered around
\begin{align*}
    \mu_M(X_1,X_2,X_3,A)=2.4 + 1.6X_1 - 1.2X_2 + 2.2X_3 + 0.8A + 1.5X_1^2 - 1.8X_2X_3.
\end{align*}
Specifically,
\begin{align*}
    M\sim\begin{cases}
        \mathcal{N}(\mu_M(X_1,X_2,X_3,A)-\delta, \sigma^2), & \text{with probability }p_{\text{left}},\\
        \mathcal{N}(\mu_M(X_1,X_2,X_3,A)+\delta, \sigma^2), & \text{with probability }1-p_{\text{left}},
    \end{cases}
\end{align*}
where $\sigma^2=1$ is the within-mode variance, $\delta=4$ is half the separation between the two modes, and $p_{\text{left}}=0.5$ controls the mixing proportion. Finally, the outcome $Y$ is generated from a Gaussian distribution with mean function
\begin{align*}
    &\mu_Y(X_1,X_2,X_3,M,A)\\
    &=6.4 + 2.7X_1 + 0.7X_2 - 3.6X_3 - 0.8X_2^2 + 2.5X_1A + 3.2MA - 2.9M + 4.5A,
\end{align*}
and variance fixed at 1.

We consider sample sizes $n\in\{1000,2000,4000\}$ and repeat simulations under each sample size for 300 times. Cross-fitting is implemented with four folds, and hyperparameters for estimating $\pi_1$ and $\pi_2$ are tuned via cross-validation using a 75/25 train-validation split. For comparison, we also present the results for the TTS estimator $\hat{\psi}^{\text{TTS}}$ and a naive estimator $\hat{\psi}^{\text{naive}}=\frac{1}{L}\sum_{\ell=1}^L E_m[\hat\mu_{2,\ell}(X)]$, which is constructed from the identification formula $E\{E[E[Y \mid A=1,M,X] \mid A=0,X]\}$. In our implementation, $\mu_1$ and $\mu_2$ are estimated via random forests. The naive estimator uses the same $\mu_2$, while the TTS estimator also relies on random forests for nuisance estimation, except for the conditional density of the mediator, which is obtained using kernel density estimator. Replication code for the simulation study is available at \href{https://github.com/changliues/TwoStage.git}{https://github.com/changliues/TwoStage.git}.

\begin{table}[t]
\centering
\caption{Results for comparing TTS, naive, and two-stage estimators.}
\begin{tabular}{c | c c c c c} 
 \hline
 Sample Size & Estimators & Mean Bias & Median Bias & RMSE & SD   \\ [0.5ex] 
 \hline
n=1000 & $\hat{\psi}^{\text{TTS}}$ & $2.17\times10^{33}$ & $-0.662\times10^{-1}$ & $3.70\times10^{34}$  & $3.70\times10^{34}$  \\
 & $\hat{\psi}^{\text{naive}}$  & $1.59\times10^{-1}$ & $1.70\times10^{-1}$ & $2.69\times10^{-1}$ & $2.17\times10^{-1}$ \\
 & $\hat{\psi}^{\text{2S}}$ & $1.68\times10^{-1}$ & $1.80\times10^{-1}$ & $2.76\times10^{-1}$ & $2.18\times10^{-1}$ \\
\hline
n=2000 & $\hat{\psi}^{\text{TTS}}$  & $-4.15\times10^{39}$ & $2.47\times10^{-1}$ & $7.19\times10^{40}$  & $7.19\times10^{40}$ \\
& $\hat{\psi}^{\text{naive}}$ & $1.29\times10^{-1}$ & $1.40\times10^{-1}$ & $1.81\times10^{-1}$ & $1.27\times10^{-1}$ \\
 & $\hat{\psi}^{\text{2S}}$ & $1.18\times10^{-1}$ & $1.22\times10^{-1}$ & $1.76\times10^{-1}$ & $1.31\times10^{-1}$ \\
 \hline
 n=4000 & $\hat{\psi}^{\text{TTS}}$ & $3.70\times10^{60}$ & $-7.70\times10^{-1}$ & $1.07\times10^{149}$  & $1.07\times10^{149}$ \\
 & $\hat{\psi}^{\text{naive}}$ & $1.13\times10^{-1}$ & $1.16\times10^{-1}$ & $1.48\times10^{-1}$ & $0.965\times10^{-1}$ \\
& $\hat{\psi}^{\text{2S}}$ & $0.686\times10^{-1}$ & $0.696\times10^{-1}$ & $1.20\times10^{-1}$ & $0.986\times10^{-1}$ \\
 \hline
\end{tabular}
\label{table:comparison}
\end{table}

\begin{table}[t]
\centering
\caption{95\% confidence intervals and coverage for the two-stage estimator.}
\begin{tabular}{c | c c} 
 \hline
 Sample Size & 95\% CI & Coverage   \\ [0.5ex] 
 \hline
n=1000 & [9.590, 10.346] & 0.870  \\
\hline
n=2000 & [9.643, 10.191] &  0.883 \\
 \hline
 n=4000 & [9.670, 10.066] &  0.910 \\
 \hline
\end{tabular}
\label{table:inference}
\end{table}

\begin{table}[t]
\centering
\caption{Robustness results for the two-stage estimator: Bias (RMSE).}
\begin{tabular}{c | c c c c c} 
 \hline
Sample size & all true & $\{\pi_1,\pi_2\}$ & $\{\mu_1,\mu_2\}$ & $\{\pi_1,\mu_1\}$ & all false\\ [0.5ex] 
\hline
n=1000 & 0.168 (0.276) & 0.156 (0.276) & 0.189 (0.287) & 0.163 (0.273) & 0.163 (0.273) \\
\hline
n=2000 & 0.118 (0.176) & 0.137 (0.193) & 0.184 (0.224) & 0.155 (0.202) & 0.185 (0.224) \\
 \hline
 n=4000 & 0.0686 (0.120) & 0.108 (0.149) & 0.160 (0.187) & 0.140 (0.172) & 0.203 (0.226) \\
\hline
\end{tabular}
\label{table:robustness}
\end{table}

Table~\ref{table:comparison} compares the TTS estimator, naive estimator, and two-stage estimator. The TTS estimator is highly unstable when the mediator is continuous, with mean bias and RMSE becoming excessively large due to the density in the denominator. In contrast, the naive estimator and the two-stage estimator remain stable. However, the proposed two-stage estimator exhibits lower bias and consistently outperforms the naive estimator in large samples. 
Results in Table~\ref{table:comparison} demonstrate the consistency of the two stage estimator. Moreover, 
as shown in Table~\ref{table:inference}, coverage for the two-stage estimator approaches the nominal 95\% as sample size increases.

Table~\ref{table:robustness} illustrates the robustness property of the two-stage estimator. Each scenario corresponds to a setting in which specific subsets of nuisance functions are estimated well by the random forests, while the others are fit with suboptimal hyperparameters, matching the cases described in Proposition~\ref{prop:MR}. Across the first four scenarios, the two-stage estimator remains stable, with bias decreasing as sample size grows, and it performs best when all nuisance functions are well-estimated. Even in the all false case, the estimator still remains stable. Together with Table~\ref{table:comparison}, this confirms that the proposed two-stage estimator is both more accurate than the other two approaches and robust to nuisance function misspecification.

\section{Application: Effect of Obesity on Cardiovascular Disease}
\label{sec:application}
In this section, we apply our proposed method to study the direct and indirect effect of obesity on coronary heart disease (CHD) mediated through the Glycohemoglobin (HbA1c). HbA1c, measured as a percentage of long-term blood glucose control, is the standard biomarker for diagnosing and monitoring diabetes, which is itself a major risk factor for CHD \citep{xu2017role}. Disentangling the extent to which obesity contributes to CHD risk through HbA1c, as opposed to other biological pathways, is of both scientific and policy relevance. If a substantial portion of the effect is mediated through HbA1c, interventions focused on glycemic control may effectively reduce obesity-related CHD risk; if not, alternative mechanisms such as inflammation, lipid metabolism, or hypertension may be more important to target. Partitioning the total effect into direct and indirect components thus provides critical insight for policymakers in determining whether diabetes prevention and management strategies are sufficient to mitigate CHD burden, or whether obesity itself must be targeted more aggressively regardless of glycemic control.

We use data from the 2013--2014 cycle of National Health and Nutrition Examination Survey (NHANES, \cite{NHANES2013}) for analysis. Obesity is defined using body mass index (BMI): individuals with BMI $\leq29$ are classified as non-exposed (non-obese), and those with BMI $\geq31$ are classified as exposed (obese). To sharpen the contrast between exposure groups, we exclude individuals with BMI in the range 29--31, leaving a deliberate gap around the threshold. The outcome of interest is a binary indicator based on self-reported physician diagnosis, where $Y=1$ indicates that the individual has previously been diagnosed with CHD. We adjust for eleven covariates that are potential confounders: age, gender, race, education level, annual family income, alcohol use, smoking history, vigorous work activity, diet quality, sleeping habits, and health insurance coverage. We perform complete case analysis, yielding a final sample size of 2685 and report effects in terms of odds ratios (OR) for ease of interpretation.

To motivate the definitions of direct and indirect ORs, recall from Section \ref{sec:desc} that the ATE can be decomposed into the sum of the NIE and NDE. 
In an analogous manner, when considering binary outcomes, we define odds ratios in terms of the corresponding counterfactual event odds. The natural direct effect odds ratio is defined as
\begin{align*}
    \text{OR}_{\text{NDE}}=\frac{p(Y^{(1,M^{(0)})}=1)/(1-p(Y^{(1,M^{(0)})}=1))}{p(Y^{(0,M^{(0)})}=1)/(1-p(Y^{(0,M^{(0)})}=1))},
\end{align*}
the natural indirect effect odds ratio is defined as
\begin{align*}
    \text{OR}_{\text{NIE}}=\frac{p(Y^{(1,M^{(1)})}=1)/(1-p(Y^{(1,M^{(1)})}=1))}{p(Y^{(1,M^{(0)})}=1)/(1-p(Y^{(1,M^{(0)})}=1))},
\end{align*}
and the total effect odds ratio equals their product, $\text{OR}_{\text{TOT}}=\text{OR}_{\text{NIE}}\times\text{OR}_{\text{NDE}}$. While our earlier parameter of interest, $\theta_0 = E[Y^{(1,M^{(0)})}]$, was expressed as a mean potential outcome, note that because the outcome is binary, this mean coincides with the probability $p(Y^{(1,M^{(0)})}=1)$, which then serves as the basis for constructing odds ratios. Since the denominator of $\text{OR}_{\text{NDE}}$ and numerator of $\text{OR}_{\text{NIE}}$ do not involve the mediator mechanism, we use the estimator $\frac{1}{L}\sum_{\ell=1}^L E_m[\hat{p}_{\ell}(Y=1\mid A=a,X)]$ for $a\in\{0,1\}$.

The two-stage estimator yielded $\widehat{\text{OR}}_{\text{NIE}}=1.18$ (95\% CI: [1.05, 1.33]) and $\widehat{\text{OR}}_{\text{NDE}}=1.36$ (95\% CI: [0.97, 1.83]). The naive estimator produced comparable results, with $\widehat{\text{OR}}_{\text{NIE}}=1.17$ (95\% CI: [1.05, 1.29]) and $\widehat{\text{OR}}_{\text{NDE}}=1.37$ (95\% CI: [0.98, 1.86]).\footnote{Replication code for the application is available at \href{https://github.com/changliues/TwoStage.git}{https://github.com/changliues/TwoStage.git}.} Combining the direct and indirect components, the total effect of obesity on CHD was estimated at $\widehat{\text{OR}}_{\text{TOT}}=1.60$ (95\% CI: [1.18, 2.15]) for both estimators. The direct effect estimates align closely with those reported in a two-step Mendelian randomization analysis of BMI and CHD by \cite{xu2017role}, which found that, per 1 standard deviation (kg/$\text{m}^2$) increase in BMI, the total effect was 1.45 (95\% CI: [1.27, 1.66]), and the direct effect was 1.36 (95\% CI: [1.19, 1.57]) after excluding the indirect effect mediated through HbA1c. Differences in magnitude across studies are expected, as our analysis defines obesity using BMI categories, whereas \cite{xu2017role} consider a one-standard deviation increase in continuous BMI.

Taken together, our findings suggest that obesity results in approximately 60\% higher odds of CHD compared to non-obesity, and that, considering only the direct pathway independent of HbA1c, obesity leads to about 35\% higher odds of CHD.

\section{Conclusion}
\label{sec:conc}

When estimating a functional of interest that depends on nuisance functions, naively aiming for perfect fits for these nuisance functions can make the resulting estimator highly sensitive, as even small errors in the nuisance function estimators may be amplified, especially in structures involving ratios of density functions or scores in the denominator. Instead, it is important to investigate the role of nuisance functions in the bias structure of the estimator. Based on this vision, we proposed a new approach for estimating the nuisance functions in mediation analysis for both single- and multiple-mediator settings, tailored to the bias structure of the influence function-based estimators. Our proposed nuisance estimators are based on a nonparametric balancing that can target the bias of the parameter of interest directly and also can lead to an estimator which is less sensitive to error. We studied the multiple robustness, root-$n$ consistency, and asymptotic normality for the proposed estimators. We evaluated our proposed approach on synthetic data and applied our method to estimate the effect of obesity on coronary heart disease mediated by Glycohemoglobin, finding that obesity results in a 60\% higher odds of coronary heart disease, of which roughly 35\% is attributable to the direct pathway independent of Glycohemoglobin.

\bibliography{paper-ref.bib}
\bibliographystyle{apalike}

\newpage

\appendices

\section{Influence function of $\psi_0$}
The authors of \citep{tchetgen2012semiparametric} showed that in a nonparametric model, the efficient influence function of the mediation functional $\psi_0$ is given by
\begin{align*}
IF_{\psi_0}=&\frac{I(A=1)p(M\mid A=0,X)}{p(A=1\mid X)p(M\mid A=1,X)}\{Y-E[Y\mid A=1,M,X]\}\\
&+\frac{I(A=0)}{p(A=0\mid X)}\{E[Y\mid A=1,M,X]-\eta(1,0,X)\}+\eta(1,0,X)-\psi_0,
\end{align*}
where
\[
\eta(a',a,X)=\int E[Y\mid A=a',M=m,X]p(M=m\mid A=a,X)dm.
\]

\section{Influence function of $\Delta^{M_j}$}
The authors of \citep{xia2022decomposition} showed that in a nonparametric model, the efficient influence function of $\Delta^{M_j}$ is given by
\begin{align*}
    &IF_{\Delta^{M_j}}\\
    &=\frac{I(A=1)}{p(A=1\mid X)}\Big(1-\frac{p(M_j\mid A=0,X)}{p(M_j\mid A=1,X)}\Big)\frac{p(M_j\mid A=1,X)p(M_{-j}\mid A=1,X)}{p(M_j,M_{-j}\mid A=1,X)}\\
    &~~~~~~\times\big(Y-E[Y\mid A=1,M_j,M_{-j},X]\big)\\
    &~~~+\frac{I(A=1)}{p(A=1\mid X)}\int E[Y\mid A=1,M_j=m_j,M_{-j},X]\\
    &~~~~~~~~~~~~~~~~~~~~~~~\times\Big(p(M_j=m_j\mid A=1,X)-p(M_j=m_j\mid A=0,X)\Big)dm_j\\
    &~~~+\frac{I(A=1)}{p(A=1\mid X)}\eta_1(M_j,X)-\frac{I(A=0)}{p(A=0\mid X)}\eta_1(M_j,X)\\
    &~~~+\Big(1-\frac{2I(A=1)}{p(A=1\mid X)}\Big)\int\eta_1(m_j,X)p(M_j=m_j\mid A=1,X)dm_j\\
    &~~~-\Big(1-\frac{I(A=1)}{p(A=1\mid X)}-\frac{I(A=0)}{p(A=0\mid X)}\Big)\int\eta_1(m_j,X)p(M_j=m_j\mid A=0,X)dm_j-\Delta^{M_j},
\end{align*}
where
\[
\eta_1(M_j,X)=\int E[Y\mid A=1,M_j,M_{-j}=m_{-j},X]p(M_{-j}=m_{-j}\mid A=1,X)dm_{-j}.
\]

\section{Proofs}\label{appendix:proofs}

\begin{proof}[Proof of Proposition \ref{prop:stabel_pi1}]
We note that	
\begin{align*}
&\arg\max_{h_1}E\left[\left\{(1-A)\pi_1(X)-1\right\}h_{1}(X)-\frac{1}{4}h^2_{1}(X)\right]\\
&=\arg\min_{h_1}E\left[\left\{(1-A)\pi_1(X)-1\right\}^2+\frac{1}{4}h^2_{1}(X)-\left\{(1-A)\pi_1(X)-1\right\}h_{1}(X)\right]\\
&=\arg\min_{h_1}E\left[\left(\frac{1}{2}h_1(X)-\left\{(1-A)\pi_1(X)-1\right\}\right)^2\right]\\
&=2E\bigg[(1-A)\pi_1(X)-1\bigg| X\bigg].
\end{align*}
Therefore,
\begin{align*}
\breve{\pi}_1
&=\arg\min_{\pi_1}\max_{h_1}E\left[\left\{(1-A)\pi_1(X)-1\right\}h_{1}(X)-\frac{1}{4}h^2_{1}(X)\right]\\
&=\arg\min_{\pi_1}E\left[2\left\{(1-A)\pi_1(X)-1\right\}E\bigg[(1-A)\pi_1(X)-1\bigg| X\bigg]-E\bigg[(1-A)\pi_1(X)-1\bigg| X\bigg]^2\right]\\
&=\arg\min_{\pi_1}E\left[E\bigg[(1-A)\pi_1(X)-1\bigg| X\bigg]^2\right],
\end{align*}
which since zero is achievable by $\frac{1}{p(A=0\mid X)}$, it implies that 
\[
E\bigg[(1-A)\breve{\pi}_1(X)-1\bigg| X\bigg]=0.
\]
Therefore, for all $h_1\in\mathcal{H}_1$, 
\[
E\left[\left\{(1-A)\breve{\pi}_1(X)-1\right\}h_1(X)\right]
=E\left[E\bigg[(1-A)\breve{\pi}_1(X)-1\bigg| X\bigg]h_1(X)\right]
=0.
\]

\end{proof}

\begin{proof}[Proof of Proposition \ref{prop:stabel_pi2}]
For any give $\hat\pi_1$, we note that	
\begin{align*}
&\arg\max_{h_2}E\left[\left\{A\pi_2(M,X)-(1-A)\hat\pi_1(X)\right\}h_{2}(M,X)-\frac{1}{4}h^2_{2}(M,X)\right]\\
&=\arg\min_{h_2}E\left[\left\{A\pi_2(M,X)(1-A)\hat\pi_1(X)\right\}^2+\frac{1}{4}h^2_{2}(M,X)-\left\{A\pi_2(M,X)-(1-A)\hat\pi_1(X)\right\}h_{2}(M,X)\right]\\
&=\arg\min_{h_2}E\left[\left(\frac{1}{2}h_2(M,X)-\left\{A\pi_2(M,X)-(1-A)\hat\pi_1(X)\right\}\right)^2\right]\\
&=2E\bigg[A\pi_2(M,X)-(1-A)\hat\pi_1(X)\bigg| M,X\bigg].
\end{align*}
Therefore,
\begin{align*}
\breve{\pi}_2
&=\arg\min_{\pi_2}\max_{h_2}E\left[\left\{A\pi_2(M,X)-(1-A)\hat\pi_1(X)\right\}h_{2}(M,X)-\frac{1}{4}h^2_{2}(M,X)\right]\\
&=\arg\min_{\pi_2}E\bigg[2\left\{A\pi_2(M,X)-(1-A)\hat\pi_1(X)\right\}E\big[A\pi_2(M,X)-(1-A)\hat\pi_1(X)\big|M, X\big]\\
&~~~~~~~~~~~~~~~~~~~~ -E\big[A\pi_2(M,X)-(1-A)\pi_1(X)\big| M,X\big]^2\bigg]\\
&=\arg\min_{\pi_2}E\left[E\bigg[A\pi_2(M,X)-(1-A)\hat\pi_1(X)\bigg| M,X\bigg]^2\right],
\end{align*}
which since zero is achievable by $\frac{p(M, A=0\mid  X)\hat\pi_1(X)}{p(M, A=1\mid X)}$, it implies that 
\[
E\bigg[A\breve{\pi}_2(M,X)-(1-A)\hat\pi_1(X)\bigg| M,X\bigg]=0.
\]
Therefore, for all $h_2\in\mathcal{H}_2$, 
\begin{align*}
&E\left[\left\{A\breve{\pi}_2(M,X)-(1-A)\pi_1(X)\right\}h_2(M,X)\right]\\
&
=E\left[E\bigg[A\breve{\pi}_2(M,X)-(1-A)\pi_1(X)\bigg| M,X\bigg]h_2(M,X)\right]=0.
\end{align*}

\end{proof}

\begin{proof}[Proof of Proposition \ref{prop:pisol}]
We first show that, for any function $\pi_1$, we have
\begin{align*}
&\max_{h_1\in\mathcal{H}}E_n\left[\left\{(1-A)\pi_1(X)-1\right\}h_{1}(X)-\frac{1}{4}h^2_{1}(X)\right]-\lambda_{\mathcal{H}_1}\|h_1\|_{\mathcal{H}_1}^2\\
&=\{\xi_n(\pi_1)\}^\top K_{\mathcal{H}_1,n}(\frac{1}{n}K_{\mathcal{H}_1,n}+\lambda_{\mathcal{H}_1} I_n)^{-1}\{\xi_n(\pi_1)\},
\end{align*}
where $\xi_n(\pi_1)=\frac{1}{n}\big((1-A_i)\pi_1(X_i)-1\big)_{i=1}^n$.

To see this, by representer theorem, the solution to the maximization will be of the form,
\begin{align*}
    h_1(x)=\sum_{j=1}^n\alpha^\ast_j K_{\mathcal{H}_1}(X_j,x),
\end{align*}
and it suffices to find the coefficients $\alpha^\ast=(\alpha_j^\ast)_{j=1}^n$.

Note that we have,
\begin{align*}
    E_n[\{(1-A)\pi_1(X)-1\}h_1(X)]&=\sum_{i=1}^nh_1(X_i)\{\xi_n(\pi_1)\}_i={\alpha^\ast}^\top K_{\mathcal{H}_1,n}\{\xi_n(\pi_1)\},\\
    E_n[h_1^2(X)]&=\frac{1}{n}{\alpha^\ast}^\top K_{\mathcal{H}_1,n}^2{\alpha^\ast},\\
    \|h_1\|_{\mathcal{H}_1}^2&={\alpha^\ast}^\top K_{\mathcal{H}_1,n}{\alpha^\ast}.
\end{align*}

Therefore we have,
\begin{align*}
    &E_n\Big[\{(1-A)\pi_1(X)-1\}h_1(X)-\frac{1}{4}h_1^2(X)\Big]-\lambda_{\mathcal{H}_1}\|h_1\|_{\mathcal{H}_1}^2\\
    &={\alpha^\ast}^\top K_{\mathcal{H}_1,n}\{\xi_n(\pi_1)\}-{\alpha^\ast}^\top (\frac{1}{4n}K_{\mathcal{H}_1,n}^2+\lambda_{\mathcal{H}_1}K_{\mathcal{H}_1,n}){\alpha^\ast}.
\end{align*}

Taking the derivative and setting it to zero, the optimal coefficients are obtained as,
\begin{align*}
    \alpha^\ast=\frac{1}{2}(\frac{1}{4n}K_{\mathcal{H}_1,n}+\lambda_{\mathcal{H}_1} I_n)^{-1}\{\xi_n(\pi_1)\},
\end{align*}
and we have,
\begin{align*}
    &\max_{h_1\in\mathcal{H}_1}E_n\Big[\{(1-A)\pi_1(X)-1\}h_1(X)-\frac{1}{4}h_1^2(X)\Big]-\lambda_{\mathcal{H}_1}\|h_1\|_{\mathcal{H}_1}^2\\
&=\frac{1}{4}\{\xi_n(\pi_1)\}^\top K_{\mathcal{H}_1,n}(\frac{1}{4n}K_{\mathcal{H}_1,n}+\lambda_{\mathcal{H}_1} I_n)^{-1}\{\xi_n(\pi_1)\}.
\end{align*}

Therefore, the outer minimization problem is reduced to:
\begin{align*}
    \hat{\pi}_1=\arg\min_{\pi_1\in\Pi_1}\{\xi_n(\pi_1)\}^\top \Gamma_1\{\xi_n(\pi_1)\}+\lambda_{\Pi_1}\|\pi_1\|_{\Pi_1}^2,
\end{align*}
where $\Gamma_1=\frac{1}{4}K_{\mathcal{H}_1,n}(\frac{1}{4n}K_{\mathcal{H}_1,n}+\lambda_{\mathcal{H}_1} I_n)^{-1}$.

By representer theorem, the solution to the maximization will be of the form,
\begin{align*}
    \pi_1(x)=\sum_{j=1}^n\alpha_j K_{\Pi_1}(X_j,x),
\end{align*}
and it suffices to find the coefficients $\alpha=(\alpha_j)_{j=1}^n$.

Note that,
\begin{align*}
    \xi_n(\pi_1)&=\frac{1}{n}(diag(1-A)K_{\Pi_1,n}\alpha-e_n),\\
    \|\pi_1\|_{\Pi_1}^2&=\alpha^\top K_{\Pi_1,n}\alpha.
\end{align*}

Therefore, we have,
\begin{align*}
    &\min_{\pi_1\in\Pi_1}\{\xi_n(\pi_1)\}^\top \Gamma_1\{\xi_n(\pi_1)\}+\lambda_{\Pi_1}\|\pi_1\|_{\Pi_1}^2\\
    &=\min_{\alpha\in\mathbb{R}^n}\frac{1}{n^2}\alpha^\top K_{\Pi_1,n} diag(1-A)\Gamma_1 diag(1-A)K_{\Pi_1,n}\alpha-\frac{2}{n^2}e_n^\top \Gamma_1 diag(1-A)K_{\Pi_1,n}\alpha+\lambda_{\Pi_1}\alpha^\top  K_{\Pi_1,n}\alpha+c,
\end{align*}
which is solved by 
\begin{align*}
    \alpha=\Big(K_{\Pi_1,n} diag(1-A)\Gamma_1 diag(1-A)K_{\Pi_1,n}+n^2\lambda_{\Pi_1} K_{\Pi_1,n}\Big)^{\dagger}K_{\Pi_1,n}diag(1-A)\Gamma_1 e_n.
\end{align*}

The derivation for $\beta$ is similar and omitted for brevity.

\end{proof}

\begin{proof}[Proof of Proposition \ref{prop:true}]
\begin{enumerate}[label=(\alph*)]
\item For all $h_1\in L^2(P_{0,X})$, we have
\[
E\left[\left\{(1-A)\pi_1(X)-1\right\}h_1(X)\right]
=E\left[\left\{E[1-A\mid X]\pi_1(X)-1\right\}h_1(X)\right]
=0.
\]
\item For all $h_2\in L^2(P_{0,(M,X)})$, we have
\begin{align*}
&E\left[\left\{A\pi_2(M,X)-(1-A)\pi_1(X)\right\}h_2(M,X)\right]\\
&=E\left[\left\{\frac{p(A=1\mid M,X)p(M\mid A=0,X)}{p(M, A=1\mid X)}-\frac{p(A=0\mid M,X)}{p(A=0\mid X)}\right\}h_2(M,X)\right]\\
&=E\left[\left\{\frac{p(M\mid A=0,X)}{p(M\mid X)}-\frac{p(A=0\mid M,X)}{p(A=0\mid X)}\right\}h_2(M,X)\right]\\
&=E\left[\left\{\frac{p(M,A=0\mid X)}{p(M\mid X)p(A=0\mid X)}-\frac{p(M,A=0\mid X)}{p(M\mid X)p(A=0\mid X)}\right\}h_2(M,X)\right]\\
&=0.
\end{align*}
\end{enumerate}

\end{proof}

\begin{proof}[Proof of Proposition \ref{prop:MR}]
\begin{align*}
    &E[\hat\psi^\text{2S}]-\psi_0\\
    &=E[\hat\mu_2(X)+(1-A)\hat\pi_1(X)\{\hat\mu_1(M,X)-\hat\mu_2(X)\}+A\hat\pi_2(M,X)\{Y-\hat\mu_1(M,X)\}]-\psi_0\\
    &=E\left[ \left\{(1-A)\hat\pi_{1}(X)-1\right\}\{\mu_2(X)-\hat\mu_{2}(X)\}\right]\\
    &~~~~~+E\left[\left\{A\hat\pi_{2}(M,X)-(1-A)\hat\pi_{1}(X)\right\}\{\mu_1(M,X)-\hat\mu_{1}(M,X)\} \right]\\
&=E\left[ \left\{E\left[1-A\mid X\right]\hat\pi_{1}(X)-1\right\}\{\mu_2(X)-\hat\mu_{2}(X)\}\right]\\
&~~~~~+E\left[\left\{E\left[A\mid M,X\right]\hat\pi_{2}(M,X)-E\left[1-A\mid M,X\right]\hat\pi_{1}(X)\right\}\{\mu_1(M,X)-\hat\mu_{1}(M,X)\} \right]\\
&=E\left[ \left\{p\left(A=0\mid X\right)\hat\pi_{1}(X)-1\right\}\{\mu_2(X)-\hat\mu_{2}(X)\}\right]\\
&~~~~~+E\left[\left\{p\left(A=1\mid M,X\right)\hat\pi_{2}(M,X)-p\left(A=0\mid M,X\right)\hat\pi_{1}(X)\right\}\{\mu_1(M,X)-\hat\mu_{1}(M,X)\} \right].
\end{align*}
Using Cauchy-Schwarz inequality, we have
\begingroup
\allowdisplaybreaks	
\begin{align*}
&E[\hat\psi^\text{2S}]-\psi_0\\
&\leq\Big\|p\left(A=0\mid X\right)\hat\pi_{1}(X)-1\Big\|_2\Big\|\mu_2(X)-\hat\mu_{2}(X)\Big\|_2\\
&~~~~~+\Big\|p\left(A=1\mid M,X\right)\hat\pi_{2}(M,X)-p\left(A=0\mid M,X\right)\hat\pi_{1}(X)\Big\|_2\Big\|\mu_1(M,X)-\hat\mu_{1}(M,X)\Big\|_2\\
&\leq\Big\|p\left(A=0\mid X\right)\Big(\hat\pi_{1}(X)-\pi_1(X)\Big)\Big\|_2\Big\|\mu_2(X)-\hat\mu_{2}(X)\Big\|_2\\
&~~~~~+\Big\|\frac{p\left(A=1\mid M,X\right)}{\pi_1(X)}\Big(\pi_1(X)\hat\pi_{2}(M,X)-\pi_2(M,X)\hat\pi_{1}(X)\Big)\Big\|_2\Big\|\mu_1(M,X)-\hat\mu_{1}(M,X)\Big\|_2\\
&\leq\Big\|\hat\pi_{1}(X)-\pi_1(X)\Big\|_2\Big\|\mu_2(X)-\hat\mu_{2}(X)\Big\|_2\\
&~~~~~+\Big(\Big\|\pi_1(X)\Big(\hat\pi_{2}(M,X)-\pi_2(M,X)\Big)\Big\|_2+\Big\|\pi_2(M,X)\Big(\hat\pi_{1}(X)-\pi_1(X)\Big)\Big\|_2\Big)\Big\|\mu_1(M,X)-\hat\mu_{1}(M,X)\Big\|_2\\
&\leq\Big\|\hat\pi_{1}(X)-\pi_1(X)\Big\|_2\Big\|\mu_2(X)-\hat\mu_{2}(X)\Big\|_2\\
&~~~~~+\frac{1}{\epsilon^2}\Big(\Big\|\hat\pi_{2}(M,X)-\pi_2(M,X)\Big\|_2+\Big\|\hat\pi_{1}(X)-\pi_1(X)\Big\|_2\Big)\Big\|\mu_1(M,X)-\hat\mu_{1}(M,X)\Big\|_2,
\end{align*}
\endgroup
where the last equation holds by Assumption \ref{ass:stg_pos}. Therefore if at least one of the conditions in the proposition holds, we have
\begin{align*}
E[\hat\psi^\text{2S}]-\psi_0&=o_p(1).
\end{align*}
\end{proof}

\begin{proof}[Proof of Proposition \ref{prop:stabel_pi}]
    The proof of Proposition \ref{prop:stabel_pi} follows arguments similar to those used in the proof of Proposition \ref{prop:stabel_pi1} and is therefore omitted for brevity.
\end{proof}

\begin{proof}[Proof of Proposition \ref{prop:stabel_rho}]
    The proof of Proposition \ref{prop:stabel_rho} follows arguments similar to those used in the proof of Proposition \ref{prop:stabel_pi2} and is therefore omitted for brevity.
\end{proof}

\begin{proof}[Proof of Proposition \ref{prop:stabel_omega}]
   For any given $\hat\pi$, suppose $\omega_0$ satisfies Equation \ref{eq:balance5} with equality, i.e.,
   \begin{align*}
       E\left[A\hat\pi(X)\omega_0(M_j,M_{-j},X)g_3(M_j,M_{-j},X)-A\hat\pi(X)f(M_{-j},X)\right]=0,
   \end{align*}
   where $f(m_{-j},x)=E[g_3(M_j,m_{-j},x)\mid A=1,X]$ for any $m_{-j},x$.

   For any $\omega$, we have
   \begin{align*}
       &\arg\max_{g_3}E\Big[A\hat\pi(X)\omega(M_j,M_{-j},X)g_{3}(M_j,M_{-j},X)-A\hat\pi(X)f(M_{-j},X)-\frac{1}{4}g^2_{3}(M_j,M_{-j},X)\Big]\\
       &=\arg\max_{g_3}E\Big[A\hat\pi(X)\omega(M_j,M_{-j},X)g_{3}(M_j,M_{-j},X)-A\hat\pi(X)\omega_0(M_j,M_{-j},X)g_3(M_j,M_{-j},X)\\
       &~~~~~~~~~~-\frac{1}{4}g^2_{3}(M_j,M_{-j},X)\Big]\\
       &=\arg\max_{g_3}E\Big[-\Big(A\hat\pi(X)\big(\omega(M_j,M_{-j},X)-\omega_0(M_j,M_{-j},X)\big)-\frac{1}{2}g_3(M_j,M_{-j},X)\Big)^2\\
       &~~~~~~~~~~+\Big(A\hat\pi(X)\big(\omega(M_j,M_{-j},X)-\omega_0(M_j,M_{-j},X)\big)\Big)^2\Big]\\
       &=\arg\max_{g_3}E\Big[-\Big(A\hat\pi(X)\big(\omega(M_j,M_{-j},X)-\omega_0(M_j,M_{-j},X)\big)-\frac{1}{2}g_3(M_j,M_{-j},X)\Big)^2\Big]\\
       &=A\hat\pi(X)\big(\omega(M_j,M_{-j},X)-\omega_0(M_j,M_{-j},X)\big).
   \end{align*}

   Therefore,
   \begin{align*}
       \breve{\omega}&=\arg\min_{\omega}\max_{g_3}E\Big[A\hat\pi(X)\omega(M_j,M_{-j},X)g_{3}(M_j,M_{-j},X)-A\hat\pi(X)f(M_{-j},X)-\frac{1}{4}g^2_{3}(M_j,M_{-j},X)\Big]\\
       &=\arg\min_{\omega}E\Big[\Big(A\hat\pi(X)\big(\omega(M_j,M_{-j},X)-\omega_0(M_j,M_{-j},X)\big)\Big)^2\Big],
   \end{align*}
   which since zero is achievable by $\omega_0$, it implies that
   \begin{align*}
       \breve{\omega}=\omega_0.
   \end{align*}
\end{proof}

\begin{proof}[Proof of Proposition \ref{prop:omegasol}]

We first show that, Equation \eqref{eq:estg3} achieves its optimum at,
\begin{align*}
    \hat\mu_{\lambda,m_{-j},n_1}=\sum_{l:A_l=1}K_{\Gamma}(X_l,\cdot)\sum_{r:A_r=1}(K_{\Gamma,n_1}+\lambda I_{n_1})^{-1}_{rl}K_{\mathcal{G}_3}((M_{j,r},m_{-j},X_r),\cdot).
\end{align*}

   For any fixed vector $m_{-j}$, we first propose an estimator $\hat{f}(m_{-j},x)$ of $E[g_3(M_j,m_{-j},X)\mid A=1,X=x]$, for any function $g_3\in\mathcal{G}_3$.

The argument follows the work of \cite{grunewalder2012conditional}. Since conditional expectations $E[g_3(M_j,m_{-j},X)\mid A=1,X=x]$ are linear in the argument $g_3$, when we consider $g_3\in\mathcal{G}_3$, the Riesz representation theorem implies the existence of an element $\mu_{m_{-j}}(x)\in\mathcal{G}_3$, such that $E[g_3(M_j,m_{-j},X)\mid A=1,X=x]=\langle g_3,\mu_{m_{-j}}(x)\rangle_{\mathcal{G}_3}$ for all $g_3$. A natural optimization problem associated to this approximation problem is to therefore find a function $\mu_{m_{-j}}:\mathcal{X}\rightarrow \mathcal{G}_3$ such that the following objective is small,
\begin{align*}
    \mathcal{E}[\mu_{m_{-j}}]:=\sup_{\|g_3\|_{\mathcal{G}_3}\leq 1}E\Big[(E[g_3(M_j,m_{-j},X)\mid A=1,X]-\langle g_3,\mu_{m_{-j}}(X)\rangle_{\mathcal{G}_3})^2\Big|A=1\Big].
\end{align*}

Since we do not directly observe $E[g_3(M_j,m_{-j},X)\mid A=1,X]$, we cannot directly use $\mathcal{E}[\mu_{m_{-j}}]$ for estimation. Instead, we bound this risk function with a surrogate risk function with a sample-based version.
\begin{align*}
    &\sup_{\|g_3\|_{\mathcal{G}_3}\leq 1}E\Big[(E[g_3(M_j,m_{-j},X)\mid A=1,X]-\langle g_3,\mu_{m_{-j}}(X)\rangle_{\mathcal{G}_3})^2\Big|A=1\Big]\\
    &= \sup_{\|g_3\|_{\mathcal{G}_3}\leq 1}E\Big[(E[\langle g_3,K_{\mathcal{G}_3}\big((M_j,m_{-j},X),\cdot\big)\rangle_{\mathcal{G}_3}\mid A=1,X]-\langle g_3,\mu_{m_{-j}}(X)\rangle_{\mathcal{G}_3})^2\Big|A=1\Big]\\
    &\leq \sup_{\|g_3\|_{\mathcal{G}_3}\leq 1}E\Big[(\langle g_3,K_{\mathcal{G}_3}\big((M_j,m_{-j},X),\cdot\big)-\mu_{m_{-j}}(X)\rangle_{\mathcal{G}_3})^2\Big|A=1\Big]\\
    &\leq\sup_{\|g_3\|_{\mathcal{G}_3}\leq 1}\|g_3\|_{\mathcal{G}_3}^2E\Big[\|K_{\mathcal{G}_3}\big((M_j,m_{-j},X),\cdot\big)-\mu_{m_{-j}}(X)\|_{\mathcal{G}_3}^2\Big|A=1\Big]\\
    &=E\Big[\|K_{\mathcal{G}_3}\big((M_j,m_{-j},X),\cdot\big)-\mu_{m_{-j}}(X)\|_{\mathcal{G}_3}^2\Big|A=1\Big]\\
    &=:\mathcal{E}_s[\mu_{m_{-j}}].
\end{align*}
We then replace the expectation with empirical expectation to obtain the sample-based loss,
\begin{align*}
    \hat{\mathcal{E}}_n[\mu_{m_{-j}}]:=\sum_{r:A_r=1}\|K_{\mathcal{G}_3}\big((M_{j,r},m_{-j},X_r),\cdot\big)-\mu_{m_{-j}}(X_r)\|_{\mathcal{G}_3}^2.
\end{align*}

We further add a regularization term to provide a well-posed problem and prevent overfitting,
\begin{align*}
    \hat{\mathcal{E}}_{\lambda,n}[\mu_{m_{-j}}]:=\sum_{r:A_r=1}\|K_{\mathcal{G}_3}\big((M_{j,r},m_{-j},X_r),\cdot\big)-\mu_{m_{-j}}(X_r)\|_{\mathcal{G}_3}^2+\lambda\|\mu_{m_{-j}}\|_{\mathcal{H}_\Gamma}^2,
\end{align*}
where $\mathcal{H}_\Gamma$ is an RKHS of functions from $\mathcal{X}$ to $\mathcal{G}_3$ with elements defined as functions via $(g_3,K_{\Gamma}(x,\cdot))(x'):=K_{\Gamma}(x,x')g_3$, with inner product
\begin{align*}
    \langle g_{3,1}K_{\Gamma}(x,\cdot),g_{3,2}K_{\Gamma}(x',\cdot)\rangle_{\Gamma}:=\langle g_{3,1},g_{3,2}\rangle K_{\Gamma}(x,x'),
\end{align*}
for all $g_{3,1},g_{3,2}\in\mathcal{G}_3$.

We denote the minimizer of $\hat{\mathcal{E}}_{\lambda,n}[\mu_{m_{-j}}]$ by,
$\hat\mu_{\lambda,m_{-j},n}$, and it is of the form,
\begin{align*}
    \hat\mu_{\lambda,m_{-j},n}=\sum_{l:A_l=1}K_{\Gamma}(X_l,\cdot)\sum_{r:A_r=1}(K_{\Gamma,n_1}+\lambda I_{n_1})^{-1}_{rl}K_{\mathcal{G}_3}((M_{j,r},m_{-j},X_r),\cdot),
\end{align*}
where $K_{\Gamma,n_1}\in\mathbb{R}^{n_1\times n_1}$, $(K_{\Gamma,n_1})_{rl}=K_{\Gamma}(X_r,X_l)$, for $r,l$ such that $A_r=1$ and $A_l=1$. 

Therefore, we have,
\begin{align*}
\hat{f}(m_{-j},x)&=\langle g_3,\hat\mu_{\lambda,m_{-j},n}(x)\rangle_{\mathcal{G}_3}\\
&=\sum_{r:A_r=1}g_3(M_{j,r},m_{-j},X_r)\sum_{l:A_l=1}W_{rl}K_{\Gamma}(X_l,x),
\end{align*}
where $W=(K_{\Gamma,n_1}+\lambda I_{n_1})^{-1}\in\mathbb{R}^{n_1\times n_1}$.

We then show that, given $\hat{\pi}(X)$, for any function $\omega$, we have
\begin{align*}
&\max_{g_3\in\mathcal{G}_3} E_n[A\hat{\pi}(X)\omega(M_j,M_{-j},X)g_3(M_j,M_{-j},X)-A\hat{\pi}(X)\hat{f}(M_{-j},X)-\frac{1}{4}g_3^2(M_j,M_{-j},X)]-\lambda_{\mathcal{G}_3}\|g_3\|_{\mathcal{G}_3}^2\\
&=\frac{1}{4}\{K_{\mathcal{G}_3,n}\xi_n^1(\omega)-C\xi^2_n \}^\top \Big(\frac{1}{4n}K_{\mathcal{G}_3,n}^2+\lambda_{\mathcal{G}_3}K_{\mathcal{G}_3,n}\Big)^{-1}\{K_{\mathcal{G}_3,n}\xi_n^1(\omega)-C\xi^2_n \},
\end{align*}
where $\xi_n^1(\omega)=\frac{1}{n}\big(A_i\hat{\pi}(X_i)\omega(M_{j,i},M_{3-j,i},X_i)\big)_{i=1}^n$, $
    \xi_n^2=\frac{1}{n}\big(A_i\hat{\pi}(X_i)\big)_{i=1}^n$, $C\in\mathbb{R}^{n\times n}$, and $C_{si}=\sum_{r:A_r=1}K_{\mathcal{G}_3}\big((M_{j,s},M_{3-j,s},X_s),(M_{j,r},M_{3-j,i},X_r)\big)(WK_{\Gamma,n_1n})_{ri}$, in which $K_{\Gamma,n_1n}$ is an $n_1\times n$ matrix with $(K_{\Gamma,n_1n})_{rs}=K(X_r,X_s)$ for $r,s$ such that $A_r=1$ and $s=1,\dots,n$.

To see this, by representer theorem, the solution to the maximization will be of the form,

\begin{align*}
    g_3(m_j,m_{-j},x)=\sum_{i=1}^n\alpha_{g_3,i}K_{\mathcal{G}_3}((M_{j,i},M_{3-j,i},X_i),(m_j,m_{-j},x)),
\end{align*}

and it suffices to find the coefficients $\alpha_{g_3}=(\alpha_{g_3,i})_{i=1}^n$.

\begin{align*}
    \xi_n^1(\omega)&=\frac{1}{n}\big(A_i\hat{\pi}(X_i)\omega(M_{j,i},M_{3-j,i},X_i)\big)_{i=1}^n\\
    \xi_n^2&=\frac{1}{n}\big(A_i\hat{\pi}(X_i)\big)_{i=1}^n\\
    \hat{f}(M_{3-j,i},X_i)&=\sum_{r:A_r=1}g_3(M_{j,r},M_{3-j,i},X_r)(WK_{\Gamma,n_1n})_{ri}\\
    g_3(M_{j,r},M_{3-j,i},X_r)&=\sum_{s=1}^n\alpha_{g_3,s}K_{\mathcal{G}_3}\big((M_{j,s},M_{3-j,s},X_s),(M_{j,r},M_{3-j,i},X_r)\big)\\
    \hat{f}(M_{3-j,i},X_i)&=\sum_{r:A_r=1}\sum_{s=1}^n\alpha_{g_3,s}K_{\mathcal{G}_3}\big((M_{j,s},M_{3-j,s},X_s),(M_{j,r},M_{3-j,i},X_r)\big)(WK_{\Gamma,n_1n})_{ri}\\
    &=\sum_{s=1}^n\alpha_{g_3,s}\sum_{r:A_r=1}K_{\mathcal{G}_3}\big((M_{j,s},M_{3-j,s},X_s),(M_{j,r},M_{3-j,i},X_r)\big)(WK_{\Gamma,n_1n})_{ri}\\
    &=\sum_{s=1}^n\alpha_{g_3,s}C_{si}
\end{align*}

Therefore, we have,
\begin{align*}
    &E_n[A\hat{\pi}(X)\omega(M_j,M_{-j},X)g_3(M_j,M_{-j},X)-A\hat{\pi}(X)\hat{f}(M_{-j},X)]\\
    &=\frac{1}{n}\sum_{i=1}^nA_i\hat{\pi}(X_i)\omega(M_{j,i},M_{3-j,i},X_i)g_3((M_{j,i},M_{3-j,i},X_i))-A_i\hat{\pi}(X_i)\hat{f}(M_{3-j,i},X_i)\\
    &=\alpha_{g_3}^\top K_{\mathcal{G}_3,n}\{\xi_n^1(\omega)\}-\alpha_{g_3}^\top C\{\xi^2_n\}.
\end{align*}

\begin{align*}
    E_n[\frac{1}{4}g_3^2(M_j,M_{-j},X)]&=\frac{1}{4n}\alpha_{g_3}^\top K_{\mathcal{G}_3,n}^2\alpha_{g_3},
\end{align*}

\begin{align*}
    \|g_3\|_{\mathcal{G}_3}^2&=\alpha_{g_3}^\top K_{\mathcal{G}_3,n}\alpha_{g_3}
\end{align*}

Therefore we have,
\begin{align*}
&E_n[A\hat{\pi}(X)\omega(M_j,M_{-j},X)g_3(M_j,M_{-j},X)-A\hat{\pi}(X)\hat{f}(M_{-j},X)-\frac{1}{4}g_3^2(M_j,M_{-j},X)]-\lambda_{\mathcal{G}_3}\|g_3\|_{\mathcal{G}_3}^2\\
&=\alpha_{g_3}^\top(K_{\mathcal{G}_3,n}\{\xi_n^1(\omega)\}-C\{\xi^2_n\})-\alpha_{g_3}^\top (\frac{1}{4n}K_{\mathcal{G}_3,n}^2+\lambda_{\mathcal{G}_3}K_{\mathcal{G}_3,n})\alpha_{g_3}.
\end{align*}

Taking the derivative and setting it to zero, the optimal coefficients are obtained as
\begin{align*}
    \alpha_{g_3}=\frac{1}{2}(\frac{1}{4n}K_{\mathcal{G}_3,n}^2+\lambda_{\mathcal{G}_3}K_{\mathcal{G}_3,n})^{-1}\{K_{\mathcal{G}_3,n}\xi_n^1(\omega)-C\xi^2_n \},
\end{align*}

and we have,
\begin{align*}
&\max_{g_3\in\mathcal{G}_3} E_n[A\hat{\pi}(X)\omega(M_j,M_{-j},X)g_3(M_j,M_{-j},X)-A\hat{\pi}(X)\hat{f}(M_{-j},X)-\frac{1}{4}g_3^2(M_j,M_{-j},X)]-\lambda_{\mathcal{G}_3}\|g_3\|_{\mathcal{G}_3}^2\\
&=\frac{1}{4}\{K_{\mathcal{G}_3,n}\xi_n^1(\omega)-C\xi^2_n \}^\top \Big(\frac{1}{4n}K_{\mathcal{G}_3,n}^2+\lambda_{\mathcal{G}_3}K_{\mathcal{G}_3,n}\Big)^{-1}\{K_{\mathcal{G}_3,n}\xi_n^1(\omega)-C\xi^2_n \}.
\end{align*}

Therefore, the outer minimization problem is reduced to:
\begin{align*}
\hat\omega&=\arg\min_{\omega\in\Omega} \frac{1}{4}\{K_{\mathcal{G}_3,n}\xi_n^1(\omega)-C\xi^2_n \}^\top \Big(\frac{1}{4n}K_{\mathcal{G}_3,n}^2+\lambda_{\mathcal{G}_3}K_{\mathcal{G}_3,n}\Big)^{-1}\{K_{\mathcal{G}_3,n}\xi_n^1(\omega)-C\xi^2_n \}+\lambda_{\Omega}\|\omega\|_{\Omega}^2.
\end{align*}

By representer theorem, the solution to the maximization will be of the form,
\begin{align*}
    \omega(m_j,m_{-j},x)=\sum_{i=1}^n\alpha_{\omega,i}K_{\Omega}((M_{j,i},M_{3-j,i},X_i),(m_j,m_{-j},x)),
\end{align*}
and it suffices to find the coefficients $\alpha_{\omega}=(\alpha_{\omega,i})_{i=1}^n$.

Note that,
\begin{align*}
    K_{\mathcal{G}_3,n}\xi_n^1(\omega)-C\xi^2_n&=\frac{1}{n}K_{\mathcal{G}_3,n}diag(A\hat{\pi}(X))K_{\Omega,n}\alpha_{\omega}-\frac{1}{n}C(A\hat{\pi}(X))_n,\\
    \|\omega\|_{\Omega}^2&=\alpha_{\omega}^\top K_{\Omega,n}\alpha_{\omega}.
\end{align*}

Therefore, we have,
\begin{align*}
    &\min_{\omega\in\Omega} \frac{1}{4}\{K_{\mathcal{G}_3,n}\xi_n^1(\omega)-C\xi^2_n \}^\top \Big(\frac{1}{4n}K_{\mathcal{G}_3,n}^2+\lambda_{\mathcal{G}_3}K_{\mathcal{G}_3,n}\Big)^{-1}\{K_{\mathcal{G}_3,n}\xi_n^1(\omega)-C\xi^2_n \}+\lambda_{\Omega}\|\omega\|_{\Omega}^2\\
    &=\min_{\alpha_{\omega}\in\mathbb{R}^n}\frac{1}{n^2}\alpha_{\omega}^\top K_{\Omega,n}diag(A\hat{\pi}(X))K_{\mathcal{G}_3,n}\frac{1}{4}\Big(\frac{1}{4n}K_{\mathcal{G}_3,n}^2+\lambda_{\mathcal{G}_3}K_{\mathcal{G}_3,n}\Big)^{-1}K_{\mathcal{G}_3,n}diag(A\hat{\pi}(X))K_{\Omega,n}\alpha_{\omega}\\
    &~~~-\frac{2}{n^2}\big(C(A\hat{\pi}(X))_n\big)^{\top}\frac{1}{4}\Big(\frac{1}{4n}K_{\mathcal{G}_3,n}^2+\lambda_{\mathcal{G}_3}K_{\mathcal{G}_3,n}\Big)^{-1}K_{\mathcal{G}_3,n}diag(A\hat{\pi}(X))K_{\Omega,n}\alpha_{\omega}+\lambda_{\Omega}\alpha_{\omega}^\top K_{\Omega,n}\alpha_{\omega}+c,
\end{align*}
which is solved by
\begin{align*}
    \alpha_{\omega}&=\Bigg(K_{\Omega,n}diag(A\hat{\pi}(X))K_{\mathcal{G}_3,n}\frac{1}{4}\Big(\frac{1}{4n}K_{\mathcal{G}_3,n}^2+\lambda_{\mathcal{G}_3}K_{\mathcal{G}_3,n}\Big)^{-1}K_{\mathcal{G}_3,n}diag(A\hat{\pi}(X))K_{\Omega,n}+n^2\lambda_{\Omega}K_{\Omega,n}\Bigg)^{\dagger}\\
    &~~~~K_{\Omega,n}diag(A\hat{\pi}(X))K_{\mathcal{G}_3,n}\frac{1}{4}\Big(\frac{1}{4n}K_{\mathcal{G}_3,n}^2+\lambda_{\mathcal{G}_3}K_{\mathcal{G}_3,n}\Big)^{-1}C(A\hat{\pi}(X))_n\\
    &=\Bigg(K_{\Omega,n}diag(A\hat{\pi}(X))K_{\mathcal{G}_3,n}\frac{1}{4}\Big(\frac{1}{4n}K_{\mathcal{G}_3,n}+\lambda_{\mathcal{G}_3}I_n\Big)^{-1}diag(A\hat{\pi}(X))K_{\Omega,n}+n^2\lambda_{\Omega}K_{\Omega,n}\Bigg)^{\dagger}\\
    &~~~~K_{\Omega,n}diag(A\hat{\pi}(X))\frac{1}{4}\Big(\frac{1}{4n}K_{\mathcal{G}_3,n}+\lambda_{\mathcal{G}_3}I_n\Big)^{-1}C(A\hat{\pi}(X))_n.
\end{align*}
\end{proof}

\begin{proof}[Proof of Proposition \ref{prop:true_multi}]
   \begin{enumerate}[label=(\alph*)]
   \item For all $g_1\in L^2(P_{0,X})$, we have
   \begin{align*}
       E\left[\left\{A\pi(X)-1\right\}g_1(X)\right]
=E\left[\left\{E[A\mid X]\pi(X)-1\right\}g_1(X)\right]
=0.
   \end{align*}
   \item For all $g_2\in L^2(P_{0,M_j,X})$, we have
\begin{align*}
    &E\Big[A\pi(X)\rho(M_j,X)g_2(M_j,X)\Big]\\
    &=E\Big[\frac{p(A=1\mid M_j,X)}{p(A=1\mid X)}\frac{p(M_j\mid A=0,X)}{p(M_j\mid A=1,X)}g_2(M_j,X)\Big]\\
    &=E\Big[\frac{p(A=1,M_j\mid X)}{p(M_j\mid X)}\frac{p(M_j\mid A=0,X)}{p(M_j ,A=1\mid X)}g_2(M_j,X)\Big]\\
    &=E\Big[\frac{p(M_j\mid A=0,X)}{p(M_j\mid X)}g_2(M_j,X)\Big]\\
    &=E\Big[\frac{p(M_j,A=0\mid X)}{p(A=0\mid X)p(M_j\mid X)}g_2(M_j,X)\Big]\\
    &=E\Big[\frac{p(A=0\mid M_j,X)}{1-f(A=1\mid X)}g_2(M_j,X)\Big]\\
    &=E\Big[\frac{1-A}{1-\pi( X)}g_2(M_j,X)\Big].
\end{align*}
       \item For all $g_3\in L^2(P_{0,M_j,M_{-j},X})$, we have
\begin{align*}
    &E\Big[A\pi(X)\omega(M_j,M_{-j},X)g_3(M_j,M_{-j},X)\Big]\\
    &=E\Big[\frac{p(A=1\mid M_j,M_{-j},X)}{p(A=1\mid X)}\frac{p(M_j\mid A=1,X)p(M_{-j}\mid A=1,X)}{p(M_j,M_{-j}\mid A=1,X)}g_3(M_j,M_{-j},X)\Big]\\
    &=E\Big[\frac{p(X)p(M_j\mid A=1,X)p(M_{-j}\mid A=1,X)}{p(M_j,M_{-j},X)}g_3(M_j,M_{-j},X)\Big]\\
    &=\iiint\frac{p(x)p(m_j\mid A=1,x)p(m_{-j}\mid A=1,x)}{p(m_j,m_{-j},x)}g_3(m_j,m_{-j},x)\cdot p(m_j,m_{-j},x)dm_jdm_{-j}dx\\
    &=\int\Big(\iint p(m_j\mid A=1,x)p(m_{-j}\mid A=1,x)g_3(m_j,m_{-j},x)dm_jdm_{-j}\Big)p(x)dx\\
    &=\int\Big(\int f(m_{-j},x)p(m_{-j}\mid A=1,x)dm_{-j}\Big)p(x)dx\\
    &=E\Big[\E[f(M_{-j},X)\mid A=1,X]\Big]\\
    &=E\Big[A\pi(X)f(M_{-j},X)\Big],
\end{align*}
where $f(m_{-j},x)=E[g_3(M_j,m_{-j},X)\mid A=1,X=x]$.
   \end{enumerate}
\end{proof}

\begin{proof}[Proof of Proposition \ref{prop:MR_multi}]
\begin{align*}
    &E[\hat\Delta^{M_j,\text{2S}}]-\Delta^{M_j}\\
    &=E\Big[\Big\{A\hat{\pi}(X)-1\Big\}E\big[\hat{R}(M_j,M_{-j},X)\big\{\mu(M_j,M_{-j},X)-\hat{\mu}(M_j,M_{-j},X)\big\}\mid A=1,X\big]\Big]\\
    &~~~~~~~~~+E\Big[\Big\{A\hat{\pi}(X)-1\Big\}
    \Big\{E[\hat{\eta}_2(M_{-j},X)\mid A=1,X]-\hat{\gamma}_1(X)\Big\}\Big]\\
    &~~~~~~~~~+E\Big[\Big\{\frac{(1-A)\hat{\pi}(X)}{\hat{\pi}(X)-1}-1\Big\}\Big\{\hat{\gamma}_{1,0}(X)-E[\hat{\eta}_1(M_j,X)\mid A=0,X]\Big\}\Big]\\
    &~~~~~~~~~+{E\Big[E\big[\{\hat{R}(M_j,M_{-j},X)-R(M_j,M_{-j},X)\}\big\{\mu(M_j,M_{-j},X)-\hat{\mu}(M_j,M_{-j},X)\big\}\mid A=1,X\big]\Big]}\\
&~~~~~~~~~+E\Big[E\big[\{\eta_1^\ast(M_j,X)-\hat{\eta}_1(M_j,X)\}\big\{\rho(M_j,X)-\hat{\rho}(M_j,X)\big\}\mid A=1,X\big]\Big]\\
&=E\Big[\underbrace{\Big\{p(A=1\mid X)\hat{\pi}(X)-1\Big\}}_{B_{11}}\underbrace{E\big[\hat{R}(M_j,M_{-j},X)\big\{\mu(M_j,M_{-j},X)-\hat{\mu}(M_j,M_{-j},X)\big\}\mid A=1,X\big]}_{B_{12}}\Big]\\
    &~~~~~~~~~+E\Big[\underbrace{\Big\{p(A=1\mid X)\hat{\pi}(X)-1\Big\}}_{B_{21}}
    \underbrace{\Big\{E[\hat{\eta}_2(M_{-j},X)\mid A=1,X]-\hat{\gamma}_1(X)\Big\}}_{B_{22}}\Big]\\
    &~~~~~~~~~+E\Big[\underbrace{\Big\{\frac{p(A=0\mid X)\hat{\pi}(X)}{\hat{\pi}(X)-1}-1\Big\}}_{B_{31}}\underbrace{\Big\{\hat{\gamma}_{1,0}(X)-E[\hat{\eta}_1(M_j,X)\mid A=0,X]\Big\}}_{B_{32}}\Big]\\
    &~~~~~~~~~+E\Big[E\big[\underbrace{\{\big(1-\hat\rho(M_j,X)\big)\hat\omega(M_j,M_{-j},X))-\big(1-\rho(M_j,X)\big)\omega(M_j,M_{-j},X)\}}_{B_{41}}\\
    &~~~~~~~~~~~~~~~~~~~~~\times\underbrace{\big\{\mu(M_j,M_{-j},X)-\hat{\mu}(M_j,M_{-j},X)\big\}}_{B_{42}}\mid A=1,X\big]\Big]\\
&~~~~~~~~~+E\Big[E\big[\underbrace{\{E_{M_{-j}\mid A=1,X}[\hat\mu(M_j,M_{-j},X)-\mu(M_j,M_{-j},X)\mid A=1,X]+\eta_1(M_j,X)-\hat{\eta}_1(M_j,X)\}}_{B_{51}}\\
&~~~~~~~~~~~~~~~~~~~~~\times\underbrace{\big\{\rho(M_j,X)-\hat{\rho}(M_j,X)\big\}}_{B_{52}}\mid A=1,X\big]\Big].
\end{align*}

\begin{enumerate}[label=(\roman*)]
\item If $\hat\pi=\pi,\hat\rho=\rho$, and $\hat\omega=\omega$, we have $B_{11}=B_{21}=B_{31}=B_{41}=B_{52}=0$. Therefore,
\begin{align*}
    E[\hat\Delta^{M_j,\text{2S}}]-\Delta^{M_j}=0,
\end{align*}
and $\hat\Delta^{M_j,\text{2S}}$ is unbiased.
 
\item If $\hat\pi=\pi,\hat\rho=\rho$, and $\hat\mu=\mu$, we have $B_{11}=B_{21}=B_{31}=B_{42}=B_{52}=0$. Therefore,
\begin{align*}
    E[\hat\Delta^{M_j,\text{2S}}]-\Delta^{M_j}=0,
\end{align*}
and $\hat\Delta^{M_j,\text{2S}}$ is unbiased.
\item If $\hat\pi=\pi,\hat\mu=\mu$, and $\hat\eta_1=\eta_1$, we have $B_{11}=B_{21}=B_{31}=B_{42}=B_{51}=0$. Therefore,
\begin{align*}
    E[\hat\Delta^{M_j,\text{2S}}]-\Delta^{M_j}=0,
\end{align*}
and $\hat\Delta^{M_j,\text{2S}}$ is unbiased.
\item If $\hat\mu=\mu,\hat\rho=\rho,\hat\omega=\omega$, and $\hat\eta_1=\eta_1$, we have $B_{12}=B_{42}=B_{52}=0$. We also have,
\begin{align*}
    B_{22}&=E[\hat{\eta}_2(M_{-j},X)\mid A=1,X]-\hat{\gamma}_1(X)\\
    &=E_{M_{-j}\mid A=1,X}[E_{M_j\mid A=1,X}[\hat\mu(M_j,M_{-j},X)\big(1-\hat\rho(M_j,X)\big)\mid A=1,X]\mid A=1,X]\\
    &~~~~~-E[\hat\eta_1(M_j,X)(1-\hat\rho(M_j,X))\mid A=1,X]\\
    &=E_{M_{-j}\mid A=1,X}[E_{M_j\mid A=1,X}[\hat\mu(M_j,M_{-j},X)\big(1-\hat\rho(M_j,X)\big)\mid A=1,X]\mid A=1,X]\\
    &~~~~~-E_{M_{-j}\mid A=1,X}[E_{M_j\mid A=1,X}[\mu(M_j,M_{-j},X)\big(1-\hat\rho(M_j,X)\big)\mid A=1,X]\mid A=1,X]\\
    &~~~~~+E_{M_{-j}\mid A=1,X}[E_{M_j\mid A=1,X}[\mu(M_j,M_{-j},X)\big(1-\hat\rho(M_j,X)\big)\mid A=1,X]\mid A=1,X]\\
    &~~~~~-E_{M_{-j}\mid A=1,X}[E_{M_j\mid A=1,X}[\mu(M_j,M_{-j},X)\big(1-\rho(M_j,X)\big)\mid A=1,X]\mid A=1,X]\\
    &~~~~~+E_{M_{-j}\mid A=1,X}[E_{M_j\mid A=1,X}[\mu(M_j,M_{-j},X)\big(1-\rho(M_j,X)\big)\mid A=1,X]\mid A=1,X]\\
    &~~~~~-E[\hat\eta_1(M_j,X)(1-\hat\rho(M_j,X))\mid A=1,X]\|_2\\
    &=E_{M_{-j}\mid A=1,X}[E_{M_j\mid A=1,X}[\big(\hat\mu(M_j,M_{-j},X)-\mu(M_j,M_{-j},X)\big)\big(1-\hat\rho(M_j,X)\big)\mid A=1,X]\mid A=1,X]\\
    &~~~~~+E_{M_{-j}\mid A=1,X}[E_{M_j\mid A=1,X}[\mu(M_j,M_{-j},X)\big(\rho(M_j,X)-\hat\rho(M_j,X)\big)\mid A=1,X]\mid A=1,X]\\
    &~~~~~+E[\eta_1(M_j,X)(1-\rho(M_j,X))\mid A=1,X]-E[\hat\eta_1(M_j,X)(1-\hat\rho(M_j,X))\mid A=1,X]\\
    &=E_{M_{-j}\mid A=1,X}[E_{M_j\mid A=1,X}[\big(\underbrace{\hat\mu(M_j,M_{-j},X)-\mu(M_j,M_{-j},X)}_{=0}\big)\big(1-\hat\rho(M_j,X)\big)\mid A=1,X]\mid A=1,X]\\
    &~~~~~+E_{M_{-j}\mid A=1,X}[E_{M_j\mid A=1,X}[\mu(M_j,M_{-j},X)\big(\underbrace{\rho(M_j,X)-\hat\rho(M_j,X)}_{=0}\big)\mid A=1,X]\mid A=1,X]\\
    &~~~~~+E[\eta_1(M_j,X)(\underbrace{\hat\rho(M_j,X)-\rho(M_j,X)}_{=0})\mid A=1,X]\\
    &~~~~~+E[\big(\underbrace{\eta_1(M_j,X)-\hat\eta_1(M_j,X)}_{=0}\big)(1-\hat\rho(M_j,X))\mid A=1,X]\\
    &=0,
\end{align*}
and
\begin{align*}
    B_{32}&=\hat{\gamma}_{1,0}(X)-E[\hat{\eta}_1(M_j,X)\mid A=0,X]\\
    &=E[\hat\eta_1(M_j,X)\hat\rho(M_j,X)\mid A=1,X]-E[\hat{\eta}_1(M_j,X)\mid A=0,X]\\
    &=E[\hat\eta_1(M_j,X)\hat\rho(M_j,X)\mid A=1,X]-E[\hat\eta_1(M_j,X)\rho(M_j,X)\mid A=1,X]\\
    &=E[\hat\eta_1(M_j,X)\big(\hat\rho(M_j,X)-\rho(M_j,X)\big)\mid A=1,X]\\
    &=0.
\end{align*}
Therefore,
\begin{align*}
    E[\hat\Delta^{M_j,\text{2S}}]-\Delta^{M_j}=0,
\end{align*}
and $\hat\Delta^{M_j,\text{2S}}$ is unbiased.
\end{enumerate}
\end{proof}

\begin{proof}[Proof of Theorem \ref{thm:main}]
\begingroup
\allowdisplaybreaks	
\begin{align*}
&E[\hat\psi^\text{2S}]-\psi_0\\
&=E\left[ \hat\mu_2(X)+(1-A)\hat\pi_1(X)\{\hat\mu_1(M,X)-\hat\mu_2(X)\}+A\hat\pi_2(M,X)\{Y-\hat\mu_1(M,X)\} \right]-E[\mu_2(X)]\\
&=E\left[ \hat\mu_2(X)-\mu_2(X)+(1-A)\hat\pi_1(X)\{\mu_1(M,X)-\hat\mu_2(X)\}\right]\\
&~~~+E\left[A\hat\pi_2(M,X)\{Y-\hat\mu_1(M,X)\}-(1-A)\hat\pi_1(X)\{\mu_1(M,X)-\hat\mu_1(M,X)\}\right]\\
&=E\left[ (1-A)\hat\pi_1(X)\{\mu_1(M,X)-\hat\mu_2(X)\}-\{\mu_2(X)-\hat\mu_2(X)\}\right]\\
&~~~+E\left[A\hat\pi_2(M,X)\{Y-\hat\mu_1(M,X)\}-(1-A)\hat\pi_1(X)\{\mu_1(M,X)-\hat\mu_1(M,X)\} \right]\\
&=E\left[ E\left[(1-A)\hat\pi_1(X)\{\mu_1(M,X)-\hat\mu_2(X)\}\bigg|X\right]-\{\mu_2(X)-\hat\mu_2(X)\}\right]\\
&~~~+E\left[E\left[A\hat\pi_2(M,X)\{Y-\hat\mu_1(M,X)\}\bigg|M,X\right]-(1-A)\hat\pi_1(X)\{\mu_1(M,X)-\hat\mu_1(M,X)\} \right]\\
&=E\left[ E[(1-A)\mu_1(M,X)\mid X]\hat\pi_1(X)-\frac{\hat\pi_1(X)\hat\mu_2(X)}{\pi_1(X)}-\{\mu_2(X)-\hat\mu_2(X)\}\right]\\
&~~~+E\Big[\hat\pi_2(M,X)\big(E[AY\mid M,X]-E[A\mid M,X]\hat\mu_1(M,X)\big)-(1-A)\hat\pi_1(X)\{\mu_1(M,X)-\hat\mu_1(M,X)\}\Big]\\
&=E\left[ E[(1-A)\mid X]\{\mu_2(X)-\hat\mu_2(X)\}\hat\pi_1(X)-\{\mu_2(X)-\hat\mu_2(X)\}\right]\\
&~~~+E\left[\hat\pi_2(M,X)E[A\mid M,X]\{\mu_1(M,X)-\hat\mu_1(M,X)\}-(1-A)\hat\pi_1(X)\{\mu_1(M,X)-\hat\mu_1(M,X)\} \right]\\
&=E\left[ \left\{(1-A)\hat\pi_1(X)-1\right\}\{\mu_2(X)-\hat\mu_2(X)\}\right]\\
&~~~+E\left[\left\{A\hat\pi_2(M,X)-(1-A)\hat\pi_1(X)\right\}\{\mu_1(M,X)-\hat\mu_1(M,X)\} \right].
\end{align*}
\endgroup
\end{proof}

\begin{proof}[Proof of Theorem \ref{thm:CAN}]

Assume the size of each fold of the data is $m$, that is, $n=mL$. Let $(\hat{\pi}_{1,\ell},\hat{\pi}_{2,\ell},\hat{\mu}_{1,\ell},\hat{\mu}_{2,\ell})$ be the estimators of $(\pi_1,\pi_2,\mu_1,\mu_2)$ obtained on subsamples $I_\ell^c$. Then $\hat\psi^\text{2S}_\ell$ can be written as
\begin{align*}
    \hat\psi^\text{2S}_\ell=E_m\left[ \phi(O;\hat{\pi}_{1,\ell},\hat{\pi}_{2,\ell},\hat{\mu}_{1,\ell},\hat{\mu}_{2,\ell}) \right].
\end{align*}
Therefore, we have
\begin{align*}
    \sqrt{n}(\hat\psi^\text{2S}-\psi_0)&=\sqrt{n}\Big\{\frac{1}{L}\sum_{\ell=1}^LE_m\left[ \phi(O;\hat{\pi}_{1,\ell},\hat{\pi}_{2,\ell},\hat{\mu}_{1,\ell},\hat{\mu}_{2,\ell})\right]-\psi_0\Big\}\\
    &=\frac{\sqrt{n}}{L\sqrt{m}}\sum_{\ell=1}^L\Bigg\{\left\{\mathbb{G}_m\left[ \phi(O;\hat{\pi}_{1,\ell},\hat{\pi}_{2,\ell},\hat{\mu}_{1,\ell},\hat{\mu}_{2,\ell})\right]-\mathbb{G}_m\left[ \phi(O;\pi_1,\pi_2,\mu_1,\mu_2)\right]\right\}\\
    &\quad\quad\quad\quad\quad\quad\quad +\mathbb{G}_m\left[ \phi(O;\pi_1,\pi_2,\mu_1,\mu_2)\right]\\
    &\quad\quad\quad\quad\quad\quad\quad +\sqrt{m}\Big\{E\left[ \phi(O;\hat{\pi}_{1,\ell},\hat{\pi}_{2,\ell},\hat{\mu}_{1,\ell},\hat{\mu}_{2,\ell})\right]-\psi_0\Big\}\Bigg\},
\end{align*}
where for each function $f(O)$,
\begin{align*}
    \mathbb{G}_n[f]=\sqrt{n}\{E_n[f]-E[f]\},
\end{align*}
represents the empirical process. Recalling that $n=mL$, we have
\begin{align*}
    \sqrt{n}(\hat\psi^\text{2S}-\psi_0)&=\frac{1}{\sqrt{L}}\sum_{\ell=1}^L\left\{\mathbb{G}_m\left[ \phi(O;\hat{\pi}_{1,\ell},\hat{\pi}_{2,\ell},\hat{\mu}_{1,\ell},\hat{\mu}_{2,\ell})\right]-\mathbb{G}_m\left[ \phi(O;\pi_1,\pi_2,\mu_1,\mu_2)\right]\right\} &(T_1)\\
    &\quad +\frac{1}{\sqrt{L}}\sum_{\ell=1}^L\mathbb{G}_m\left[ \phi(O;\pi_1,\pi_2,\mu_1,\mu_2)\right] &(T_2)\\
    &\quad +\frac{1}{\sqrt{L}}\sum_{\ell=1}^L\sqrt{m}\Big\{E\left[ \phi(O;\hat{\pi}_{1,\ell},\hat{\pi}_{2,\ell},\hat{\mu}_{1,\ell},\hat{\mu}_{2,\ell})\right]-\psi_0\Big\}. &(T_3)
\end{align*}
We will show that $(T_2)$ provides us with the term $\sqrt{n}E_n\left[ \phi(O;\pi_1,\pi_2,\mu_1,\mu_2) \right]$, and under the assumptions of the theorem, $(T_1)$ and $(T_3)$ are $o_p(1)$.

\medskip
\noindent
\textbf{Analysis of $T_1$}

Define
\begin{align*}
    A_m^\ell :=\left\{\mathbb{G}_m\left[ \phi(O;\hat{\pi}_{1,\ell},\hat{\pi}_{2,\ell},\hat{\mu}_{1,\ell},\hat{\mu}_{2,\ell})\right]-\mathbb{G}_m\left[ \phi(O;\pi_1,\pi_2,\mu_1,\mu_2)\right]\right\},
\end{align*}
where the empirical processes are evaluated on $I_\ell$. We note that
\allowdisplaybreaks
\begin{align*}
    &var\left(A_m^\ell\mid I_\ell^c\right)\\
    &=m\cdot var\left(E_m\left[ \phi(O;\hat{\pi}_{1,\ell},\hat{\pi}_{2,\ell},\hat{\mu}_{1,\ell},\hat{\mu}_{2,\ell})\right]-E_m\left[ \phi(O;\pi_1,\pi_2,\mu_1,\mu_2)\right]\mid I_\ell^c\right)\\
    &=var\big(\phi(O;\hat{\pi}_{1,\ell},\hat{\pi}_{2,\ell},\hat{\mu}_{1,\ell},\hat{\mu}_{2,\ell})-\phi(O;\pi_1,\pi_2,\mu_1,\mu_2)\mid I_\ell^c\big)\\
    &\leq E\big[\{\phi(O;\hat{\pi}_{1,\ell},\hat{\pi}_{2,\ell},\hat{\mu}_{1,\ell},\hat{\mu}_{2,\ell})-\phi(O;\pi_1,\pi_2,\mu_1,\mu_2)\}^2\mid I_\ell^c\big]\\
    &=E\Big[\Big\{\hat\mu_{2,\ell}(X)+(1-A)\hat\pi_{1,\ell}(X)\{\hat\mu_{1,\ell}(M,X)-\hat\mu_{2,\ell}(X)\}+A\hat\pi_{2,\ell}(M,X)\{Y-\hat\mu_{1,\ell}(M,X)\}\\
    &\quad\quad -\mu_2(X)-(1-A)\pi_1(X)\{\mu_1(M,X)-\mu_2(X)\}-A\pi_2(M,X)\{Y-\mu_1(M,X)\}\Big\}^2\Big| I_\ell^c\Big]\\
    &=E\Big[\Big\{\hat\mu_{2,\ell}(X)\left(1-(1-A)\hat\pi_{1,\ell}(X)\right)+\hat\mu_{1,\ell}(M,X)\left((1-A)\hat\pi_{1,\ell}(X)-A\hat\pi_{2,\ell}(M,X)\right)+AY\hat\pi_{2,\ell}(M,X)\\
    &\quad\quad -\mu_2(X)\left(1-(1-A)\pi_1(X)\right)-\mu_1(M,X)\left((1-A)\pi_1(X)-A\pi_2(M,X)\right)-AY\pi_2(M,X)\Big\}^2\Big| I_\ell^c\Big]\\
    &\leq5\Big(\|\hat\mu_{2,\ell}(X)-\mu_2(X)\|_2^2+\|(1-A)(\hat\mu_{2,\ell}(X)\hat\pi_{1,\ell}(X)-\mu_2(X)\pi_1(X))\|_2^2\\
    &\quad\quad+\|(1-A)(\hat\mu_{1,\ell}(M,X)\hat\pi_{1,\ell}(X)-\mu_1(M,X)\pi_1(X))\|_2^2\\
    &\quad\quad+\|A(\hat\mu_{1,\ell}(M,X)\hat\pi_{2,\ell}(M,X)-\mu_1(M,X)\pi_2(M,X))\|_2^2+\|AY(\hat\pi_{2,\ell}(M,X)-\pi_{2}(M,X))\|_2^2\Big)\\
    &\leq5\Big(\|\hat\mu_{2,\ell}(X)-\mu_2(X)\|_2^2+\|\hat\mu_{2,\ell}(X)\hat\pi_{1,\ell}(X)-\mu_2(X)\pi_1(X)\|_2^2\\
    &\quad\quad+\|\hat\mu_{1,\ell}(M,X)\hat\pi_{1,\ell}(X)-\mu_1(M,X)\pi_1(X)\|_2^2\\
    &\quad\quad+\|\hat\mu_{1,\ell}(M,X)\hat\pi_{2,\ell}(M,X)-\mu_1(M,X)\pi_2(M,X)\|_2^2+\|Y(\hat\pi_{2,\ell}(M,X)-\pi_{2}(M,X))\|_2^2\Big)\\
    &\leq10\Big(\|\hat\mu_{2,\ell}(X)-\mu_2(X)\|_2^2+\|\hat\mu_{2,\ell}(X)(\hat\pi_{1,\ell}(X)-\pi_{1}(X))\|_2^2+\|\pi_{1}(X)(\hat\mu_{2,\ell}(X)-\mu_2(X))\|_2^2\\
    &\quad\quad+\|\hat\mu_{1,\ell}(M,X)(\hat\pi_{1,\ell}(X)-\pi_{1}(X))\|_2^2+\|\pi_1(X)(\hat\mu_{1,\ell}(M,X)-\mu_1(M,X))\|_2^2\\
    &\quad\quad+\|\hat\mu_{1,\ell}(M,X)(\hat\pi_{2,\ell}(M,X)-\pi_2(M,X))\|_2^2+\|\pi_2(M,X)(\hat\mu_{1,\ell}(M,X)-\mu_1(M,X))\|_2^2\\
    &\quad\quad+\|Y(\hat\pi_{2,\ell}(M,X)-\pi_{2}(M,X))\|_2^2\Big)\\
    &\leq10\Big(\|\hat\mu_{2,\ell}(X)-\mu_2(X)\|_2^2+\|\hat\mu_{2,\ell}^2(X)\|_2\|(\hat\pi_{1,\ell}(X)-\pi_{1}(X))^2\|_2+\|\pi_{1}(X)(\hat\mu_{2,\ell}(X)-\mu_2(X))\|_2^2\\
    &\quad\quad+\|\hat\mu_{1,\ell}^2(M,X)\|_2\|(\hat\pi_{1,\ell}(X)-\pi_{1}(X))^2\|_2+\|\pi_1(X)(\hat\mu_{1,\ell}(M,X)-\mu_1(M,X))\|_2^2\\
    &\quad\quad+\|\hat\mu_{1,\ell}^2(M,X)\|_2\|(\hat\pi_{2,\ell}(M,X)-\pi_2(M,X))^2\|_2+\|\pi_2(M,X)(\hat\mu_{1,\ell}(M,X)-\mu_1(M,X))\|_2^2\\
    &\quad\quad+\|Y^2\|_2\|(\hat\pi_{2,\ell}(M,X)-\pi_{2}(M,X))^2\|_2\Big).
    \end{align*}

    By Assumption \ref{assm:CAN} (a) and (b), $\|\mu_1(M,X)-\hat\mu_{1,\ell}(M,X)\|_2$, $\|\hat\mu_{2,\ell}(X)-\mu_2(X)\|_2$, $\|\hat\mu_{1,\ell}^2(M,X)-\mu_{1}^2(M,X)\|_2$, $\|\hat\mu_{2,\ell}^2(X)-\mu_{2}^2(X)\|_2$, $\|(\hat\pi_{1,\ell}(X)-\pi_1(X))^2\|_2$, and $\|(\hat\pi_{2,\ell}(M,X)-\pi_2(M,X))^2\|_2$ converge to zero in probability. By Assumption \ref{ass:stg_pos},
    \begin{align*}
        \|\pi_{1}(X)(\hat\mu_{2,\ell}(X)-\mu_2(X))\|_2&\leq\frac{1}{\epsilon}\|\mu_2(X)-\hat\mu_{2,\ell}(X)\|_2,\\
        \|\pi_{1}(X)(\hat\mu_{1,\ell}(M,X)-\mu_1(M,X))\|_2&\leq\frac{1}{\epsilon}\|\mu_1(M,X)-\hat\mu_{1,\ell}(M,X)\|_2,\\
        \|\pi_{2}(M,X)(\hat\mu_{1,\ell}(M,X)-\mu_1(M,X))\|_2&=\|\frac{p(A=0\mid M,X)}{p(A=1\mid M,X)}\pi_1(X)(\mu_1(M,X)-\hat\mu_{1,\ell}(M,X))\|_2\\
        &\leq\frac{1}{\epsilon^2}\|\mu_1(M,X)-\hat\mu_{1,\ell}(M,X)\|_2.
    \end{align*}
    Therefore, $\|\pi_{1}(X)(\hat\mu_{2,\ell}(X)-\mu_2(X))\|_2$, $\|\pi_{1}(X)(\hat\mu_{1,\ell}(M,X)-\mu_1(M,X))\|_2$, and $\|\pi_{2}(M,X)(\hat\mu_{1,\ell}(M,X)-\mu_1(M,X))\|_2$ also converge to zero in probability.
    
    By Assumption \ref{assm:CAN} (c), $\|\mu_{1}^2(M,X)\|_2$, $\|\mu_{2}^2(X)\|_2$, and $\|Y^2\|_2$ are bounded by some constants, which imply that
\begin{align}
\|\hat\mu_{1,\ell}^2(M,X)\|_2&\leq\|\hat\mu_{1,\ell}^2(M,X)-\mu_{1}^2(M,X)\|_2+\|\mu_{1}^2(M,X)\|_2\leq o_p(1)+C=O_p(1), \label{ineq_mu1hat}\\
    \|\hat\mu_{2,\ell}^2(X)\|_2&\leq\|\hat\mu_{2,\ell}^2(X)-\mu_{2}^2(X)\|_2+\|\mu_{2}^2(X)\|_2\leq o_p(1)+C=O_p(1). \label{ineq_mu2hat}
\end{align}

Therefore,
\begin{align*}
    &var\left(A_m^\ell\mid I_\ell^c\right)=o_p(1).
\end{align*}
 Also note that,
 \begin{align*}
     &E\left[A_m^\ell\mid I_\ell^c\right]\\
     &=E\Big[\hat\mu_{2,\ell}(X)+(1-A)\hat\pi_{1,\ell}(X)\{\hat\mu_{1,\ell}(M,X)-\hat\mu_{2,\ell}(X)\}+A\hat\pi_{2,\ell}(M,X)\{Y-\hat\mu_{1,\ell}(M,X)\}\\
    &\quad\quad -\mu_2(X)-(1-A)\pi_1(X)\{\mu_1(M,X)-\mu_2(X)\}-A\pi_2(M,X)\{Y-\mu_1(M,X)\}\Big| I_\ell^c\Big]\\
    &=E\Big[\hat\mu_{2,\ell}(X)\left(1-(1-A)\hat\pi_{1,\ell}(X)\right)+\hat\mu_{1,\ell}(M,X)\left((1-A)\hat\pi_{1,\ell}(X)-A\hat\pi_{2,\ell}(M,X)\right)+AY\hat\pi_{2,\ell}(M,X)\\
    &\quad\quad -\mu_2(X)\left(1-(1-A)\pi_1(X)\right)-\mu_1(M,X)\left((1-A)\pi_1(X)-A\pi_2(M,X)\right)-AY\pi_2(M,X)\Big| I_\ell^c\Big]\\
    &=E\Big[\hat\mu_{2,\ell}(X)\left(1-p(A=0\mid X)\hat\pi_{1,\ell}(X)\right)\\
    &\quad\quad +\hat\mu_{1,\ell}(M,X)\left(p(A=0\mid M, X)\hat\pi_{1,\ell}(X)-p(A=1\mid M, X)\hat\pi_{2,\ell}(M,X)\right)\\
    &\quad\quad -\mu_2(X)\underbrace{\left(1-p(A=0\mid X)\pi_1(X)\right)}_{=0}\\
    &\quad\quad-\mu_1(M,X)\underbrace{\left(p(A=0\mid M, X)\pi_1(X)-p(A=1\mid M, X)\pi_2(M,X)\right)}_{=0}\\
    &\quad\quad +\E[AY\mid M,X]\hat\pi_{2,\ell}(M,X)-\E[AY\mid M,X]\pi_2(M,X)\Big| I_\ell^c\Big]\\
    &=E\Big[\hat\mu_{2,\ell}(X)\left(1-p(A=0\mid X)\hat\pi_{1,\ell}(X)\right)\\
    &\quad\quad+\hat\mu_{1,\ell}(M,X)\left(p(A=0\mid M, X)\hat\pi_{1,\ell}(X)-p(A=1\mid M, X)\hat\pi_{2,\ell}(M,X)\right)\\
    &\quad\quad +p(A=1\mid M,X)\mu_1(M,X)\hat\pi_{2,\ell}(M,X)-p(A=1\mid M,X)\mu_1(M,X)\pi_2(M,X)\Big| I_\ell^c\Big].
    \end{align*}

    By Jensen's inequality and Cauchy-Schwarz inequality, 
    \begin{align*}
     &E\left[A_m^\ell\mid I_\ell^c\right]\\
    &\leq\Big\|\hat\mu_{2,\ell}(X)\left(1-p(A=0\mid X)\hat\pi_{1,\ell}(X)\right)\\
    &\quad\quad+\hat\mu_{1,\ell}(M,X)\left(p(A=0\mid M, X)\hat\pi_{1,\ell}(X)-p(A=1\mid M, X)\hat\pi_{2,\ell}(M,X)\right)\\
    &\quad\quad +p(A=1\mid M,X)\mu_1(M,X)\hat\pi_{2,\ell}(M,X)-p(A=1\mid M,X)\mu_1(M,X)\pi_2(M,X)\Big\|_2\\
    &=\Big\|\hat\mu_{2,\ell}(X)\left(1-p(A=0\mid X)\hat\pi_{1,\ell}(X)\right)\\
    &\quad\quad+\hat\mu_{1,\ell}(M,X)(p(A=0\mid M, X)\hat\pi_{1,\ell}(X)-p(A=0\mid M, X)\pi_{1}(X)\\
    &\quad\quad\quad+p(A=1\mid M, X)\pi_{2}(M,X)-p(A=1\mid M, X)\hat\pi_{2,\ell}(M,X))\\
    &\quad\quad +p(A=1\mid M,X)\mu_1(M,X)(\hat\pi_{2,\ell}(M,X)-\pi_2(M,X))\Big\|_2\\
    &\leq\Big\|\hat\mu_{2,\ell}(X)\left(1-p(A=0\mid X)\hat\pi_{1,\ell}(X)\right)\Big\|_2\\
    &\quad\quad+\Big\|\hat\mu_{1,\ell}(M,X)p(A=0\mid M, X)\left(\hat\pi_{1,\ell}(X)-\pi_{1,\ell}(X)\right)\Big\|_2\\
    &\quad\quad+\Big\|\mu_{1,\ell}(M,X)p(A=1\mid M, X)(\pi_{2}(M,X)-\hat\pi_{2,\ell}(M,X))\Big\|_2\\
    &\quad\quad +\Big\|p(A=1\mid M,X)\mu_1(M,X)(\hat\pi_{2,\ell}(M,X)-\pi_2(M,X))\Big\|_2\\
    &\leq\Big\|\hat\mu_{2,\ell}^2(X)\Big\|_2^{1/2}\Big\|(1-p(A=0\mid X)\hat\pi_{1,\ell}(X))^2\Big\|_2^{1/2}\\
    &\quad\quad +\Big\|(p(A=0\mid M, X)\hat\mu_{1,\ell}(M,X))^2\Big\|_2^{1/2}\Big\|\big(\hat\pi_{1,\ell}(X)-\pi_{1}(X)\big)^2\Big\|_2^{1/2}\\
    &\quad\quad +\Big\|(p(A=1\mid M, X)\hat\mu_{1,\ell}(M,X))^2\Big\|_2^{1/2}\Big\|\big(\hat\pi_{2,\ell}(M,X)-\pi_{2}(M,X)\big)^2\Big\|_2^{1/2}\\
    &\quad\quad +\Big\|(p(A=1\mid M,X)\mu_{1}(M,X))^2\Big\|_2^{1/2}\Big\|(\hat\pi_{2,\ell}(M,X)-\pi_2(M,X))^2\Big\|_2^{1/2}\\
    &\leq\Big\|\hat\mu_{2,\ell}^2(X)\Big\|_2^{1/2}\Big\|(1-p(A=0\mid X)\hat\pi_{1,\ell}(X))^2\Big\|_2^{1/2}+\Big\|\hat\mu_{1,\ell}^2(M,X)\Big\|_2^{1/2}\Big\|\big(\hat\pi_{1,\ell}(X)-\pi_{1}(X)\big)^2\Big\|_2^{1/2}\\
    &\quad\quad +\Big\|\hat\mu_{1,\ell}^2(M,X)\Big\|_2^{1/2}\Big\|\big(\hat\pi_{2,\ell}(M,X)-\pi_{2}(M,X)\big)^2\Big\|_2^{1/2}+\Big\|\mu_{1}^2(M,X)\Big\|_2^{1/2}\Big\|(\hat\pi_{2,\ell}(M,X)-\pi_2(M,X))^2\Big\|_2^{1/2}.
 \end{align*}

Also note that,
\begin{align*}
    \Big\|(1-p(A=0\mid X)\hat\pi_{1,\ell}(X))^2\Big\|_2&=\Big\|(p(A=0\mid X))^2(\pi_{1}(X)-\hat\pi_{1,\ell}(X))^2\Big\|_2\\
    &\leq\Big\|(\pi_{1}(X)-\hat\pi_{1,\ell}(X))^2\Big\|_2.
\end{align*}
Hence, by Assumption \ref{assm:CAN} (b), $\Big\|(1-p(A=0\mid X)\hat\pi_{1,\ell}(X))^2\Big\|_2$ converges to zero in probability.

Similar to inequalities \eqref{ineq_mu1hat} and \eqref{ineq_mu2hat}, $\Big\|\hat\mu_{1,\ell}^2(M,X)\Big\|_2^{1/2}$ and $\Big\|\hat\mu_{2,\ell}^2(X)\Big\|_2^{1/2}$ are $O_p(1)$.  Therefore, by Assumption \ref{assm:CAN} (a), (b) and (c),
 \begin{align*}
     &E\left[A_m^\ell\mid I_\ell^c\right]=o_p(1).
 \end{align*}
 
 Therefore,
\begin{align*}
    E\left[(A_m^\ell)^2\mid I_\ell^c\right]=var\left(A_m^\ell\mid I_\ell^c\right)+\Big(E\left[A_m^\ell\mid I_\ell^c\right]\Big)^2\stackrel{p.}{\rightarrow}0,\quad\text{ as }m\rightarrow\infty.
\end{align*}

Then by conditional Chebyshev inequality, for all $\delta>0$, we have,
\begin{align*}
    P(|A_m^\ell|>\delta\mid I_\ell^c)\stackrel{p.}{\rightarrow}0,\quad\text{ as }m\rightarrow\infty.
\end{align*}

That is, $P(|A_m^\ell|>\delta\mid I_\ell^c)$ is a bounded sequence that converges to zero in probability. Therefore,
\begin{align*}
    E\left[P(|A_m^\ell|>\delta\mid I_\ell^c)\right]=P(|A_m^\ell|>\delta)\stackrel{p.}{\rightarrow}0,\quad\text{ as }m\rightarrow\infty.
\end{align*}

That is, $A_m^\ell\stackrel{p.}{\rightarrow}0$ as $m\rightarrow\infty$.

The conclusion holds for all $\ell\in\{1,\dots,L\}$ and hence, we conclude that $T_1=o_p(1)$.

\medskip
\noindent
\textbf{Analysis of $T_2$}
\begin{align*}
    T_2&=\frac{1}{\sqrt{L}}\sum_{\ell=1}^L\mathbb{G}_m\left[ \phi(O;\pi_1,\pi_2,\mu_1,\mu_2)\right]\\
    &=\frac{\sqrt{m}}{\sqrt{L}}\sum_{\ell=1}^L \left\{E_m[\phi(O;\pi_1,\pi_2,\mu_1,\mu_2)]-E[\phi(O;\pi_1,\pi_2,\mu_1,\mu_2)]\right\}\\
    &=\frac{\sqrt{m}}{\sqrt{L}}\frac{1}{m}\sum_{\ell=1}^L \sum_{i=1}^m\left\{\phi(O_i;\pi_1,\pi_2,\mu_1,\mu_2)-E[\phi(O;\pi_1,\pi_2,\mu_1,\mu_2)]\right\}\\
    &=\frac{1}{\sqrt{n}}\sum_{i=1}^n\left\{\phi(O_i;\pi_1,\pi_2,\mu_1,\mu_2)-E[\phi(O;\pi_1,\pi_2,\mu_1,\mu_2)]\right\}.
\end{align*}

By central limit theorem, the last expression converges to the distribution $\mathcal{N}(0,var(\phi(O;\pi_1,\pi_2,\mu_1,\mu_2))$ if $E[\phi(O;\pi_1,\pi_2,\mu_1,\mu_2)^2]<\infty$. Note that $E[\phi(O;\pi_1,\pi_2,\mu_1,\mu_2)^2]$ can be upper bounded as
\begin{align*}
    &E[\phi(O;\pi_1,\pi_2,\mu_1,\mu_2)^2]\\
    &=E\Big[\Big\{\mu_2(X)+(1-A)\pi_1(X)\{\mu_1(M,X)-\mu_2(X)\}+A\pi_2(M,X)\{Y-\mu_1(M,X)\}\Big\}^2\Big]\\
    &=E\Big[\Big\{\mu_2(X)\left(1-(1-A)\pi_1(X)\right)+\mu_1(M,X)\left((1-A)\pi_1(X)-A\pi_2(M,X)\right)+AY\pi_2(M,X)\Big\}^2\Big]\\
    &=\|\mu_2(X)\left(1-(1-A)\pi_1(X)\right)+\mu_1(M,X)\left((1-A)\pi_1(X)-A\pi_2(M,X)\right)+AY\pi_2(M,X)\|_2^2\\
    &\leq\Big(\|\mu_2(X)\big(1-(1-A)\pi_1(X)\big)\|_2+\|\mu_1(M,X)\big((1-A)\pi_1(X)-A\pi_2(M,X)\big)\|_2\\
    &\quad\quad+\|AY\pi_2(M,X)\|_2\Big)^2\\
    &\leq\Big(\|\mu_2(X)\|_2+\|(1-A)\mu_2(X)\pi_1(X)\|_2+\|(1-A)\mu_1(M,X)\pi_1(X)\|_2+\|A\mu_1(M,X)\pi_2(M,X)\|_2\\
    &\quad\quad+\|AY\pi_2(M,X)\|_2\Big)^2\\
    &\leq\Big(\|\mu_2(X)\|_2+\|\mu_2(X)\pi_1(X)\|_2+\|\mu_1(M,X)\pi_1(X)\|_2+\|\mu_1(M,X)\pi_2(M,X)\|_2+\|Y\pi_2(M,X)\|_2\Big)^2.
\end{align*}

By Assumption \ref{ass:stg_pos} and Jensen's inequality,
\begin{align*}
&\|\mu_2(X)\|_2\leq\|\mu_2^2(X)\|_2^{1/2},\\
    &\|\mu_2(X)\pi_1(X)\|_2\leq\frac{1}{\epsilon}\|\mu_2(X)\|_2\leq\frac{1}{\epsilon}\|\mu_2^2(X)\|_2^{1/2},\\
    &\|\mu_1(M,X)\pi_1(X)\|_2\leq\frac{1}{\epsilon}\|\mu_1(M,X)\|_2\leq\frac{1}{\epsilon}\|\mu_1^2(M,X)\|_2^{1/2},\\
    &\|\mu_1(M,X)\pi_2(M,X)\|_2\leq\frac{1}{\epsilon^2}\|\mu_1(M,X)\|_2\leq\frac{1}{\epsilon^2}\|\mu_1^2(M,X)\|_2^{1/2},\\
    &\|Y\pi_2(M,X)\|_2\leq\frac{1}{\epsilon^2}\|Y\|_2\leq\frac{1}{\epsilon^2}\|Y^2\|_2^{1/2}.
\end{align*}

Therefore, by Assumption \ref{assm:CAN} (c), we have $E[\phi(O;\pi_1,\pi_2,\mu_1,\mu_2)^2]<\infty$.

\medskip
\noindent
\textbf{Analysis of $T_3$}

Using result of Theorem \ref{thm:main} and Law of Total Probability, we have,
\begin{align*}
    T_3&=\frac{1}{\sqrt{L}}\sum_{\ell=1}^L\sqrt{m}\Big\{E\left[ \phi(O;\hat{\pi}_{1,\ell},\hat{\pi}_{2,\ell},\hat{\mu}_{1,\ell},\hat{\mu}_{2,\ell})\right]-\psi_0\Big\}\\
    &=\frac{\sqrt{m}}{\sqrt{L}}\sum_{\ell=1}^L\left\{E\left[\hat\psi^\text{2S}_\ell\right]-\psi_0\right\}\\
    &=\frac{\sqrt{m}}{\sqrt{L}}\sum_{\ell=1}^L\Bigg\{E\left[ \left\{(1-A)\hat\pi_{1,\ell}(X)-1\right\}\{\mu_2(X)-\hat\mu_{2,\ell}(X)\}\right]\\
&~~~~~~+E\left[\left\{A\hat\pi_{2,\ell}(M,X)-(1-A)\hat\pi_{1,\ell}(X)\right\}\{\mu_1(M,X)-\hat\mu_{1,\ell}(M,X)\} \right]\Bigg\}\\
&=\frac{\sqrt{m}}{\sqrt{L}}\sum_{\ell=1}^L\Bigg\{E\left[ \left\{E\left[1-A\mid X\right]\hat\pi_{1,\ell}(X)-1\right\}\{\mu_2(X)-\hat\mu_{2,\ell}(X)\}\right]\\
&~~~~~~+E\left[\left\{E\left[A\mid M,X\right]\hat\pi_{2,\ell}(M,X)-E\left[1-A\mid M,X\right]\hat\pi_{1,\ell}(X)\right\}\{\mu_1(M,X)-\hat\mu_{1,\ell}(M,X)\} \right]\Bigg\}\\
&=\frac{\sqrt{m}}{\sqrt{L}}\sum_{\ell=1}^L\Bigg\{E\left[ \left\{p\left(A=0\mid X\right)\hat\pi_{1,\ell}(X)-1\right\}\{\mu_2(X)-\hat\mu_{2,\ell}(X)\}\right]\\
&~~~~~~~~~~~~~~~+E[\{p\left(A=1\mid M,X\right)\hat\pi_{2,\ell}(M,X)\\
    &~~~~~~~~~~~~~~~-p\left(A=0\mid M,X\right)\hat\pi_{1,\ell}(X)\}\{\mu_1(M,X)-\hat\mu_{1,\ell}(M,X)\} ]\Bigg\}.
\end{align*}
Using Cauchy-Schwarz inequality, we have
\begin{align*}
T_3&\leq\frac{\sqrt{m}}{\sqrt{L}}\sum_{\ell=1}^L\Bigg\{\Big\|p\left(A=0\mid X\right)\hat\pi_{1,\ell}(X)-1\Big\|_2\Big\|\mu_2(X)-\hat\mu_{2,\ell}(X)\Big\|_2\\
&~~~~~~+\Big\|p\left(A=1\mid M,X\right)\hat\pi_{2,\ell}(M,X)-p\left(A=0\mid M,X\right)\hat\pi_{1,\ell}(X)\Big\|_2\Big\|\mu_1(M,X)-\hat\mu_{1,\ell}(M,X)\Big\|_2\Bigg\}.
\end{align*}

Therefore by Assumption \ref{assm:CAN} (d), we have
\begin{align*}
T_3&=\frac{\sqrt{m}}{\sqrt{L}}\sum_{\ell=1}^L\Bigg\{o_p(n^{-1/2})+o_p(n^{-1/2})\Bigg\}\\
&=\frac{\sqrt{m}}{\sqrt{L}}\times L\times o_p(n^{-1/2})\\
&=\sqrt{mL}\times o_p(n^{-1/2})\\
&=\sqrt{n} o_p(n^{-1/2})\\
&=o_p(1).
\end{align*}
\end{proof}

\begin{proof}[Proof of Theorem \ref{thm:main_multi}]
\begin{align*}
    &E[\hat\Delta^{M_j,\text{2S}}]-\Delta^{M_j}\\
    &=E\Big[A\hat{\pi}(X)\Big\{\hat{R}(M_j,M_{-j},X)\big(Y-\hat{\mu}(M_j,M_{-j},X)\big)+\hat{\eta}_2(M_{-j},X)-\hat{\gamma}_1(X)+\hat{\eta}_1(M_j,X)-\hat{\gamma}_{1,1}(X)\Big\}\\
    &~~~~~~~~~+\frac{(1-A)\hat{\pi}(X)}{\hat{\pi}(X)-1}\Big\{\hat{\gamma}_{1,0}(X)-\hat{\eta}_1(M_j,X)\Big\}+\hat{\gamma}_1(X)-\gamma_1(X)\Big]\\
    &=E\Big[\hat{\pi}(X)\hat{R}(M_j,M_{-j},X)\big(AY-A\hat{\mu}(M_j,M_{-j},X)\big)\Big]\\
    &~~~~~~~~~+E\Big[A\hat{\pi}(X)\hat{\eta}_2(M_{-j},X)-A\hat{\pi}(X)\hat{\gamma}_1(X)\Big]+E\Big[A\hat{\pi}(X)
    \hat{\eta}_1(M_j,X)-A\hat{\pi}(X)\hat{\gamma}_{1,1}(X)\Big]\\
    &~~~~~~~~~+E\Big[\frac{(1-A)\hat{\pi}(X)}{\hat{\pi}(X)-1}\hat{\gamma}_{1,0}(X)-\frac{(1-A)\hat{\pi}(X)}{\hat{\pi}(X)-1}\hat{\eta}_1(M_j,X)\Big]+E\Big[\hat{\gamma}_1(X)-\gamma_1(X)\Big]\\
    &=E\Big[E\big[\frac{1}{\hat{\pi}(X)}\hat{R}(M_j,M_{-j},X)\big(AY-A\hat{\mu}(M_j,M_{-j},X)\big)\mid M_j,M_{-j},X\big]\Big]\\
    &~~~~~~~~~+E\Big[E[A\hat{\pi}(X)
    \hat{\eta}_2(M_{-j},X)\mid A,X]-A\hat{\pi}(X)\hat{\gamma}_1(X)\Big]+E\Big[E[A\hat{\pi}(X)
    \hat{\eta}_1(M_j,X)\mid A,X]-A\hat{\pi}(X)\hat{\gamma}_{1,1}(X)\Big]\\
    &~~~~~~~~~+E\Big[\frac{(1-A)\hat{\pi}(X)}{\hat{\pi}(X)-1}\hat{\gamma}_{1,0}(X)-E[\frac{(1-A)\hat{\pi}(X)}{\hat{\pi}(X)-1}\hat{\eta}_1(M_j,X)\mid A,X]\Big]+E\Big[\hat{\gamma}_1(X)-\gamma_1(X)\Big]\\
    &=E\Big[\hat{\pi}(X)\hat{R}(M_j,M_{-j},X)\big(\underbrace{E[AY\mid M_j,M_{-j},X]}_{=E[A\mid M_j,M_{-j},X]\mu(M_j,M_{-j},X)}-E[A\mid M_j,M_{-j},X]\hat{\mu}(M_j,M_{-j},X)\big)\Big]\\
    &~~~~~~~~~+E\Big[A\hat{\pi}(X)
    E[\hat{\eta}_2(M_{-j},X)\mid A=1,X]-A\hat{\pi}(X)\hat{\gamma}_1(X)\Big]\\
    &~~~~~~~~~+\underbrace{E\Big[A\hat{\pi}(X)
    \underbrace{E[\hat{\eta}_1(M_j,X)\mid A=1,X]}_{=\hat{\gamma}_{1,1}(X)}-A\hat{\pi}(X)\hat{\gamma}_{1,1}(X)\Big]}_{=0}\\
    &~~~~~~~~~+E\Big[\frac{(1-A)\hat{\pi}(X)}{\hat{\pi}(X)-1}\hat{\gamma}_{1,0}(X)-\frac{(1-A)\hat{\pi}(X)}{\hat{\pi}(X)-1}E[\hat{\eta}_1(M_j,X)\mid A=0,X]\Big]+E\Big[\hat{\gamma}_1(X)-\gamma_1(X)\Big]\\
    &=E\Big[\hat{\pi}(X)\hat{R}(M_j,M_{-j},X)E[A\mid M_j,M_{-j},X]\big\{\mu(M_j,M_{-j},X)-\hat{\mu}(M_j,M_{-j},X)\big\}\Big]\\
    &~~~~~~~~~+E\Big[A\hat{\pi}(X)
    \Big\{E[\hat{\eta}_2(M_{-j},X)\mid A=1,X]-\hat{\gamma}_1(X)\Big\}\Big]\\
    &~~~~~~~~~+E\Big[\frac{(1-A)\hat{\pi}(X)}{\hat{\pi}(X)-1}\Big\{\hat{\gamma}_{1,0}(X)-E[\hat{\eta}_1(M_j,X)\mid A=0,X]\Big\}\Big]\\
    &~~~~~~~~~+E\Big[\hat{\gamma}_1(X)-\gamma_1(X)\Big]\\
    &=\underbrace{E\Big[A\hat{\pi}(X)\hat{R}(M_j,M_{-j},X)\big\{\mu(M_j,M_{-j},X)-\hat{\mu}(M_j,M_{-j},X)\big\}\Big]}_{=E\Big[E\big[A\hat{\pi}(X)\hat{R}(M_j,M_{-j},X)\big\{\mu(M_j,M_{-j},X)-\hat{\mu}(M_j,M_{-j},X)\big\}\mid A,X\big]\Big]}\\
    &~~~~~~~~~+E\Big[A\hat{\pi}(X)
    \Big\{E[\hat{\eta}_2(M_{-j},X)\mid A=1,X]-\hat{\gamma}_1(X)\Big\}\Big]\\
    &~~~~~~~~~+E\Big[\frac{(1-A)\hat{\pi}(X)}{\hat{\pi}(X)-1}\Big\{\hat{\gamma}_{1,0}(X)-E[\hat{\eta}_1(M_j,X)\mid A=0,X]\Big\}\Big]\\
    &~~~~~~~~~+E\Big[\hat{\gamma}_1(X)-\gamma_1(X)\Big]\\
    &=E\Big[A\hat{\pi}(X)E\big[\hat{R}(M_j,M_{-j},X)\big\{\mu(M_j,M_{-j},X)-\hat{\mu}(M_j,M_{-j},X)\big\}\mid A=1,X\big]\Big]\\
    &~~~~~~~~~+E\Big[A\hat{\pi}(X)
    \Big\{E[\hat{\eta}_2(M_{-j},X)\mid A=1,X]-\hat{\gamma}_1(X)\Big\}\Big]\\
    &~~~~~~~~~+E\Big[\frac{(1-A)\hat{\pi}(X)}{\hat{\pi}(X)-1}\Big\{\hat{\gamma}_{1,0}(X)-E[\hat{\eta}_1(M_j,X)\mid A=0,X]\Big\}\Big]\\
    &~~~~~~~~~+E\Big[\hat{\gamma}_1(X)-\gamma_1(X)\Big]\\
    &=E\Big[\Big\{A\hat{\pi}(X)-1\Big\}E\big[\hat{R}(M_j,M_{-j},X)\big\{\mu(M_j,M_{-j},X)-\hat{\mu}(M_j,M_{-j},X)\big\}\mid A=1,X\big]\Big]\\
    &~~~~~~~~~+E\Big[\Big\{A\hat{\pi}(X)-1\Big\}
    \Big\{E[\hat{\eta}_2(M_{-j},X)\mid A=1,X]-\hat{\gamma}_1(X)\Big\}\Big]\\
    &~~~~~~~~~+E\Big[\Big\{\frac{(1-A)\hat{\pi}(X)}{\hat{\pi}(X)-1}-1\Big\}\Big\{\hat{\gamma}_{1,0}(X)-E[\hat{\eta}_1(M_j,X)\mid A=0,X]\Big\}\Big]\\
    &~~~~~~~~~+E\Big[\hat{\gamma}_1(X)-\gamma_1(X)\Big]+E\Big[E\big[\hat{R}(M_j,M_{-j},X)\big\{\mu(M_j,M_{-j},X)-\hat{\mu}(M_j,M_{-j},X)\big\}\mid A=1,X\big]\Big]\\
    &~~~~~~~~~+E\Big[E[\hat{\eta}_2(M_{-j},X)\mid A=1,X]-\hat{\gamma}_1(X)\Big]+E\Big[\hat{\gamma}_{1,0}(X)-E[\hat{\eta}_1(M_j,X)\mid A=0,X]\Big]\\
    &=E\Big[\Big\{A\hat{\pi}(X)-1\Big\}E\big[\hat{R}(M_j,M_{-j},X)\big\{\mu(M_j,M_{-j},X)-\hat{\mu}(M_j,M_{-j},X)\big\}\mid A=1,X\big]\Big]\\
    &~~~~~~~~~+E\Big[\Big\{A\hat{\pi}(X)-1\Big\}
    \Big\{E[\hat{\eta}_2(M_{-j},X)\mid A=1,X]-\hat{\gamma}_1(X)\Big\}\Big]\\
    &~~~~~~~~~+E\Big[\Big\{\frac{(1-A)\hat{\pi}(X)}{\hat{\pi}(X)-1}-1\Big\}\Big\{\hat{\gamma}_{1,0}(X)-E[\hat{\eta}_1(M_j,X)\mid A=0,X]\Big\}\Big]\\
    &~~~~~~~~~+{E\Big[E\big[\hat{R}(M_j,M_{-j},X)\big\{\mu(M_j,M_{-j},X)-\hat{\mu}(M_j,M_{-j},X)\big\}\mid A=1,X\big]\Big]}\\
    &~~~~~~~~~+{E\Big[E[\hat{\eta}_2(M_{-j},X)\mid A=1,X]\Big]}+{E\Big[\hat{\gamma}_{1,0}(X)-E[\hat{\eta}_1(M_j,X)\mid A=0,X]\Big]}-{E\Big[\gamma_1(X)\Big]}\\
    &=E\Big[\Big\{A\hat{\pi}(X)-1\Big\}E\big[\hat{R}(M_j,M_{-j},X)\big\{\mu(M_j,M_{-j},X)-\hat{\mu}(M_j,M_{-j},X)\big\}\mid A=1,X\big]\Big]\\
    &~~~~~~~~~+E\Big[\Big\{A\hat{\pi}(X)-1\Big\}
    \Big\{E[\hat{\eta}_2(M_{-j},X)\mid A=1,X]-\hat{\gamma}_1(X)\Big\}\Big]\\
    &~~~~~~~~~+E\Big[\Big\{\frac{(1-A)\hat{\pi}(X)}{\hat{\pi}(X)-1}-1\Big\}\Big\{\hat{\gamma}_{1,0}(X)-E[\hat{\eta}_1(M_j,X)\mid A=0,X]\Big\}\Big]\\
    &~~~~~~~~~+{E\Big[E\big[\{\hat{R}(M_j,M_{-j},X)-R(M_j,M_{-j},X)\}\big\{\mu(M_j,M_{-j},X)-\hat{\mu}(M_j,M_{-j},X)\big\}\mid A=1,X\big]\Big]}\\
    &~~~~~~~~~+\underbrace{E\Big[E\big[R(M_j,M_{-j},X)\big\{\mu(M_j,M_{-j},X)-\hat{\mu}(M_j,M_{-j},X)\big\}\mid A=1,X\big]\Big]}_{=T_1}\\
    &~~~~~~~~~+\underbrace{E\Big[E[\hat{\eta}_2(M_{-j},X)\mid A=1,X]\Big]}_{=T_2}+\underbrace{E\Big[\hat{\gamma}_{1,0}(X)-E[\hat{\eta}_1(M_j,X)\mid A=0,X]\Big]}_{=T_3}-\underbrace{E\Big[\gamma_1(X)\Big]}_{=T_4}.
\end{align*}

\begin{align*}
    &T_1=E\Big[E\big[R(M_j,M_{-j},X)\big\{\mu(M_j,M_{-j},X)-\hat{\mu}(M_j,M_{-j},X)\big\}\mid A=1,X\big]\Big]\\
    &=E\Big[E\big[\big(1-\rho(M_j,X)\big){\omega}(M_j,M_{-j},X)\big\{\mu(M_j,M_{-j},X)-\hat{\mu}(M_j,M_{-j},X)\big\}\mid A=1,X\big]\Big]\\
    &=E\Big[\int\int\big(1-\rho(m_j,X)\big){\omega}(m_j,m_{-j},X)\big\{\mu(m_j,m_{-j},X)-\hat{\mu}(m_j,m_{-j},X)\big\}f(m_j,m_{-j}\mid A=1,X)dm_jdm_{-j}\Big]\\
    &=E\Big[\int\big(1-\rho(m_j,X)\big)\Big\{\int\big\{\mu(m_j,m_{-j},X)-\hat{\mu}(m_j,m_{-j},X)\big\}f(m_{-j}\mid A=1,X)dm_{-j}\Big\}f(m_j\mid A=1,X)dm_j\Big]\\
    &=E\Big[\int\big(1-\rho(m_j,X)\big)\Big\{\eta_1(m_j,X)-\eta_1^\ast(m_j,X)\Big\}f(m_j\mid A=1,X)dm_j\Big]\\
    &=E\Big[E\big[\big(1-\rho(M_j,X)\big)\Big\{\eta_1(M_j,X)-\eta_1^\ast(M_j,X)\Big\}\mid A=1,X\big]\Big],
\end{align*}
where $\eta_1^\ast(M_j,X)=\int \hat{\mu}(m_j,m_{-j},X)f(m_{-j}\mid A=1,X)dm_{-j}$.

\begin{align*}
    &T_2=E\Big[E[\hat{\eta}_2(M_{-j},X)\mid A=1,X]\Big]\\
    &=E\Big[E_{M_{-j}\mid A=1,X}[E_{M_j\mid A=1,X}[\hat{\mu}(M_j,M_{-j},X)\big(1-\hat{\rho}(M_j,X)\big)\mid A=1,X]\mid A=1,X]\Big]\\
    &=E\Big[E[\int\hat{\mu}(m_j,m_{-j},X)\big(1-\hat{\rho}(m_j,X)\big)f(m_j\mid A=1,X)dm_j\mid A=1,X]\Big]\\
    &=E\Big[\int\int\hat{\mu}(m_j,m_{-j},X)\big(1-\hat{\rho}(m_j,X)\big)f(m_j\mid A=1,X)f(m_{-j}\mid A=1,X)dm_jdm_{-j}\Big]\\
    &=E\Big[\int\Big(\int\hat{\mu}(m_j,m_{-j},X)f(m_{-j}\mid A=1,X)dm_{-j}\Big)\big(1-\hat{\rho}(m_j,X)\big)f(m_j\mid A=1,X)dm_j\Big]\\
    &=E\Big[\int\eta_1^\ast(m_j,X)\big(1-\hat{\rho}(m_j,X)\big)f(m_j\mid A=1,X)dm_j\Big]\\
    &=E\Big[E\big[\eta_1^\ast(M_j,X)\big(1-\hat{\rho}(M_j,X)\big)\mid A=1,X\big]\Big].
\end{align*}

\begin{align*}
    &T_3=E\Big[\hat{\gamma}_{1,0}(X)-E[\hat{\eta}_1(M_j,X)\mid A=0,X]\Big]\\
    &=E\Big[E[\hat{\eta}_1(M_j,X)\hat{\rho}(M_j,X)\mid A=1,X]-E[\hat{\eta}_1(M_j,X)\rho(M_j,X)\mid A=1,X]\Big]\\
    &=E\Big[E[\hat{\eta}_1(M_j,X)\{\hat{\rho}(M_j,X)-\rho(M_j,X)\}\mid A=1,X]\Big].
\end{align*}

\begin{align*}
    &T_4=E\Big[\gamma_1(X)\Big]\\
    &=E\Big[\gamma_{1,1}(X)-\gamma_{1,0}(X)\Big]\\
    &=E\Big[E[\eta_1(M_j,X)\mid A=1,X]-E[\eta_1(M_j,X)\rho(M_j,X)\mid A=1,X]\Big]\\
    &=E\Big[E[\eta_1(M_j,X)\big(1-\rho(M_j,X)\big)\mid A=1,X]\Big].
\end{align*}

Combining, we have,
\begin{align*}
    &E[\hat\Delta^{M_j,\text{2S}}]-\Delta^{M_j}\\
    &=E\Big[\Big\{A\hat{\pi}(X)-1\Big\}E\big[\hat{R}(M_j,M_{-j},X)\big\{\mu(M_j,M_{-j},X)-\hat{\mu}(M_j,M_{-j},X)\big\}\mid A=1,X\big]\Big]\\
    &~~~~~~~~~+E\Big[\Big\{A\hat{\pi}(X)-1\Big\}
    \Big\{E[\hat{\eta}_2(M_{-j},X)\mid A=1,X]-\hat{\gamma}_1(X)\Big\}\Big]\\
    &~~~~~~~~~+E\Big[\Big\{\frac{(1-A)\hat{\pi}(X)}{\hat{\pi}(X)-1}-1\Big\}\Big\{\hat{\gamma}_{1,0}(X)-E[\hat{\eta}_1(M_j,X)\mid A=0,X]\Big\}\Big]\\
    &~~~~~~~~~+{E\Big[E\big[\{\hat{R}(M_j,M_{-j},X)-R(M_j,M_{-j},X)\}\big\{\mu(M_j,M_{-j},X)-\hat{\mu}(M_j,M_{-j},X)\big\}\mid A=1,X\big]\Big]}\\
    &~~~~~~~~~+E\Big[E\big[\big(1-\rho(M_j,X)\big)\Big\{\eta_1(M_j,X)-\eta_1^\ast(M_j,X)\Big\}\mid A=1,X\big]\Big]\\
    &~~~~~~~~~+E\Big[E\big[\eta_1^\ast(M_j,X)\big(1-\hat{\rho}(M_j,X)\big)\mid A=1,X\big]\Big]\\
    &~~~~~~~~~+E\Big[E[\hat{\eta}_1(M_j,X)\{\hat{\rho}(M_j,X)-\rho(M_j,X)\}\mid A=1,X]\Big]\\
    &~~~~~~~~~-E\Big[E[\eta_1(M_j,X)\big(1-\rho(M_j,X)\big)\mid A=1,X]\Big]\\
    &=E\Big[\Big\{A\hat{\pi}(X)-1\Big\}E\big[\hat{R}(M_j,M_{-j},X)\big\{\mu(M_j,M_{-j},X)-\hat{\mu}(M_j,M_{-j},X)\big\}\mid A=1,X\big]\Big]\\
    &~~~~~~~~~+E\Big[\Big\{A\hat{\pi}(X)-1\Big\}
    \Big\{E[\hat{\eta}_2(M_{-j},X)\mid A=1,X]-\hat{\gamma}_1(X)\Big\}\Big]\\
    &~~~~~~~~~+E\Big[\Big\{\frac{(1-A)\hat{\pi}(X)}{\hat{\pi}(X)-1}-1\Big\}\Big\{\hat{\gamma}_{1,0}(X)-E[\hat{\eta}_1(M_j,X)\mid A=0,X]\Big\}\Big]\\
    &~~~~~~~~~+{E\Big[E\big[\{\hat{R}(M_j,M_{-j},X)-R(M_j,M_{-j},X)\}\big\{\mu(M_j,M_{-j},X)-\hat{\mu}(M_j,M_{-j},X)\big\}\mid A=1,X\big]\Big]}\\
&~~~~~~~~~+E\Big[E\big[\{\eta_1^\ast(M_j,X)-\hat{\eta}_1(M_j,X)\}\big\{\rho(M_j,X)-\hat{\rho}(M_j,X)\big\}\mid A=1,X\big]\Big].
\end{align*}
\end{proof}

\end{document}